\documentclass[review]{elsarticle}
\usepackage[latin1]{inputenc} 
\usepackage{listings, tcolorbox}

\usepackage{bbm}
\usepackage{}

\usepackage{lineno,hyperref}
\usepackage{amssymb}
\usepackage{color}
\usepackage{amsmath}
\usepackage{tikz}
\usepackage{amsfonts}
\usepackage{mathrsfs}
\usepackage{amsthm}

\usepackage{bbding}
\usepackage{mathrsfs}
\usepackage{amsmath, amsfonts, amssymb, mathrsfs, txfonts}
\usepackage{graphicx, subfigure}
\usepackage{wasysym}
\usepackage{color}
\usepackage{bm}
\usepackage{enumerate}

\usepackage{pifont}
\usepackage{txfonts}
\usepackage{amsmath}
\usepackage{graphicx}
\usepackage{amsfonts}
\usepackage{amssymb}
\usepackage{mathrsfs,psfrag,eepic,epsfig}
\usepackage{epstopdf}

\newtheorem{thm}{Theorem}[section]
 
 \newtheorem{lem}[thm]{Lemma}
 \newtheorem{prop}[thm]{Proposition}
 \theoremstyle{definition}
 
 \theoremstyle{remark}
 \newtheorem{rem}[thm]{Remark}

\biboptions{numbers,sort&compress}

\modulolinenumbers[5]

\journal{Journal of \LaTeX\ Templates}

\makeatletter \@addtoreset{equation}{section}
\renewcommand{\theequation}{\arabic{section}.\arabic{equation}}

\begin{document}

\begin{frontmatter}
\title{Inverse scattering transform and soliton solutions for the
modified matrix Korteweg-de Vries equation with nonzero boundary conditions}
\tnotetext[mytitlenote]{
Corresponding author.\\
\hspace*{3ex}\emph{E-mail addresses}: jinjieyang@cumt.edu.cn (J.J. Yang), sftian@cumt.edu.cn,
shoufu2006@126.com (S. F. Tian) and
zqli@cumt.edu.cn (Z. Q. Li)}

\author{Jin-Jie Yang, Shou-Fu Tian$^{*}$ and Zhi-Qiang Li}
\address{
School of Mathematics and Institute of Mathematical Physics, China University of Mining and Technology,\\ Xuzhou 221116, People's Republic of China\\
}

\begin{abstract}
The theory of inverse scattering is developed to study the initial-value problem for the modified matrix Korteweg-de Vries (mmKdV) equation with the $2m\times2m$ $(m\geq 1)$ Lax pairs under the nonzero boundary conditions at infinity. In the direct problem, by introducing a suitable uniform transformation we establish the proper complex $z$-plane in order to discuss the Jost eigenfunctions, scattering matrix and their analyticity and symmetry of the equation. Moreover the asymptotic behavior of the Jost functions and scattering matrix needed in the inverse problem are analyzed via  Wentzel-Kramers-Brillouin
expansion. In the inverse problem, the generalized Riemann-Hilbert problem of the mmKdV equation is first established by using the analyticity of the modified eigenfunctions and scattering coefficients. The reconstruction formula of potential function with reflection-less case is derived by solving this Riemann-Hilbert problem and using the scattering data. In addition the dynamic behavior of the solutions for the focusing mmKdV equation including one- and two- soliton solutions are presented in detail under the the condition that the potential is scalar and the $2\times2$ symmetric matrix. Finally, we provide some detailed proofs and weak version of trace formulas to show that the asymptotic phase of the potential and the scattering data.

\end{abstract}

\begin{keyword}
The modified matrix Korteweg-de Vries equation \sep The generalized Riemann-Hilbert problem \sep  Nonzero boundary conditions  \sep Soliton solutions \sep  Breather wave solutions.
\end{keyword}

\end{frontmatter}


\tableofcontents

\section{Introduction}
The theory of nonlinear dynamics has aroused considerable interest and has established a connection with some directions in the field of soliton theory. It is well known that the Korteweg-de Vries (KdV) equation, the Sasa-Satsuma equation, the
nonlinear Schr\"{o}dinger (NLS) equation are the important typical and fully studied nonlinear integrable equation, which can describe a series of nonlinear wave phenomena in dispersive physical structures. One of the examples is the modified KdV equation
\begin{align}\label{Q1}
q_{t}+q_{xxx}-6\epsilon q^{2}q_{x}=0,
\end{align}
where the subscripts denote the corresponding partial derivatives, $q$ is the real scalar function, $(x,t)\in R^{2}$, and $\epsilon=-1,1$ denote the focusing and defocusing modified KdV equation, respectively.

The Eq.\eqref{Q1} can be applied to many fields, including Alfv\'{e}n waves in collision-less plasmas \cite{Khater-1998}, hyperbolic surfaces \cite{Schief-1995}, and thin elastic rods \cite{Matsutani-1991}  etc.  There are also lots of results about the focusing or defocusing modified KdV Eq.\eqref{Q1} \cite{Wadati-1973,Wadati-1982, Yeung-1988,Yeung-1984, He-2005, Baldwin-2013,Miura-1968,Hirota-1972,Wadati-1998} due to its simple expression and rich physical application. In addition, the more extensive form of KdV equation has also been studied in detail by some authors, such as coupled modified KdV \cite{Tian-JPA}, multi-component form \cite{Yajima-1975,Sasa-1991} and matrix form \cite{Athorne-1987}.

In this work, we consider the modified matrix KdV (mmKdV) equation read as \cite{Wadati-1998}
\begin{align}\label{T1}
Q_{t}+Q_{xxx}-3\epsilon\left(Q_{x}Q^{\dagger}Q+QQ^{\dagger}Q_{x}\right)=0,
\end{align}
where the potential function $Q(x,t)$ is a $p\times q$ matrix function, the superscript $\dagger$ represents the Hermitian conjugate, and the symbol  $\epsilon=-1,1$  denote  the focusing and defocusing mmKdV, respectively.
The multi-soliton solutions of the Eq.\eqref{T1} have been derived in \cite{Wadati-1998} via the inverse scattering method (ISM) proposed by Gardner, Greene, Kruskal and Miura \cite{GGKM-1967}. In Ref.\cite{Ablowitz-2004}, they have studied in detail the Eq.\eqref{T1} with  sufficient smooth potential by the ISM and Riemann-Hilbert (RH) problem. Since the ISM was proposed, it has become a powerful tool for the analysis of nonlinear partial differential equations (PDEs). Many valuable results  have been obtained by using the ISM, such as the general coupled NLS equation \cite{Tian-jde}, coupled mKdV system \cite{Ma-jgp}, the mixed coupled NLS equation \cite{Tian-pra}, the Fokas-Lenells equation \cite{Fan-jde,YFan-2019}, the Sasa-Satsuma equation \cite{Lfan-2019,Geng-jde}, Gerdjikov-Ivanov type of derivative NLS equation \cite{Tian-pamc}, the quartic NLS equation \cite{Liu-ND}, the Hirota equation \cite{YanZhang-2020}, three-component coupled NLS equation \cite{PTian-2019}, the Kundu-Eckhaus equation \cite{Wang-jde,Guo-jmp,FanL-2019,jianxu,YTian-2019} etc.

When the  matrix potential function $Q(x,t)$ is a $2\times2$ and symmetric matrix:
\begin{align}\label{Q}
 Q(x,t)=\left(\begin{array}{cc}
 q_{1}(x,t) & q_{0}(x,t)\\
 q_{0}(x,t) & q_{2}(x,t)\\
 \end{array}\right),
 \end{align}
then the Eq.\eqref{T1} can be written in the form of the following components
\begin{align*}
q_{1,t}+q_{1,xxx}-6\epsilon\left[q_{1,x}\left(|q_{1}|^{2}+|q_{0}|^{2}\right)+
q_{0,x}\left(q_{1}q_{0}^{*}+q_{0}q_{2}^{*}\right)\right]=0,\\
q_{0,t}+q_{0,xxx}-3\epsilon\left[q_{0,x}\left(|q_{1}|^{2}+2|q_{0}|^{2}+|q_{2}|^{2}
\right)+q_{1,x}\left(q_{0}q_{1}^{*}+q_{2}q_{0}^{*}\right)+q_{2,x}
\left(q_{1}q_{0}^{*}+q_{0}q_{2}^{*}\right)\right]=0,\\
q_{2,t}+q_{2,xxx}-6\epsilon\left[q_{2,x}\left(|q_{2}|^{2}+|q_{0}|^{2}\right)+
q_{0,x}\left(q_{0}q_{1}^{*}+q_{2}q_{0}^{*}\right)\right]=0,
\end{align*}
where the asterisk `*' represents complex conjugation.  The coupled modified KdV has been studied by Geng \cite{Geng-2017} with a rapidly decaying potential $Q(x,t)$. The non-zero boundary conditions (NZBCs) of the modified KdV equation has been given by Yan \cite{Yan-2018}. However, there are no work to study the NZBCs in the multi-component case of the modified KdV equation. For multi-component nonlinear equations, Biondini, Kraus, Ieda et al. have established the frame of NZBCs for the Schr\"{o}dinger equation \cite{Ortiz-2019, Ieda-2007, Demontis-2019, Prinari-2006, Kraus-2015, Biondini-2016}. Inspired by this, the purpose of this work is to establish a fundamental frame of NZBCs for the mmKdV equation with a general case $Q(x,t)$ is an $m\times m$  symmetric matrix $(m\geq 1)$. It's noted that we consider the system Eq.\eqref{T1} under the NZBCs as $x\rightarrow \pm\infty$
\begin{align}\label{T2}
Q(x,t)\rightarrow Q_{\pm},
\end{align}
and assume  the constraints
\begin{align}\label{T3}
Q_{\pm}Q_{\pm}^{\dagger}=Q_{\pm}^{\dagger}Q_{\pm}=k_{0}^{2}I_{m},
\end{align}
where $k_{0}$ is a real, positive constant, and the $I_{m}$ denotes the unit matrix of order $m$. For the special case Eq.\eqref{Q}, the Eq.\eqref{T3} implies
\begin{align*}
|q_{0,\pm}|^2+|q_{1,\pm}|^2=|q_{0,\pm}|^2+|q_{2,\pm}|^2=k_{0}^{2},\quad
q_{1,\pm}^{*}q_{0,\pm}+q_{0,\pm}^{*}q_{2,\pm}=0,\quad
|q_{1,\pm}|^2=|q_{2,\pm}|^2.
\end{align*}

The outline of the work is arranged as: In section 2, we consider the direct scattering problem of the spectrum problem of   Eq.\eqref{T1}, including the analytical, asymptotic and symmetric properties of the Jost eigenfunctions and scattering matrix, and analyze the discrete spectrum and residue conditions. In section 3, a generalized RH problem is established based on the modified eigenfunctions, from which the potential can be reconstructed. In section 4, we discuss two special cases in combination with the focusing mmKdV equation, one is that the potential function is a scalar, and the other is that the potential function is a $2\times2$ symmetric matrix, and the propagation behavior corresponding to the solution of one- and two-soliton solutions are given by the appropriate parameters. In section 5, some detailed proofs are presented in the appendix, and the relationship between potential function and scattering data is analyzed. Finally some conclusions and discussions are presented in the last section.

\section{Direct scattering problem with NZBCs}
\subsection{Lax pairs}
In this section, we give the Lax pairs of the system Eq.\eqref{T1} and analyze what form the Lax pairs will become under the condition of NZBCs Eq.\eqref{T2}, which provides convenience for constructing the Jost functions later.

From the Ref.\cite{Wadati-1998}, the Lax pairs of Eq.\eqref{T1} read
\begin{align}\label{Lax}
\left\{\begin{aligned}
\varphi_{x}=M\varphi,\\
\varphi_{x}=N\varphi,
\end{aligned}\right.
\end{align}
where
\begin{align}\label{T4}
&M=-ik\underline{\sigma}_{3}+\underline{Q},\quad\underline{\sigma}_{3}=\left(\begin{array}{cc}
I_{m} & 0_{m}\\
0_{m} & I_{m}
\end{array}\right),\quad
\underline{Q}=\left(\begin{array}{cc}
0_{m} & Q\\
\epsilon Q^{\dagger} & 0_{m}
\end{array}\right),\\
&N=-4ik^{3}\underline{\sigma}_{3}+4k^{2}\underline{Q}-2ik(\underline{Q}^{2}
+\underline{Q}_{x})\underline{\sigma}_{3}-\underline{Q}_{xx}+2\underline{Q}^{3}
+\underline{Q}_{x}\underline{Q}-\underline{Q}\underline{Q}_{x},
\end{align}
and $0_{m}$ represents the $m\times m$ zero matrix. The Eq.\eqref{Lax}  with $x$ derivative is called scattering problem, and the Eq.\eqref{Lax} with $t$ derivative is called time-dependent problem. It is easy to verify that Lax pairs satisfy the compatibility condition $M_{t}-N_{x}+[M,N]=0$, here $[A,B]=AB-BA$.

Under nonzero boundary conditions Eq.\eqref{T2}, the Eq.\eqref{Lax} is transformed into
\begin{align}\label{T5}
\left\{\begin{aligned}
\varphi_{x}&=M_{\pm}\varphi=(-ik\underline{\sigma}_{3}+\underline{Q}_{\pm})\varphi,\\
\varphi_{t}&=N_{\pm}\varphi=(4k^{2}+2k_{0}^{2})M_{\pm}\varphi,\end{aligned}\right.
\end{align}
which means that there is a reversible matrix to diagonalize $M_{\pm}$ and $N_{\pm}$. It's noted that
\begin{align}\label{T6}
Q_{\pm}Q_{\pm}^{\dagger}=Q_{\pm}^{\dagger}Q_{\pm}=k_{0}^{2}I_{m}\Leftrightarrow
\underline{Q}_{\pm}\underline{Q}_{\pm}^{\dagger}=
\underline{Q}_{\pm}^{\dagger}\underline{Q}_{\pm}=k_{0}^{2}I_{m},
\end{align}
with $\underline{Q}_{\pm}=\mathop{\lim}\limits_{x\rightarrow\pm\infty}\underline{Q}
=\left(\begin{array}{cc}
0_{m} & Q_{\pm}\\
\epsilon Q^{\dagger}_{\pm} & 0_{m}
\end{array}\right)$.

\subsection{Riemann surface and uniformization coordinate}
In order to discuss the analytic region of Jost functions, we need to discuss it on the proper complex plane. Note that the eigenvalues of scattering problem $M_{\pm}$ are $\pm i\sqrt{k^{2}-\epsilon k_{0}^{2}}$, each with multiplicity $2m$. A two-sheeted Riemann surface is introduced to handle the branching of the eigenvalues, namely
\begin{align}\label{T7}
\lambda^{2}=k^{2}-\epsilon k_{0}^{2},
\end{align}
the branch points can be easily derived by $k^2-k_{0}^{2}=0$, i.e., $k=\pm\sqrt{\epsilon}k_{0}$. Note that $\epsilon=-1,1$ present focusing and defocusing  case, respectively.

For the focusing case ($\epsilon=-1$), the Riemann surface determined by the equation $\lambda^{2}=k^{2}+k_{0}^{2}$ is composed of two complex $k$-planes $S_{1}$ and $S_{2}$ cut along the branch points $\pm ik_{0}$. At the same time, on the Riemann surface, the function $\lambda$ is a single-valued function of $k$, which is composed of two single-valued analytic branches. The value of the function differs by one symbol, so local polar coordinates are introduced in the sheet $S_1$.  More precisely, letting $k+ik_{0}=r_{1}e^{i\theta_{1}}$, $k-ik_{0}=r_{2}e^{i\theta_{2}}$ for $-\frac{\pi}{2}<\theta_{1},\theta_{2}<\frac{3\pi}{2}$, we can derive that \begin{align}\label{T8}
\lambda(k)=&\left\{\begin{aligned}
&(r_{1}r_{2})^\frac{1}{2}e^\frac{{\theta_{1}+\theta_{2}}}{2}, \quad &on\quad S_{1},\\
-&(r_{1}r_{2})^\frac{1}{2}e^\frac{{\theta_{1}+\theta_{2}}}{2}, \quad &on\quad S_{2}.
\end{aligned} \right.
\end{align}
Similarly, for the defocusing case ($\epsilon=1$), the Riemann surface $S$ determined by the equation $\lambda^{2}=k^{2}-k_{0}^{2}$ with the branch points $\pm k_{0}$, where the complex plane $S_{1}$ and $S_{2}$ are glued together along the cut $(-\infty,-k_{0})\cup (k_{0},+\infty)$. We also introduce polar coordinates in the complex plane, i.e.,
$k-k_{0}=r_{1}e^{i\theta_{1}}$, $k+k_{0}=r_{2}e^{i\theta_{2}}$ for the angles $0\leq\theta_{1}<2\pi$ and $-\pi\leq\theta_{2}<\pi$. Then the single-valued functions Eq.\eqref{T8} can be obtained similarly.

Resorting to \cite{Faddeev-1987}-\cite{Biondini-2014}, we define a uniformization variable $z$
\begin{align}\label{T9}
z=\lambda+k,
\end{align}
from Eq.\eqref{T7}, one has the following inverse transformation
\begin{align}\label{T10}
k(z)=\frac{1}{2}\left(z+\epsilon\frac{k_{0}^{2}}{z}\right),\quad \lambda(z)=\frac{1}{2}\left(z-\epsilon\frac{k_{0}^{2}}{z}\right).
\end{align}
In combination with these definitions, we will use $z$-plane instead of $k$-plane. The advantage of this is that it solves the multi-valued problem. For the focusing case, from the mapping Eq.\eqref{T8}, the Riemann surface $Imk>0$ of  $S_{1}$ and $Imk<0$ of $S_2$ are mapped to $Im\lambda>0$ of the $\lambda$-plane, the Riemann surface $Imk<0$ of  $S_{1}$ and $Imk>0$ of $S_2$ are mapped to $Im\lambda<0$ of the $\lambda$-plane. In additional, based on the Joukowsky transformation, one has
\begin{align*}
Im\lambda=\frac{1}{2|z|^{2}}\left(|z|^2-k_{0}^{2}\right)Imz.
\end{align*}
As a consequence, the region $Im\lambda>0$ and $Im\lambda<0$ are mapped to
\addtocounter{equation}{1}
\begin{align}
D^{+}=\left\{z\in \mathbb{\textbf{C}}:\left(|z|^{2}-k_{0}^{2}\right)Im z>0\right\},\tag{\theequation a}\\
D^{-}=\left\{z\in \mathbb{\textbf{C}}:\left(|z|^{2}-k_{0}^{2}\right)Im z<0\right\},\tag{\theequation b}
\end{align}
where $\mathbb{\textbf{C}}$ denotes the complex plane, then we take the focusing equation as an example to show the transformation between different complex planes in Fig. 1.

\centerline{\begin{tikzpicture}[scale=0.5]
\filldraw (-9,1) -- (-9,9) to (-1,9) -- (-1,1);
\filldraw (9,1) -- (9,9) to (1,9) -- (1,1);
\filldraw (9,-1) -- (9,-9) to (1,-9) -- (1,-1);
\filldraw (-9,-1) -- (-9,-9) to (-1,-9) -- (-1,-1);
\path [fill=gray] (1,-5) -- (9,-5) to (9,-9) -- (1,-9);
\filldraw[white, line width=0.5](-1,-5)--(3,-5) arc (-180:0:2);
\path [fill=white] (1,-1) -- (9,-1) to (9,-5) -- (1,-5);
\filldraw[gray, line width=0.5](3,-5)--(7,-5) arc (0:180:2);
\path [fill=white] (-9,5)--(-9,9) to (-1,9) -- (-1,5);
\path [fill=gray] (-9,1)--(-9,5) to (-1,5) -- (-1,1);
\path [fill=white] (1,1)--(1,5) to (9,5) -- (9,1);
\path [fill=gray] (1,5)--(1,9) to (9,9) -- (9,5);
\path [fill=white] (-1,-1)--(-1,-5) to (-9,-5) -- (-9,-1);
\path [fill=gray] (-1,-5)--(-1,-9) to (-9,-9) -- (-9,-5);
\filldraw[red, line width=0.5] (2,2) to (-2,-2)[->];
\filldraw[red, line width=0.5] (-2,-8) to (2,-8)[->];
\draw[fill] (-5,5)node[below]{} circle [radius=0.035];
\draw[fill] (5,5)node[below]{} circle [radius=0.035];
\draw[fill] (-5,-5)node[below]{} circle [radius=0.035];
\draw[fill] (5,-5)node[below]{} circle [radius=0.035];
\draw[-][thick](-9,5)--(-8,5);
\draw[-][thick](-8,5)--(-7,5);
\draw[-][thick](-7,5)--(-6,5);
\draw[-][thick](-6,5)--(-5,5);
\draw[-][thick](-5,5)--(-4,5);
\draw[-][thick](-4,5)--(-3,5);
\draw[-][thick](-3,5)--(-2,5);
\draw[-][thick](-2,5)--(-1,5)[->][thick]node[above]{$Rek$};;
\draw[-][thick](-5,1)--(-5,2);
\draw[-][thick](-5,2)--(-5,3);
\draw[-][thick](-5,3)--(-5,4);
\draw[-][thick](-5,4)--(-5,5);
\draw[-][thick](-5,5)--(-5,6);
\draw[-][thick](-5,6)--(-5,7);
\draw[-][thick](-5,7)--(-5,8);
\draw[-][thick](-5,8)--(-5,9)[->] [thick]node[above]{$Imk$};
\draw[-][thick](1,5)--(2,5);
\draw[-][thick](2,5)--(3,5);
\draw[-][thick](3,5)--(4,5);
\draw[-][thick](4,5)--(5,5);
\draw[-][thick](5,5)--(6,5);
\draw[-][thick](6,5)--(7,5);
\draw[-][thick](7,5)--(8,5);
\draw[-][thick](8,5)--(9,5)[->][thick]node[above]{$Rek$};
\draw[-][thick](5,1)--(5,2);
\draw[-][thick](5,2)--(5,3);
\draw[-][thick](5,3)--(5,4);
\draw[-][thick](5,4)--(5,5);
\draw[-][thick](5,5)--(5,6);
\draw[-][thick](5,6)--(5,7);
\draw[-][thick](5,7)--(5,8);
\draw[-][thick](5,8)--(5,9);
\draw[-][thick](-9,-5)--(-8,-5);
\draw[-][thick](-8,-5)--(-7,-5);
\draw[-][thick](-7,-5)--(-6,-5);
\draw[-][thick](-6,-5)--(-5,-5);
\draw[-][thick](-5,-5)--(-4,-5);
\draw[-][thick](-4,-5)--(-3,-5);
\draw[-][thick](-3,-5)--(-2,-5);
\draw[-][thick](-2,-5)--(-1,-5)[->][thick]node[above]{$Re\lambda$};
\draw[-][thick](-5,-1)--(-5,-2);
\draw[-][thick](-5,-2)--(-5,-3);
\draw[-][thick](-5,-3)--(-5,-4);
\draw[-][thick](-5,-4)--(-5,-5);
\draw[-][thick](-5,-5)--(-5,-6);
\draw[-][thick](-5,-6)--(-5,-7);
\draw[-][thick](-5,-7)--(-5,-8);
\draw[-][thick](-5,-8)--(-5,-9);
\draw[-][thick](1,-5)--(2,-5);
\draw[-][thick](2,-5)--(3,-5);
\draw[-][thick](3,-5)--(4,-5);
\draw[-][thick](4,-5)--(5,-5);
\draw[-][thick](5,-5)--(6,-5);
\draw[-][thick](6,-5)--(7,-5);
\draw[-][thick](7,-5)--(8,-5);
\draw[-][thick](8,-5)--(9,-5)[->][thick]node[above]{$Rez$};
\draw[-][thick](5,-1)--(5,-2);
\draw[-][thick](5,-2)--(5,-3);
\draw[-][thick](5,-3)--(5,-4);
\draw[-][thick](5,-4)--(5,-5);
\draw[-][thick](5,-5)--(5,-6);
\draw[-][thick](5,-6)--(5,-7);
\draw[-][thick](5,-7)--(5,-8);
\draw[-][thick](5,-8)--(5,-9);
\draw[->](5,9)[thick]node[above]{$Imk$};
\draw[->](-5,-1)[thick]node[above]{$Im\lambda$};
\draw[->](5,-1)[thick]node[above]{$Imz$};
\draw[fill] (-5,7) circle [radius=0.055]node[left]{\footnotesize$iq_{0}$};
\draw[fill] (-5,3) circle [radius=0.055]node[left]{\footnotesize$-iq_{0}$};
\draw[fill] (5,7) circle [radius=0.055]node[left]{\footnotesize$iq_{0}$};
\draw[fill] (5,3) circle [radius=0.055]node[left]{\footnotesize$-iq_{0}$};
\draw[fill] (-7,-5) circle [radius=0.055]node[below]{\footnotesize$q_{0}$};
\draw[fill] (-3,-5) circle [radius=0.055]node[below]{\footnotesize$-q_{0}$};
\draw(5,-5) [red, line width=1] circle(2);
\filldraw[red, line width=1.5] (-5,7) to (-5,3);
\filldraw[red, line width=1.5] (5,7) to (5,3);
\filldraw[red, line width=1.5] (-7,-5) to (-3,-5);
\draw[fill][black] (-8,7) [thick]node[right]{\footnotesize$S_{1}$};
\draw[fill][black] (2,7) [thick]node[right]{\footnotesize$S_{2}$};
\draw[fill][black] (-4,7) [thick]node[right]{\footnotesize$Imk>0$};
\draw[fill][black] (-4,3) [thick]node[right]{\footnotesize$Imk<0$};
\draw[fill][black] (6,7) [thick]node[right]{\footnotesize$Imk>0$};
\draw[fill][black] (6,3) [thick]node[right]{\footnotesize$Imk<0$};
\draw[fill][black] (-4,-7) [thick]node[right]{\footnotesize$Im\lambda<0$};
\draw[fill][black] (-4,-3) [thick]node[right]{\footnotesize$Im\lambda>0$};
\draw[fill][black] (7,-7) [thick]node[right]{\footnotesize$D_{-}$};
\draw[fill][black] (7,-3) [thick]node[right]{\footnotesize$D_{+}$};
\draw[fill][red] (0,5) node[]{\footnotesize$+$};
\draw[fill][black] (-2,0) [thick]node[right]
{\footnotesize$\lambda=\sqrt{k^{2}+q_{0}^{2}}$};
\draw[fill][black] (0,-8) [thick]node[below]
{\footnotesize$\lambda=(z+q_{0}^{2}/z)/2$};
\end{tikzpicture}}
\noindent {\small \textbf{Figure 1.} Transformation relation from $k$ two-sheeted Riemann surface, $\lambda$-plane and $z$-plane.}

\subsection{Jost functions and its analyticity}
As the simultaneous solution of two parts of Lax pairs Eq.\eqref{Lax}, the Jost eigenfunctions are usually defined according to the eigenvectors of the asymptotic scattering problem Eq.\eqref{T5} $M_{\pm}$. By calculation, the eigenvectors of $M_{\pm}$ can be written as
\begin{align}\label{T11}
E_{\pm}(k)=\left(\begin{array}{cc}
I_{m} & -\frac{i}{z}Q_{\pm}\\
\frac{i}{z}\epsilon Q_{\pm}^{\dagger} & I_{m}
\end{array}\right)=I_{2m}-\frac{i}{z}\underline{\sigma}_{3}\underline{Q}_{\pm}.
\end{align}
It follows that Eq.\eqref{T5} implies $[M_{\pm},N_{\pm}]=0$ under the NZBCs Eq.\eqref{T2} at infinity, which means that $N_{\pm}$ and $M_{\pm}$ are of the same eigenvectors. Thus one has
\addtocounter{equation}{1}
\begin{align}\label{T12}
M_{\pm}&=-i\lambda E_{\pm}\underline{\sigma}_{3},\tag{\theequation a}\\
N_{\pm}&=-i(4k^{2}+2k_{0}^{2})\lambda E_{\pm}\underline{\sigma}_{3}.\tag{\theequation b}
\end{align}
Obviously
\begin{align}\label{T13}
&\left\{\begin{aligned}
&\det E_{\pm}(z)=\left(\frac{2\lambda}{k+\lambda}\right)^{m}
\triangleq\gamma^{m}, \quad \gamma=1-\epsilon\frac{k_{0}^{2}}{z^{2}},\\
&E_{\pm}^{-1}=\frac{1}{\gamma}\left(I_{2m}+\frac{i}{z}
\underline{\sigma}_{3}\underline{Q}_{\pm}\right).
\end{aligned}\right.
\end{align}
Note that the inverse matrices $E_{\pm}^{-1}$ satisfy all $z$ values such that $\det E_{\pm}\neq 0$, specifically speaking, in the focusing case $z\neq\pm ik_{0}$, and $z\neq\pm k_{0}$ in the defocusing case.

The general continuous spectrum $\Sigma_{k}$ is composed of all $k$-values satisfying $\lambda(k)\in R$, namely $\Sigma_{k}=R \cup [-ik_{0},ik_{0}]$ in the focusing case, and $\Sigma_{k}=R\setminus [-k_{0},k_{0}]$ in the defocusing case.
These sets $\Sigma_{k}$ are mapped to $\Sigma_{z}=R\cup C_{0}$ and $\Sigma_{z}=R$ in the complex $z$-plane, respectively. The set $C_{0}$ denotes a circle of radius $k_{0}$ shown in Fig.1. For convenience, we omit the subscript from the context. Now we can obtain the Jost eigenfunctions $\Phi(x,t;z)$ and $\Psi(x,t;z)$  of the Lax pairs
\begin{align}\label{ME}
\begin{split}
\Phi(x,t;z)=(\phi(x,t;z),\bar{\phi}(x,t;z))=E_{-}(z)
e^{-i\theta(x,t;z)\underline{\sigma}_{3}}+o(1),
 \quad x\rightarrow-\infty,\\
\Psi(x,t;z)=(\bar{\psi}(x,t;z),\psi(x,t;z))=E_{+}(z)
e^{-i\theta(x,t;z)\underline{\sigma}_{3}}+o(1),
 \quad x\rightarrow+\infty,
 \end{split}
\end{align}
with
\begin{align}\label{T14}
\theta(x,t;z)=\lambda(z)\left(x+(4k^{2}(z)+2k_{0}^{2})t\right)
\end{align}
where the elements of the functions $\Phi(x,t;z)$ and $\Psi(x,t;z)$ are $2m\times m$ matrices. As usual, the Lax pairs Eq.\eqref{Lax} can be written as
\begin{align*}
\varphi_{\pm,x}&=M_{\pm}\varphi_{\pm}+\Delta\underline{Q}_{\pm}\varphi_{\pm},\\
\varphi_{\pm,t}&=N_{\pm}\varphi_{\pm}+\Delta\underline{\hat{Q}}_{\pm}\varphi_{\pm},
\end{align*}
with
\begin{align*}
\Delta \underline{Q}_{\pm}&=\underline{Q}-\underline{Q}_{\pm},\\
\Delta\underline{\hat{Q}}_{\pm}&=2ik\underline{Q}_{\pm}^{2}\underline{\sigma}_{3}
+2\underline{Q}_{\pm}^{3}-2ik(\underline{Q}^{2}
+\underline{Q}_{x})\underline{\sigma}_{3}-\underline{Q}_{xx}+2\underline{Q}^{3}
+\underline{Q}_{x}\underline{Q}-\underline{Q} \underline{Q}_{x}.
\end{align*}
By decomposing the  asymptotic exponential oscillations, we further introduce the modified eigenfunctions
\begin{align}\label{T15}
\begin{split}
\mathcal{W}(x,t;z)&=\left(\mathcal{\hat{W}}(x,t;z),\mathcal{\bar{W}}(x,t;z)\right)=
\Phi(x,t;z)e^{i\theta(x,t;z)\underline{\sigma}_{3}},\\
\mathcal{V}(x,t;z)&=\left(\mathcal{\bar{V}}(x,t;z),\mathcal{\hat{V}}(x,t;z)\right)=
\Psi(x,t;z)e^{i\theta(x,t;z)\underline{\sigma}_{3}}.
\end{split}
\end{align}
Note that $\mathop{\lim}\limits_{x\rightarrow-\infty} \mathcal{W}(x,t;z)=E_{-}(z)$ and
$\mathop{\lim}\limits_{x\rightarrow+\infty} \mathcal{V}(x,t;z)=E_{+}(z)$. Similar to Ref.\cite{Biondini-2014}, the following integral equations can be obtained
\begin{align}\label{T16}
\begin{split}
\mathcal{W}(x;z)
&=E_{-}+\int_{-\infty}^{x}E_{-}e^{-i\lambda(x-y)\underline{\sigma}_{3}}E_{-}^{-1}\Delta \underline{Q}_{-}(y)\left(\mathcal{\hat{W}}(y;z),\mathcal{\bar{W}}(y;z)\right)
e^{i\lambda(x-y)\underline{\sigma}_{3}}\, dy,\\
\mathcal{V}(x;z)
&=E_{+}-\int^{\infty}_{x}E_{+}e^{-i\lambda(x-y)\underline{\sigma}_{3}}E_{+}^{-1}\Delta \underline{Q}_{+}(y)\left(\mathcal{\hat{V}}(y;z),\mathcal{\bar{V}}(y;z)\right)
e^{i\lambda(x-y)\underline{\sigma}_{3}}\, dy.
\end{split}
\end{align}
From Eq.\eqref{T16}, the analyticity of the modified eigenfunctions can be summarized as follows
\begin{thm}
It is assumed that $Q(x,t)-Q_{+}\in L^{1}(a,+\infty)$ and $Q(x,t)-Q_{-}\in L^{1}(-\infty,a)$ hold for any constant $a\in R$, all $t>0$, and that the matrix potential function $Q(x,t)$ is the $m\times m$ symmetric matrix  as well as satisfies the boundary conditions Eq.\eqref{T6}. Then the modified eigenfunctions of the scattering problem determined by Eq.\eqref{ME} and \eqref{T15} satisfy that
the functions $\mathcal{\hat{W}}(x,t;z)$ and $\mathcal{\hat{V}}(x,t;z)$ are analytic in the region $D^{+}$ of $z$-plane, which are continuous up to $\partial D^{+}$; the functions $\mathcal{\bar{W}}(x,t;z)$ and $\mathcal{\bar{V}}(x,t;z)$ are analytic in the region $D^{-}$ of $z$-plane, which are continuous up to $\partial D^{-}$.
\end{thm}
\begin{proof}
For a matrix $\mathcal {M}$ with the product
\begin{align*}
e^{-i\lambda(x-y)\sigma_{3}}\mathcal {M}e^{i\lambda(x-y)\sigma_{3}}=\left(\begin{array}{cc}
m_{11} & e^{-2i\lambda(x-y)}m_{12}\\
e^{2i\lambda(x-y)}m_{21} & m_{22}\\
\end{array}\right).
\end{align*}
Taking the first column as an example, there is of following equation
\begin{align*}
e^{2i\lambda(x-y)}=e^{2i(Re\lambda+iIm\lambda)(x-y)}=
e^{2iRe\lambda(x-y)}e^{-2Im\lambda(x-y)},
\end{align*}
note that $x-y>0$, then we can derive that
$\mathcal{\hat{W}}(x,t;z)$ is analytic in the region $Im\lambda>0$, i.e., $D^{+}=\left\{z\in \mathbb{\textbf{C}}:\left(|z|^{2}-k_{0}^{2}\right)Im z>0\right\}$. The analyticity of the second column can be similarly proved.
\end{proof}
\subsection{Scattering matrix}
Because the trace of $M$ and $N$ in Lax pairs Eq.\eqref{Lax} are zero, resorting to \textbf{Liouville formula}\cite{Liu} the relationship can be derived
\begin{align*}
\partial_{x}\left(\det M\right)=\partial_{t}\left(\det N\right)=0,
\end{align*}
and for all $z\in\Sigma$,  Eq.\eqref{ME} implies that $\mathop{\lim}\limits_{x\rightarrow-\infty}\Phi(x,t;z)
e^{i\theta\underline{\sigma}_{3}}=E_{-}$ and
$\mathop{\lim}\limits_{x\rightarrow+\infty}\Psi(x,t;z)
e^{i\theta\underline{\sigma}_{3}}=E_{+}$, then we have
\begin{align}\label{T17}
\det\Psi(x,t;z)=\det\Phi(x,t;z)=\det E_{\pm}(z)=\gamma^{m}(z),\quad z\in\Sigma.
\end{align}
It is noted that the scattering problem is a first order homogeneous differential equation (ODE), and because both $\Psi(x,t;z)$ and $\Phi(x,t;z)$ are the solutions of the ODE for all $z\in\Sigma_{0}$, there is obviously a $2m\times2m$ constant matrix $S(z)$ which is independent of the variable $x$ and $t$ satisfying
\begin{align}\label{T18}
 \Phi(x,t;z)=\Psi(x,t;z)S(z),\quad z\in\Sigma_{0},
\end{align}
where $\Sigma_{0}=\Sigma\setminus\{\pm\sqrt{\epsilon}k_{0}\}$,
$S(z)=\left(\begin{array}{cc}
a(z) & \bar{b}(z)\\
b(z) & \bar{a}(z)
\end{array}\right)$, and the functions $a(z)$, $b(z)$, $\bar{b}(z)$, and $\bar{a}(z)$ are the $m\times m$ of  the scattering matrix. In additional, the elements in the scattering matrix $S(z)$ play an important role in the construction of Riemann-Hilbert problem, so their analytical regions need to be further determined. It follows from the Eqs.\eqref{ME} and \eqref{T18}  that
\begin{align}\label{T19}
\begin{split}
\phi(x,t;z)=\psi(x,t;z)b(z)+\bar{\psi}(x,t;z)a(z),\\
\bar{\phi}(x,t;z)=\psi(x,t;z)\bar{a}(z)+\psi(x,t;z)\bar{b}(z).
\end{split}
\end{align}
Combining with Eqs.\eqref{T17}, \eqref{T18}, one can obtain
\addtocounter{equation}{1}
\begin{align}
\det a(z)=\frac{Wr[\phi(x,t;z),\psi(x,t;z)]}{Wr[\bar{\psi}(x,t;z),\psi(x,t;z)]}=
\frac{Wr[\phi(x,t;z),\psi(x,t;z)]}{\gamma^m},\tag{\theequation a} \label{s11}\\
\det\bar{a}(z)=\frac{Wr[\bar{\psi}(x,t;z),\bar{\phi}(x,t;z)]}
{Wr[\bar{\psi}(x,t;z),\psi(x,t;z)]}=\frac{Wr[\bar{\psi}(x,t;z),\bar{\phi}(x,t;z)]}
{\gamma^m},\tag{\theequation b} \label{s22}
\end{align}
where the notation $Wr[\bullet,\bullet]$ represents the Wronskian determinant. In the case of scalars, similar to Ref.\cite{Biondini-2014}, the analytical region of the diagonal elements $a(z)$ and $\bar{a}(z)$ of the scattering matrix can be obtained directly from Eq.$(2.21)$, but only the analytical region of the diagonal elements determinant $\det a(z)$ and $\det\bar{a}(z)$ can be derived instead of the analytical region of the diagonal elements.
\begin{thm}\label{thm3}
If $Q(x,t)-Q_{+}\in L^{1}(a,+\infty)$ and $Q(x,t)-Q_{-}\in L^{1}(-\infty,a)$ hold for any constant $a\in R$, all $t>0$, and the potential function $Q(x,t)$ satisfies constraint condition \eqref{T6}, the scattering matrix $S(z)$ \eqref{T18} defined according to the eigenfunctions of the scattering problem satisfies that: the block $a(z)$ is analytic in the region $D^{+}$ of the $z$-plane, and continuous up to $\Sigma_{0}=\partial D^{+}\setminus\{\pm\sqrt{\epsilon}k_{0}\}$;
the block $\bar{a}(z)$ is analytic in the region $D^{-}$ of the $z$-plane, and continuous up to $\Sigma_{0}=\partial D^{-}\setminus\{\pm\sqrt{\epsilon}k_{0}\}$;
the off-diagonal blocks of the matrix $S(z)$ are nowhere analytic in general.
\end{thm}
\begin{rem}
For brevity, the proof of \textbf{Theorem} \eqref{thm3} will be explained when discussing symmetries.
\end{rem}
In order to establish a suitable Riemann-Hilbert problem, we need to properly arrange the modified eigenfunctions and the scattering coefficients so that they are analytic in the same region. Note that from Eqs.\eqref{T15} and \eqref{T19}, one has
\addtocounter{equation}{1}
\begin{align}
\mathcal{\hat{W}}(x,t;z)a^{-1}(z)&=\mathcal{\bar{V}}(x,t;z)+e^{2i\theta(x,t;z)}
\mathcal{\hat{V}}(x,t;z)\rho(z),\tag{\theequation a}\label{2.22a}\\
\mathcal{\bar{W}}(x,t;z)\bar{a}^{-1}(z)&=\mathcal{V}(x,t;z)+e^{-2i\theta(x,t;z)}
\mathcal{\bar{V}}(x,t;z)\bar{\rho}(z),\tag{\theequation b}
\end{align}
where $\mathcal{\hat{W}}(x,t;z)a^{-1}(z)$ and $\mathcal{\bar{W}}(x,t;z)\bar{a}^{-1}(z)$ are meromorphic in the region $D^{+}$ and $D^{-}$, respectively. Finally the reflection coefficients are introduced by
\begin{align}\label{Fans}
 \rho(z)=b(z)a^{-1}(z),\quad \bar{\rho}(z)=\bar{b}(z)\bar{a}^{-1}(z), \quad
 z\in\Sigma_{0}.
\end{align}
\subsection{Symmetries}
When using the Riemann-Hilbert method to solve the initial value problem, it is often necessary to consider the symmetry of the potential function in Lax pairs. This is because the symmetries of the eigenfunctions can be obtained by analyzing the symmetries of the potential function. Finally the symmetries of the scattering data are obtained, which is also the basis of the discrete spectral distribution. The symmetries of the non-zero boundary value problem is more complicated because of the fact that the Riemann surface is introduced, which causes the $\lambda(k)$ to change sign from one side of the Riemann surface to the other. It follows that from the uniformization variable $z$ Eq.\eqref{T9}:\\
$\bullet$ $z\mapsto z^{*}$ implies $(k,\lambda)\mapsto(k^{*},\lambda^{*})$;\\
$\bullet$ $z\mapsto\epsilon k_{0}^{2}/z$ implies $(k,\lambda)\mapsto(k,-\lambda)$.

It is worth noting the symmetries of the scattering problem corresponding to the above transformation, one of which is the conjugate symmetry that depends on the potential function $\underline{Q}(x,t)$ (i.e., $\underline{Q}^{\dagger}=\epsilon\underline{Q}$), and the other is due to the branching of the scattering parameter $k$-plane. In addition, we also discuss the third symmetry in combination with the assumption $Q^{T}=Q$. It is easy to verify that
\begin{align}\label{T21}
\underline{Q}=-\underline{\sigma}_{2}\underline{Q}^{T}\underline{\sigma}_{2},
\end{align}
where
$\underline{\sigma}_{2}=i\left(\begin{array}{cc}
0_{m} & I_{m}\\
-I_{m} & 0_{m}
\end{array}\right)$ as a generaliztion of the $2\times2$ Pauli matrix $\sigma_{2}$.
\subsubsection{The first symmetry}
In Ref.\cite{Ablowitz-2004}, Ablowitz and his co-authors studied  the relationship of the
scattering data and eigenfunctions for the matrix potential function with zero boundary value condition when the above transformation is involved. We will extend their method to nonzero boundary conditions. Now introducing the functions which are independent of the variable $x$  for $z\in\Sigma$
\begin{align}\label{T22}
\mathscr{A}(x,t;z)=\Phi^{\dagger}(x,t;z^{*})\mathcal {L}_{\epsilon}\Phi(x,t;z),\quad
\mathscr{B}(x,t;z)=\Psi^{\dagger}(x,t;z^{*})\mathcal {L}_{\epsilon}\Psi(x,t;z),
\end{align}
one can obtain as $x\rightarrow\pm\infty$
\begin{align}\label{T23}
\Phi^{\dagger}(x,t;z^{*})\mathcal {L}_{\epsilon}\Phi(x,t;z)=
\Psi^{\dagger}(x,t;z^{*})\mathcal {L}_{\epsilon}\Psi(x,t;z)=\gamma(z)\mathcal {L}_{\epsilon},
\end{align}
where $\mathcal {L}_{\epsilon}=\left(\begin{array}{cc}
I_{m} & 0_{m}\\
0_{m} & -\epsilon I_{m}
\end{array}\right)$, and the value of $\epsilon$ represents the focusing and defocusing case.  It follows from Eqs.\eqref{Lax} and \eqref{ME} that
\begin{align*}
&\partial_{x}\mathscr{A}(x,t;z)=\Phi_{x}^{\dagger}(x,t;z^{*})\mathcal {L}_{\epsilon}\Phi(x,t;z)+\Phi^{\dagger}(x,t;z^{*})\mathcal {L}_{\epsilon}\Phi_{x}(x,t;z)=\\ &\left(ik\underline{\sigma}_{3}+\epsilon\underline{Q}\right)
\Phi^{\dagger}(x,t;z^{*})\mathcal {L}_{\epsilon}\Phi(x,t;z)+\Phi^{\dagger}(x,t;z^{*})\mathcal {L}_{\epsilon}\left(-ik\underline{\sigma}_{3}+\underline{Q}\right)\Phi(x,t;z)=0,
\end{align*}
the other can be proved in the same way, and Eq.\eqref{T23} can be directly calculated.

Obviously Eq.\eqref{T23} is equivalent to
\begin{align}\label{T24}
\Phi^{-1}(x,t;z)=\frac{1}{\gamma(z)}\mathcal {L}_{\epsilon}\Phi^{\dagger}(x,t;z^{*})
\mathcal {L}_{\epsilon},\quad
\Psi^{-1}(x,t;z)=\frac{1}{\gamma(z)}\mathcal {L}_{\epsilon}\Psi^{\dagger}(x,t;z^{*})
\mathcal {L}_{\epsilon}.
\end{align}
\begin{prop}\label{prop4}
The elements of the scattering matrix $S(z)$ can be specifically expressed according to the Jost eigenfunctions as
\addtocounter{equation}{1}
\begin{align}
\gamma(z)a(z)=\left(\bar{\psi}^{up}(x,t;z^{*})\right)^{\dagger}\phi^{up}(x,t;z)-
\epsilon\left(\bar{\psi}^{dn}(x,t;z^{*})\right)
^{\dagger}\phi^{dn}(x,t;z),\tag{\theequation a} \label{a1}\\
\gamma(z)\bar{a}(z)=\left(\psi^{dn}(x,t;z^{*})\right)^{\dagger}\bar{\phi}^{dn}(x,t;z)-
\epsilon\left(\psi^{up}(x,t;z^{*})\right)
^{\dagger}\bar{\phi}^{up}(x,t;z),\tag{\theequation b}\label{a2}\\
\gamma(z)b(z)=\left(\psi^{dn}(x,t;z^{*})\right)^{\dagger}\phi^{dn}(x,t;z)-
\epsilon\left(\psi^{up}(x,t;z^{*})\right)
^{\dagger}\phi^{up}(x,t;z),\tag{\theequation c}\\
\gamma(z)\bar{b}(z)=\left(\bar{\psi}^{up}(x,t;z^{*})\right)^{\dagger}\bar{\phi}^{up}(x,t;z)-
\epsilon\left(\bar{\psi}^{dn}(x,t;z^{*})\right)
^{\dagger}\bar{\phi}^{dn}(x,t;z),\tag{\theequation d}
\end{align}
where $a(z)$ and $\bar{a}(z)$ are analytic in the region $D^{+}$ and $D^{-}$ of $z$-plane, respectively.
\end{prop}
\begin{proof}
For simplicity, we take the following blocks for the eigenfunctions
\begin{align*}
\Phi(x,t;z)=\left(\begin{array}{cc}
\phi^{up} & \bar{\phi}^{up}\\
\phi^{dn} & \bar{\phi}^{dn}
\end{array}\right),\quad
\Psi(x,t;z)=\left(\begin{array}{cc}
\bar{\psi}^{up} & \psi^{up}\\
\bar{\psi}^{dn} & \psi^{dn}
\end{array}\right),
\end{align*}
where $A^{up/dn}$ denote an $m\times m$ matrix.
It follows based on  Eqs.\eqref{T18} and \eqref{T24} that
\begin{align}\label{T25}
S(z)=\Psi^{-1}(x,t;z)\Phi(x,t;z)=
\frac{1}{\gamma(z)}\mathcal {L}_{\epsilon}\Psi^{\dagger}(x,t;z^{*})
\mathcal {L}_{\epsilon}\Phi(x,t;z),
\end{align}
with
\begin{align*}
&\mathcal {L}_{\epsilon}\Psi^{\dagger}(x,t;z^{*})
\mathcal {L}_{\epsilon}\Phi(x,t;z)=\\
&\left(\begin{array}{cc}
\left(\bar{\psi}^{up}(z^{*})\right)^{\dagger}\phi^{up}(z)-
\epsilon\left(\bar{\psi}^{dn}(z^{*})\right)^{\dagger}\phi^{dn}(z) & \left(\bar{\psi}^{up}(z^{*})\right)^{\dagger}\bar{\phi}^{up}(z)-
\epsilon\left(\bar{\psi}^{dn}(z^{*})\right)
^{\dagger}\bar{\phi}^{dn}(z)\\
\left(\psi^{dn}(z^{*})\right)^{\dagger}\phi^{dn}(z)-
\epsilon\left(\psi^{up}(z^{*})\right)
^{\dagger}\phi^{up}(z) & \left(\psi^{dn}(z^{*})\right)^{\dagger}\bar{\phi}^{dn}(z)-
\epsilon\left(\psi^{up}(z^{*})\right)
^{\dagger}\bar{\phi}^{up}(z)
\end{array}\right).
\end{align*}
Obviously the Eqs.($2.26$) can be derived. In fact, the analyticity of the Eqs.\eqref{a1} and \eqref{a2} can be obtained from Theorem \eqref{thm3}.
\end{proof}
\begin{rem}
Theorem \eqref{thm3} is the direct result of Proposition \eqref{prop4}.
\end{rem}
\begin{thm}
Assume that  $Q(x,t)-Q_{+}\in L^{1}(a,+\infty)$ and $Q(x,t)-Q_{-}\in L^{1}(-\infty,a)$ hold for any constant $a\in R$, all $t>0$, and also assume that the scattering coefficients $a(z)$, $b(z)$, $\bar{b}(z)$, $\bar{a}(z)$ have simple zeros at branch points $z=\pm\sqrt{\epsilon}k_{0}$, the following residue conditions can be written as
\begin{align*}
&\mathop{Res}_{z=\pm k_{0}}a(z)=\pm\frac{k_{0}}{2}\left[
\left(\bar{\psi}^{up}(x,t;\pm k_{0})\right)^{\dagger}\phi^{up}(x,t;\pm k_{0})-
\epsilon\left(\bar{\psi}^{dn}(x,t;\pm k_{0})\right)
^{\dagger}\phi^{dn}(x,t;\pm k_{0})\right],\\
&\mathop{Res}_{z=\pm k_{0}}\bar{a}(z)=\pm\frac{k_{0}}{2}\left[
\left(\psi^{dn}(x,t;\pm k_{0})\right)^{\dagger}\bar{\phi}^{dn}(x,t;\pm k_{0})-
\epsilon\left(\psi^{up}(x,t;\pm k_{0})\right)
^{\dagger}\bar{\phi}^{up}(x,t;\pm k_{0})\right],\\
&\mathop{\lim}\limits_{z\rightarrow\pm k_{0}}(z\mp k_{0})b(z)=\pm\frac{k_{0}}{2}\left[
\left(\psi^{dn}(x,t;\pm k_{0})\right)^{\dagger}\phi^{dn}(x,t;\pm k_{0})-
\epsilon\left(\psi^{up}(x,t;\pm k_{0})\right)
^{\dagger}\phi^{up}(x,t;\pm k_{0})\right],\\
&\mathop{\lim}\limits_{z\rightarrow\pm k_{0}}(z\mp k_{0})\bar{b}(z)=\pm\frac{k_{0}}{2}\left[
\left(\bar{\psi}^{up}(x,t;\pm k_{0})\right)^{\dagger}\bar{\phi}^{up}(x,t;\pm k_{0})-
\epsilon\left(\bar{\psi}^{dn}(x,t;\pm k_{0})\right)
^{\dagger}\bar{\phi}^{dn}(x,t;\pm k_{0})\right]
\end{align*}
in the defocusing case and
\begin{align*}
&\mathop{Res}_{z=\pm ik_{0}}a(z)=\pm\frac{ik_{0}}{2}\left[
\left(\bar{\psi}^{up}(x,t;\mp ik_{0})\right)^{\dagger}\phi^{up}(x,t;\pm ik_{0})+
\epsilon\left(\bar{\psi}^{dn}(x,t;\mp ik_{0})\right)
^{\dagger}\phi^{dn}(x,t;\pm ik_{0})\right],\\
&\mathop{Res}_{z=\pm ik_{0}}\bar{a}(z)=\pm\frac{ik_{0}}{2}\left[
\left(\psi^{dn}(x,t;\mp ik_{0})\right)^{\dagger}\bar{\phi}^{dn}(x,t;\pm ik_{0})+
\epsilon\left(\psi^{up}(x,t;\mp ik_{0})\right)
^{\dagger}\bar{\phi}^{up}(x,t;\pm ik_{0})\right],\\
&\mathop{\lim}\limits_{z\rightarrow\pm ik_{0}}(z\mp ik_{0})b(z)=\pm\frac{ik_{0}}{2}\left[
\left(\psi^{dn}(x,t;\mp ik_{0})\right)^{\dagger}\phi^{dn}(x,t;\pm ik_{0})+
\epsilon\left(\psi^{up}(x,t;\mp ik_{0})\right)
^{\dagger}\phi^{up}(x,t;\pm ik_{0})\right],\\
&\mathop{\lim}\limits_{z\rightarrow\pm ik_{0}}(z\mp ik_{0})\bar{b}(z)=\pm\frac{ik_{0}}{2}\left[
\left(\bar{\psi}^{up}(x,t;\mp ik_{0})\right)^{\dagger}\bar{\phi}^{up}(x,t;\pm ik_{0})+\epsilon\left(\bar{\psi}^{dn}(x,t;\mp ik_{0})\right)
^{\dagger}\bar{\phi}^{dn}(x,t;\pm ik_{0})\right],
\end{align*}
in the focusing case. The reflection coefficients $\rho(z)$ and $\bar{\rho}(z)$ determined by Eq.\eqref{Fans} are of a removable singularity under the conditions $\det a(z)\neq0$ and $\det \bar{a}(z)\neq0$ for all $z\in\Sigma$.
\end{thm}
\begin{proof}
Taking the focusing case ($\epsilon=-1$) as example, we have from Eq.\eqref{a1}
\begin{align*}
a(z)=\frac{\left(\bar{\psi}^{up}(x,t;z^{*})\right)^{\dagger}\phi^{up}(x,t;z)-
\epsilon\left(\bar{\psi}^{dn}(x,t;z^{*})\right)
^{\dagger}\phi^{dn}(x,t;z)}{\gamma(z)}\triangleq \frac{f(z)}{\gamma(z)},
\end{align*}
thus
\begin{align*}
\mathop{Res}_{z=ik_{0}}a(z)=\mathop{Res}_{z=ik_{0}}\frac{f(z)}{\gamma(z)}=
\frac{f(z)}{\partial_{z}\gamma(z)}\left| _{z=ik_{0}}\right.=\frac{ik_{0}}{2}f(ik_{0}).
\end{align*}
The rest can be proved similarly.
\end{proof}

Next we will discuss the relationships between the scattering data $\bar{a}(z)$ and $a(z)$, which are related to the distribution of zero points.
\begin{prop}
For all $z\in D^{-}$, the scattering data admit that
\begin{align}\label{TT27}
\det\bar{a}(z)=\det a^{\dagger}(z^{*})=(\det a(z^{*}))^{*}.
\end{align}
\end{prop}
\begin{proof}
 Eq.\eqref{T25} implies
\addtocounter{equation}{1}
\begin{align}\label{T26}
\left(\phi^{up}(x,t;z^{*})\right)^{\dagger}\bar{\phi}^{up}(x,t;z)=
\epsilon\left(\phi^{dn}(x,t;z^{*})\right)^{\dagger}
\bar{\phi}^{dn}(x,t;z),\tag{\theequation a}\\
\left(\psi^{up}(x,t;z^{*})\right)^{\dagger}\bar{\psi}^{up}(x,t;z)=
\epsilon\left(\psi^{dn}(x,t;z^{*})\right)^{\dagger}
\bar{\psi}^{dn}(x,t;z). \tag{\theequation b}
\end{align}
Moreover, using Eq.\eqref{T25} one has
\addtocounter{equation}{1}
\begin{align}
S^{-1}(z)&=\gamma(z)\Phi^{-1}(x,t;z)\mathcal {L}_{\epsilon}
\left(\Psi^{\dagger}\right)^{-1}(x,t;z^{*})
\mathcal {L}_{\epsilon}, \tag{\theequation a}\\
S^{\dagger}(z^{*})&=
\frac{1}{\gamma^{*}(z^{*})}\Phi^{\dagger}(x,t;z)\mathcal {L}_{\epsilon}\Psi(x,t;z)
\mathcal {L}_{\epsilon}. \tag{\theequation b}
\end{align}
Combining with Eq.\eqref{T24}, we have
\begin{align*}
\mathcal {L}_{\epsilon}S^{-1}(z)=\frac{1}{S(z)}
\Phi^{\dagger}(x,t;z^{*})\mathcal {L}_{\epsilon}\Psi(x,t;z)
=S^{\dagger}(z^{*})\mathcal {L}_{\epsilon},
\end{align*}
namely
\begin{align}\label{T27}
S^{\dagger}(z^{*})\mathcal {L}_{\epsilon}S^{-1}(z)=\mathcal {L}_{\epsilon},
\quad z\in\Sigma,
\end{align}
which in turn yields
\addtocounter{equation}{1}
\begin{align}
a^{\dagger}(z^{*})a(z)-\epsilon b^{\dagger}(z^{*})b(z)&=I_{m},\tag{\theequation a}\\
a^{\dagger}(z^{*})\bar{b}(z)-\epsilon b^{\dagger}(z^{*})\bar{a}(z)&=0_{m},\tag{\theequation b}\\
\bar{b}^{\dagger}(z^{*})a(z)-\epsilon \bar{a}^{\dagger}(z^{*})b(z)&=0_{m},\tag{\theequation c}\label{x3}\\
\bar{b}^{\dagger}(z^{*})\bar{b}(z)-\epsilon \bar{a}^{\dagger}(z^{*})\bar{a}(z)&=-\epsilon I_{m}.\tag{\theequation d}
\end{align}
On the other hand, the symmetry of the reflection coefficients can be obtained from Eq.\eqref{x3}. Taking the conjugate transpose of Eq.\eqref{x3} yields
\begin{align*}
a^{\dagger}(z)\bar{b}(z^{*})=\epsilon b^{\dagger}(z)\bar{a}(z^{*}),
\end{align*}
then
\begin{align*}
a^{\dagger}(z^{*})\bar{b}(z)=\epsilon b^{\dagger}(z^{*})\bar{a}(z)\Rightarrow
\epsilon a^{\dagger}(z^{*})\bar{b}(z)\bar{a}^{-1}(z)=b^{\dagger}(z^{*}).
\end{align*}
Therefore we have the symmetry
\begin{align}
\epsilon\rho^{\dagger}(z^{*})=\epsilon \left(a^{\dagger}(z^{*})\right)^{\dagger}
\epsilon a^{\dagger}(z^{*})\bar{b}(z)\bar{a}^{-1}(z)=\bar{\rho}(z),
\end{align}
as well as
\begin{align}\label{SSS1}
a(z)a^{\dagger}(z^{*})=\left[I_{m}-\rho^{\dagger}(z^{*})\rho(z)\right],\quad
\bar{a}(z)\bar{a}^{\dagger}(z^{*})=\left[I_{m}-\bar{\rho}
^{\dagger}(z^{*})\bar{\rho}(z)\right].
\end{align}
Note that Eq.\eqref{T27} can be written as
\begin{align}\label{T28}
S^{-1}(z)=\mathcal {L}_{\epsilon}S^{\dagger}(z^{*})\mathcal {L}_{\epsilon},\quad
S^{-1}(z)=\left(\begin{array}{cc}
\bar{c}(z) & d(z)\\
\bar{d}(z) & c(z)
\end{array}\right).
\end{align}
According to the properties of the matrix, Eq.\eqref{T28} is equivalent to the elements the matrix $S^{-1}(z)$. We can get
\addtocounter{equation}{1}
\begin{align}
\bar{c}(z)&=a^{\dagger}(z^{*}), \quad c(z)=\bar{a}^{\dagger}(z^{*}),\tag{\theequation a} \label{38a}\\
d(z)&=-\epsilon b^{\dagger}(z^{*}),\quad \bar{d}(z)=-\epsilon\bar{b}^{\dagger}(z^{*}). \tag{\theequation b}
\end{align}
Similar to the method of solving Eq.($2.21$), we can get
\addtocounter{equation}{1}
\begin{align}
\det c(z)=\frac{Wr[\phi,\psi]}{Wr[\phi,\bar{\phi}]}=
\frac{Wr[\phi,\psi]}{\gamma^m},\tag{\theequation a}\\
\det\bar{c}(z)=\frac{Wr[\bar{\psi},\bar{\phi}]}
{Wr[\phi,\bar{\phi}]}=\frac{Wr[\bar{\psi},\bar{\phi}]}
{\gamma^m}.\tag{\theequation b}
\end{align}
It is worth noting that according to the analyticity of the eigenfunctions, $\det c(z)$ and $\det\bar{c}(z)$ are analytic on regions $D^{+}$ and $D^{-}$, respectively. On the other hand, it can be summarized
\begin{align}\label{T29}
\det c(z)=\det a(z),\quad z\in D^{+},\quad
\det \bar{c}(z)=\det \bar{a}(z),\quad z\in D^{-},
\end{align}
and from Eq.($2.38$) the proposition $2.7$ is proved.
\end{proof}
\subsubsection{The second symmetry}
The process of discussing the first symmetries of nonzero boundary values is the same as the symmetries of zero boundary values, but the discussion of the second symmetry is more complicated due to the fact that Riemann surface is introduced under NZBCs. Because the Jost functions and the scattering coefficients depend on the function $\lambda(k)$ (it changes sign on the Riemann surface), we need to establish their relationships on different sheet of the Riemann surface.
\begin{prop}
The eigenvectors $E_{\pm}$ satisfy the following relationship
\begin{align}\label{T30}
E_{\pm}(z)=-\frac{i}{z}E_{\pm}\left(\frac{\epsilon k_{0}^{2}}{z}\right){\underline{\sigma}_{3}}\underline{Q}_{\pm}.
\end{align}
\end{prop}
\begin{proof}
Recalling the Eqs.\eqref{T4}, \eqref{T10} and \eqref{T14}, it is easy to check that
\begin{align}
&\left\{ \begin{aligned}
&k(\epsilon k_{0}^{2}/z)=k(z),\\
&\lambda(\epsilon k_{0}^{2}/z)=-\lambda(z),\\
&\theta(\epsilon k_{0}^{2}/z)=-\theta(z),\\
&\underline{Q}_{\pm}e^{-i\theta(z)\underline{\sigma}_{3}}=
e^{i\theta(z)\underline{\sigma}_{3}}\underline{Q}_{\pm}.
\end{aligned}\right.
\end{align}
The Eq.\eqref{T30} is equivalent to
\begin{align*}
E_{\pm}(z)=-\frac{i}{z}E_{\pm}\left(\frac{\epsilon k_{0}^{2}}{z}\right)e^{-i\theta(z)\underline{\sigma}_{3}}
{\underline{\sigma}_{3}}\underline{Q}_{\pm}e^{-i\theta(z)\underline{\sigma}_{3}},
\end{align*}
note that
\begin{align*}
-\frac{i}{z}E_{\pm}\left(\frac{\epsilon k_{0}^{2}}{z}\right)e^{-i\theta(z)\underline{\sigma}_{3}}
{\underline{\sigma}_{3}}\underline{Q}_{\pm}e^{-i\theta(z)\underline{\sigma}_{3}}
&=
-\frac{i}{z}\left(e^{-i\theta(z)\underline{\sigma}_{3}}
{\underline{\sigma}_{3}}\underline{Q}_{\pm}-\frac{i}{\epsilon k_{0}^{2}/z}{\underline{\sigma}_{3}}\underline{Q}_{\pm}
e^{-i\theta(z)\underline{\sigma}_{3}}
{\underline{\sigma}_{3}}\underline{Q}_{\pm}\right)
e^{-i\theta(z)\underline{\sigma}_{3}}\\&=
-\frac{i}{z}\left({\underline{\sigma}_{3}}\underline{Q}_{\pm}+iz\right)=
I_{m}-\frac{i}{z}
{\underline{\sigma}_{3}}\underline{Q}_{\pm}=E_{\pm}(z).
\end{align*}
\end{proof}
\begin{prop}
The Jost eigenfunctions satisfy that for $z\in\Sigma$
\begin{align}\label{T31}
\Phi(x,t;z)=\frac{1}{iz}\Phi(x,t;\epsilon k_{0}^{2}/z)
\underline{\sigma}_{3}\underline{Q}_{-},\quad
\Psi(x,t;z)=\frac{1}{iz}\Psi(x,t;\epsilon k_{0}^{2}/z)
\underline{\sigma}_{3}\underline{Q}_{+}.
\end{align}
\end{prop}
\begin{proof}
Using the Eq.\eqref{ME} implies that
\begin{align*}
\Phi(x,t;\epsilon k_{0}^{2}/z)=E_{-}(\epsilon k_{0}^{2}/z)
e^{-i\theta(\epsilon k_{0}^{2}/z)\underline{\sigma}_{3}}=E_{-}(\epsilon k_{0}^{2}/z)
e^{-i\theta(z)\underline{\sigma}_{3}},
\end{align*}
and that
\begin{align*}
\Phi(x,t;z)=E_{-}(z)e^{-i\theta(z)\underline{\sigma}_{3}}=-\frac{i}{z}
E_{-}(\epsilon k_{0}^{2}/z)\underline{\sigma}_{3}\underline{Q}_{-}
e^{-i\theta(z)\underline{\sigma}_{3}}=\\-\frac{i}{z}E_{-}(\epsilon k_{0}^{2}/z)
e^{i\theta(z)\underline{\sigma}_{3}}\underline{\sigma}_{3}\underline{Q}_{-}=
\frac{1}{iz}\Phi(x,t;\epsilon k_{0}^{2}/z)
\underline{\sigma}_{3}\underline{Q}_{+}.
\end{align*}
The other can be proved in the same way.
\end{proof}
Expanding the Eq.\eqref{T31}, one has
\addtocounter{equation}{1}
\begin{align}
\phi(x,t;z)=\frac{i\epsilon}{z}\bar{\phi}(x,t;\epsilon k_{0}^{2}/z)Q_{-}^{\dagger},\quad
\bar{\phi}(x,t;z)=-\frac{i}{z}\phi(x,t;\epsilon k_{0}^{2}/z)Q_{-},\tag{\theequation a}\label{44a}\\
\bar{\psi}(x,t;z)=\frac{i\epsilon}{z}\psi(x,t;\epsilon k_{0}^{2}/z)Q_{+}^{\dagger},\quad
\psi(x,t;z)=-\frac{i}{z}\bar{\psi}(x,t;\epsilon k_{0}^{2}/z)Q_{+}.\tag{\theequation b}\label{44b}
\end{align}
\begin{prop}
The scattering matrix $S(z)$ defined by the Eq.\eqref{T18} admits
\begin{align}\label{T32}
S(\epsilon k_{0}^{2}/z)=\underline{\sigma}_{3}\underline{Q}_{+}S(z)
\underline{Q}^{-1}\underline{\sigma}_{3}=\frac{\epsilon}{k^{2}_{0}}
\underline{\sigma}_{3}\underline{Q}_{+}S(z)\underline{Q}_{-}\underline{\sigma}_{3}.
\end{align}
\end{prop}
\begin{proof}
The Eqs.\eqref{T18} and \eqref{T31} mean that
\begin{align*}
\left\{\begin{aligned}
&S(z)=\Psi^{-1}(z)\Phi(z),\\
&\Psi^{-1}(k_{0}^{2}/z)=-\frac{i}{z}\underline{\sigma}_{3}\underline{Q}_{+}\Psi^{-1}(z),\\
&\Phi(k_{0}^{2}/z)=iz\Phi(z)\left(\underline{\sigma}_{3}\underline{Q}_{-}\right)^{-1}.
\end{aligned}\right.
\end{align*}
\begin{align*}
S(k_{0}^{2}/z)=\Psi^{-1}(k_{0}^{2}/z)\Phi(k_{0}^{2}/z)=
-\frac{i}{z}\underline{\sigma}_{3}\underline{Q}_{+}\Psi^{-1}(z)
iz\Phi(z)\left(\underline{\sigma}_{3}\underline{Q}_{-}\right)^{-1}\\=
\underline{\sigma}_{3}\underline{Q}_{+}\Psi^{-1}(z)\Phi(z)\underline{Q}_{-}^{-1}
\underline{\sigma}_{3}=\frac{\epsilon}{k^{2}_{0}}
\underline{\sigma}_{3}\underline{Q}_{+}S(z)\underline{Q}_{-}\underline{\sigma}_{3}.
\end{align*}
Note that $\underline{Q}_{-}^{-1}=\frac{\epsilon}{k^{2}_{0}}\underline{Q}_{-}$.
\end{proof}

Combining with the Eq.\eqref{T18} and expanding the Eq.\eqref{T32} we have
\addtocounter{equation}{1}
\begin{align}
a(k_{0}^{2}/z)&=\frac{1}{k_{0}^{2}}Q_{+}\bar{a}(z)Q_{-}^{\dagger},\quad
\bar{a}(k_{0}^{2}/z)=\frac{1}{k_{0}^{2}}Q_{+}^{\dagger}a(z)Q_{-},\tag{\theequation a}\label{46a}\\
b(k_{0}^{2}/z)&=-\frac{\epsilon}{k_{0}^{2}}Q_{+}^{\dagger}\bar{b}(z)Q_{-}^{\dagger},\quad
\bar{b}(k_{0}^{2}/z)=-\frac{\epsilon}{k_{0}^{2}}Q_{+}b(z)Q_{-}.
\tag{\theequation b}\label{46b}
\end{align}
Based on the Eq.($2.46$), the reflection coefficient $\rho(z)$ satisfies the following symmetry for $z\in\Sigma$
\begin{align}
\rho(k_{0}^{2}/z)=-\epsilon Q_{+}^{\dagger}\bar{\rho}(z)Q_{+}^{-1}=
-\frac{\epsilon}{k_{0}^{2}}Q_{+}^{\dagger}\bar{\rho}(z)Q_{+}^{\dagger}.
\end{align}

\subsubsection{The third symmetry}
The third symmetry is based on our assumption that the  potential function $Q(x,t)$ is a symmetric matrix. We next analyze the properties of the scattering data $S(z)$ under this condition, i.e., $Q=Q^{T}$. Similar to the first symmetric process, we first introduce two functions $\tilde{f}(x,t;z)$ and $\tilde{g}(x,t;z)$ that are independent of the variable $x$ by using the Jost eigenfunctions of the Lax pairs to further derive the relationship of the matrix scattering data, i.e.,
\begin{align}
\tilde{f}(x,t;z)=\Phi^{T}(x,t;z)\underline{\sigma}_{2}\Phi(x,t;z),\quad
\tilde{g}(x,t;z)=\Psi^{T}(x,t;z)\underline{\sigma}_{2}\Psi(x,t;z).
\end{align}
From Eq.\eqref{Lax}, one has
\begin{align*}
\partial_{x}\tilde{f}(x,t;z)&=\Phi^{T}_{x}(x,t;z)\underline{\sigma}_{2}\Phi(x,t;z)+
\Phi^{T}(x,t;z)\underline{\sigma}_{2}\Phi_{x}(x,t;z)\\&=
\Phi^{T}_{x}(x,t;z)\left(-ik\underline{\sigma}_{3}\underline{\sigma}_{2}+
\underline{Q}^{T}\underline{\sigma}_{2}-ik\underline{\sigma}_{2}\underline{\sigma}_{3}
+\underline{\sigma}_{2}\underline{Q}\right)\Phi_{x}(x,t;z)=0_{2m},
\end{align*}
which shows that the function $\tilde{f}(x,t;z)$ is independent with the variable $x$.
\begin{prop}
The scattering matrix $S(z)$ admits
\begin{align}\label{T33}
S^{T}(z)\underline{\sigma}_{2}S(z)=\underline{\sigma}_{2},\quad z\in\Sigma.
\end{align}
\end{prop}
\begin{proof}
The relationship can be obtained by the Eqs.\eqref{ME} and \eqref{T14}
\begin{align*}
\Phi(z)=E_{-}(z)e^{-i\theta(z)\underline{\sigma}_{3}},\quad
\Psi(z)=E_{+}(z)e^{-i\theta(z)\underline{\sigma}_{3}}.
\end{align*}
Moreover, one has
\begin{align}
\begin{split}
\Phi^{T}(x,t;z)\underline{\sigma}_{2}\Phi(x,t;z)=
e^{-i\theta(z)\underline{\sigma}_{3}}E^{T}_{-}(z)\underline{\sigma}_{2}E_{-}(z)
e^{-i\theta(z)\underline{\sigma}_{3}}=\gamma(z)\underline{\sigma}_{2},
\end{split}
\end{align}
as $x\rightarrow\pm\infty$, namely
\addtocounter{equation}{1}
\begin{align}
\Phi^{T}(x,t;z)\underline{\sigma}_{2}\Phi(x,t;z)=\gamma(z)\underline{\sigma}_{2},
\tag{\theequation a}\\
\Psi^{T}(x,t;z)\underline{\sigma}_{2}\Psi(x,t;z)=\gamma(z)\underline{\sigma}_{2}.
\tag{\theequation b}\label{52b}
\end{align}
The Eq.\eqref{52b} can be derived in the same way. On the other hand, the Eq.\eqref{T18} implies
\begin{align}
\Phi^{T}(z)=S^{T}(z)\Psi^{T}(z).
\end{align}
Combining with Eq.($2.52$) we have
\begin{align}
S^{T}(z)\Psi^{T}(z)\underline{\sigma}_{2}\Psi(z)S(z)=
S^{T}(z)\gamma(z)\underline{\sigma}_{2}S(z)=\gamma(z)\underline{\sigma}_{2}
\Rightarrow S^{T}(z)\underline{\sigma}_{2}S(z)=\underline{\sigma}_{2}.
\end{align}
\end{proof}
Expanding   Eq.\eqref{T33} yields
\addtocounter{equation}{1}
\begin{align}
&b^{T}(z)a(z)=a^{T}(z)b(z),\tag{\theequation a}\\
&\bar{b}^{T}(z)\bar{a}(z)=\bar{a}^{T}(z)\bar{b}(z),\tag{\theequation b}\\
&\bar{a}^{T}(z)\bar{a}(z)-\bar{b}^{T}(z)\bar{b}(z)=I_{m},\tag{\theequation c}
\end{align}
which in turn imply for the reflection coefficients
\begin{align}
\rho(z)=\rho^{T}(z),\quad \bar{\rho}(z)=\bar{\rho}^{T}(z), \quad z\in\Sigma.
\end{align}
In addition, we can derive
\begin{align}\label{T34}
a(z)a^{T}(z)=\left(I_{m}-\bar{\rho}(z)\rho(z)\right)^{-1},\quad z\in\Sigma.
\end{align}
A brief proof is given below
\begin{align*}
a(z)a^{T}(z)\left(I_{m}-\bar{\rho}(z)\rho(z)\right)&=
a(z)a^{T}(z)-a(z)\bar{a}^{T}(z)\bar{b}(z)\bar{a}^{-1}(z)b(z)a^{-1}(z)\\&=
a(z)a^{T}(z)-a(z)(a^{T}(z)a(z)-I_{m})a^{-1}(z)=I_{m},
\end{align*}
which is equivalent to the Eq.\eqref{T34}. Finally, we give the relationship between the elements of the inverse scattering matrix $S^{-1}(z)$ and the elements of the scattering matrix $S(z)$, i.e.,
\begin{align}
S^{-1}(z)=\underline{\sigma}_{2}S^{T}(z)\underline{\sigma}_{2},\quad z\in\Sigma,
\end{align}
which in particular means that
\addtocounter{equation}{1}
\begin{align}
\bar{c}(z)&=\bar{a}^{T}(z),\quad c(z)=a^{T}(z),\tag{\theequation a}\\
d(z)&=-\bar{b}^{T}(z),\quad \bar{d}(z)=-b^{T}(z),\tag{\theequation b}
\end{align}
and that
\begin{align}
\bar{a}(z)=a^{*}(z^{*}),\quad \bar{b}(z)=\epsilon b^{*}(z^{*}),\quad  z\in\Sigma.
\end{align}

Based on the analysis of the above three symmetries, we summarize the scattering data as follows
\begin{thm}
Assume that  $Q(x,t)-Q_{+}\in L^{1}(a,+\infty)$ and $Q(x,t)-Q_{-}\in L^{1}(-\infty,a)$ hold for any constant $a\in R$, all $t>0$. The reflection coefficients $\rho(z)$ and $\bar{\rho}(z)$ admit that
\begin{align}
\rho(z)=\epsilon\bar{\rho}(z^{*}),\quad
\rho(\epsilon k_{0}^{2}/z)=-\frac{\epsilon}{k_{0}^{2}/z}
Q_{+}^{\dagger}\bar{\rho}(z)Q_{+}^{\dagger},\quad z\in\Sigma_{0}.
\end{align}
In addition, the transmission coefficients satisfy the symmetry for $z\in D^{-}\cup\Sigma_{0}$
\begin{align}
\det\bar{a}(z)=\det\bar{a}(z^{*}),\quad \det\bar{a}(z)=\frac{k_{0}^{2m}}{\det Q_{+}(\det Q_{-})^{*}}\det a(\epsilon k_{0}^{2}/z).
\end{align}
If the potential $Q(x,t)$ is a symmetric matrix, the reflection coefficients are also the symmetric matrix, i.e.,
\begin{align}
\rho(z)=\rho^{T}(z),\quad \bar{\rho}(z)=\bar{\rho}^{T}(z), \quad z\in\Sigma_{0},
\end{align}
and
\begin{align}
\bar{a}(z)=a^{*}(z^{*}),\quad z\in D^{-}\cup\Sigma_{0}.
\end{align}
It is worth noting that when $Q(x,t)-Q_{+}\in L^{1}(a,+\infty)$ and $Q(x,t)-Q_{-}\in L^{1}(-\infty,a)$ hold for any constant $a\in R$, all $t>0$, the symmetry discussed above can also be extended to the branch point, so it is equally valid in region $z\in\Sigma$ and region $z\in D^{-}\cup\Sigma$.
\end{thm}
\subsection{Asymptotic behavior}
The asymptotic behaviors of eigenfunctions and scattering matrix need to be given further, which determine the establishment of an appropriate RH problem. At the same time, we can reconstruct the potential function $Q(x,t)$ according to the asymptotic behaviors of the eigenfunctions. Moreover, $k\rightarrow\infty$ in the $k$-plane, which corresponds to $z\rightarrow\infty$ in the $S_{1}$ plane and $z\rightarrow0$ in the $S_{2}$ plane.  The asymptotic formula of eigenfunctions can be derived by using the Wentzel-Kramers-Brillouin (WKB) expansion. Note that the function $\mathcal {H}=\varphi e^{i\theta(xx,t;z)\underline{\sigma}_{3}}$ such that
\begin{align}
\partial_{x}\mathcal {H}=\left(-ik\underline{\sigma}_{3}+\underline{Q}\right)\mathcal {H}
+i\lambda\mathcal {H}\underline{\sigma}_{3},
\end{align}
where $\varphi$ is the solution of scattering problem for the Lax pair. Recalling the definition of the  modified eigenfunctions Eq.\eqref{T15}, we get the following proposition, i.e.,
\begin{prop}
The modified eigenfunctions satisfy the following equations
\addtocounter{equation}{1}
\begin{align}
&\partial_{x}\mathcal{\hat{W}}^{up}(x,t;z)=-\frac{i\epsilon k_{0}^{2}}{z}\mathcal{\hat{W}}^{up}(x,t;z)+Q\mathcal{\hat{W}}^{dn}(x,t;z),
\tag{\theequation a}\label{65a}\\
&\partial_{x}\mathcal{\hat{W}}^{dn}(x,t;z)=\epsilon Q^{\dagger}
\mathcal{\hat{W}}^{up}(x,t;z)+iz\mathcal{\hat{W}}^{dn}(x,t;z),
\tag{\theequation b}\label{65b}\\
&\partial_{x}\mathcal{\bar{W}}^{up}(x,t;z)=-iz
\mathcal{\bar{W}}^{up}(x,t;z)+Q\mathcal{\bar{W}}^{dn}(x,t;z),
\tag{\theequation c}\label{65c}\\
&\partial_{x}\mathcal{\bar{W}}^{dn}(x,t;z)=\epsilon Q^{\dagger}
\mathcal{\bar{W}}^{up}(x,t;z)+i\frac{i\epsilon k_{0}^{2}}{z}\mathcal{\bar{W}}^{dn}(x,t;z).
\tag{\theequation d}\label{65d}
\end{align}
\end{prop}

\begin{prop}\label{2.14}
The asymptotic behavior of the modified eigenfunctions satisfy the following relations
\addtocounter{equation}{1}
\begin{align}
\mathcal{\hat{W}}(x,t;z)=\left(\begin{array}{c}
I_{m}+\frac{i\epsilon}{z}\int_{-\infty}^{x}
\left(Q(\zeta,t)Q^{\dagger}(\zeta,t)-k_{0}^{2}I_{m}\right)\, d\zeta +O(\frac{1}{z^{2}})\\
\frac{i\epsilon}{z}Q^{\dagger}(x,t)+O(\frac{1}{z^{2}})
\end{array}\right),\tag{\theequation a}\label{66a}\\
\mathcal{\bar{W}}(x,t;z)=\left(\begin{array}{c}
-\frac{i}{z}Q(x,t)+O(\frac{1}{z^{2}})\\
I_{m}-\frac{i\epsilon}{z}\int_{-\infty}^{x}
\left(Q(\zeta,t)Q^{\dagger}(\zeta,t)-k_{0}^{2}I_{m}\right)\, d\zeta +O(\frac{1}{z^{2}})
\end{array}\right),\tag{\theequation b}\label{66b}\\
\mathcal{\bar{V}}(x,t;z)=\left(\begin{array}{c}
I_{m}+\frac{i\epsilon}{z}\int^{+\infty}_{x}
\left(Q(\zeta,t)Q^{\dagger}(\zeta,t)-k_{0}^{2}I_{m}\right)\, d\zeta +O(\frac{1}{z^{2}})\\
\frac{i\epsilon}{z}Q^{\dagger}(x,t)+O(\frac{1}{z^{2}})
\end{array}\right),\tag{\theequation c}\label{66c}\\
\mathcal{\hat{V}}(x,t;z)=\left(\begin{array}{c}
-\frac{i}{z}Q(x,t)+O(\frac{1}{z^{2}})\\
I_{m}-\frac{i\epsilon}{z}\int^{+\infty}_{x}
\left(Q(\zeta,t)Q^{\dagger}(\zeta,t)-k_{0}^{2}I_{m}\right)\, d\zeta +O(\frac{1}{z^{2}})
\end{array}\right).\tag{\theequation d}\label{66d}
\end{align}
Note that the conditions Eqs.\eqref{66a} and \eqref{66d}  hold for $z\rightarrow\infty$ and $z\in D^{+}$, moreover the conditions Eqs.\eqref{66b} and \eqref{66c} hold for $z\rightarrow\infty$ and $z\in D^{-}$.
\end{prop}
\begin{proof}
Using Eq.\eqref{66a} as an example, the rest can be proved similarly. Taking the WKB expansion one has
\begin{align}\label{T39}
\begin{split}
&\mathcal{\hat{W}}^{up}(x,t;z)=I_{m}+\mathscr{U}_{1}(x,t;z)/z+h.o.t,\\
&\mathcal{\hat{W}}^{dn}(x,t;z)=\mathscr{V}_{1}(x,t;z)/z+\mathscr{V}_{2}(x,t;z)
/z^{2}+h.o.t,
\end{split}
\end{align}
where $h.o.t$ represents higher order terms, and $\mathscr{U}_{1}(x,t;z)$, $\mathscr{V}_{1}(x,t;z)$ as well as $\mathscr{V}_{2}(x,t;z)$ are the $m\times m$ matrix. Substituting Eq.\eqref{T39} into Eq.\eqref{66a} and comparing the coefficients of different powers $z$, we have
\begin{align}
&z^{-1}:\quad \partial_{x}\mathscr{U}_{1}(x,t;z)=-i\epsilon k^{2}_{0}+Q\mathscr{V}_{1}(x,t;z),\\
&z^{0}:\quad i\mathscr{V}_{1}(x,t;z)-\epsilon Q^{\dagger}I_{m}=0,
\end{align}
which in turn imply that
\begin{align}
\begin{split}
&\mathcal{\hat{W}}^{up}(x,t;z)=I_{m}+\frac{i\epsilon}{z}\int_{-\infty}^{x}
\left(Q(\zeta,t)Q^{\dagger}(\zeta,t)-k_{0}^{2}I_{m}\right)(\zeta,t)\, d\zeta +O(\frac{1}{z^{2}}),\\
&\mathcal{\hat{W}}^{dn}(x,t;z)=\frac{i\epsilon}{z}Q^{\dagger}(\zeta,t)
+O(\frac{1}{z^{2}}).
\end{split}
\end{align}
\end{proof}
 Taking  the following WKB expansions we similarly can derive the asymptotic behavior of the modified eigenfunctions as $z\rightarrow0$
\begin{align}\label{T40}
\begin{split}
&\mathcal{\hat{W}}^{up}(x,t;z)=\mathscr{U}_{0}(x,t;z)+\mathscr{U}_{0}(x,t;z)z+h.o.t,\\
&\mathcal{\hat{W}}^{dn}(x,t;z)=\mathscr{V}_{-1}(x,t;z)/z+\mathscr{V}_{0}(x,t;z)+h.o.t.
\end{split}
\end{align}
\begin{prop}\label{2.15}
The asymptotic behavior of the modified eigenfunctions as $z\rightarrow0$ such that
\addtocounter{equation}{1}
\begin{align}
&\mathcal{\hat{W}}(x,t;z)=\left(\begin{array}{c}
QQ_{-}^{\dagger}/k_{0}^{2}+O(z)\\
i\epsilon Q_{-}^{\dagger}+O(1)
\end{array}\right),\quad
\mathcal{\bar{W}}(x,t;z)=\left(\begin{array}{c}
-iQ_{-}/z+O(1)\\
Q^{\dagger} Q_{-}/k_{0}^{2}+O(z)
\end{array}\right),\tag{\theequation a}\\
&\mathcal{\bar{V}}(x,t;z)=\left(\begin{array}{c}
QQ_{+}^{\dagger}/k_{0}^{2}+O(z)\\
i\epsilon Q_{+}^{\dagger}+O(1)
\end{array}\right),\quad
\mathcal{\hat{V}}(x,t;z)=\left(\begin{array}{c}
-iQ_{+}/z+O(1)\\
Q^{\dagger} Q_{+}/k_{0}^{2}+O(z)
\end{array}\right).\tag{\theequation b}
\end{align}
\end{prop}
Finally, according to the definition of scattering matrix Eq.\eqref{T18}, we get
\begin{prop}
The asymptotic behavior of scattering matrix $S(z)$ in a suitable region can be given
\begin{align}\label{SJJ}
S(z)=\left\{\begin{aligned}
&I_{2m}+O(1/z),\qquad\qquad\qquad\quad z\rightarrow\infty,\\
&\frac{1}{k_{0}^{2}}\left(\begin{array}{cc}
Q_{+}Q_{-}^{\dagger} & 0_{m}\\
0_{m} & Q_{+}^{\dagger}Q_{-}
\end{array}\right)+O(z),\quad z\rightarrow0.\end{aligned} \right.
\end{align}
Note that for $z\rightarrow\infty$, the asymptotic behavior for the elements $a(z)$ and $\bar{a}(z)$ of the scattering matrix $S(z)$ hold respectively in the region  $Imz\geq0$ and $Imz\leq0$, as well as $z\in\Sigma$ for $b(z)$  and $\bar{b}(z)$. For the case $z\rightarrow0$, there is a similar result, i.e., the elements $a(z)$ and $\bar{a}(z)$ can be extended to $D^{+}$ and $D^{-}$ respectively, while the asymptotic behavior for $b(z)$  and $\bar{b}(z)$ hold for $z\in\Sigma$.
\end{prop}
It is easy to prove \textbf{Proposition} $2.16$ by combining \textbf{Proposition} \eqref{2.14} and \eqref{2.15} and the Eqs.\eqref{T15}, \eqref{a1}.

\subsection{Discrete spectrum and residue conditions}
The set of the discrete spectrum for the scattering problem consists of all the value $z\in C\setminus\Sigma$ satisfying that the scattering problem admits eigenfunctions in $L^{2}(R)$. Similar to Refs.\cite{PF-2018}-\cite{TianLY-2019}, these discrete spectral points are the zeros of the function $\det a(z)$ in the region $D^{+}$, and the zeros of the function $\det \bar{a}(z)$ in the region $D^{-}$. Assume that $\det a(z)$ has $\mathcal{N}$ simple zeros $z_{1},\cdots, z_{\mathcal{N}}$ in $D^{+}\cap\left\{z\in C: Imz>0\right\}$, where simple zeros $z_{j}$ represent $\det a(z_{j})=0$, but $(\det a)'(z_{j})\neq0$ for $1,\cdots,\mathcal{N}$. Note that from the Eqs.\eqref{TT27}
 and ($2.46$), we have a quartet of discrete eigenvalues
\begin{align}\label{T35}
\det a(z_{j})=0\Leftrightarrow \det \bar{a}(z_{j}^{*})=0 \Leftrightarrow \det a(\epsilon k_{0}^{2}/z_{j})=0\Leftrightarrow
\det \bar{a}(\epsilon k_{0}^{2}/z_{j}^{*})=0,
\end{align}
which imply that the discrete spectrum is defined by the quartet of discrete eigenvalues
\begin{align}\label{T36}
P=\left\{z_{j}, \epsilon k_{0}^{2}/{z_{j}^{*}},
z_{j}^{*}, \epsilon k_{0}^{2}/{z_{j}}\right\}_{j=1}^{\mathcal{N}}.
\end{align}
\begin{rem}
For the defocusing case, the self-adjointness of the scattering problem shows that the discrete spectrum is real in the $k$-plane. In Ref.\cite{Faddeev-1987}, the
authors presented that the discrete spectral points are simple. Moreover, they shown that the finite number of discrete eigenvalues belong to the circle $\{z:|z|=k_{0}\}$ in Ref.\cite{Demontis-2013}.
\end{rem}

Based on the above analysis, we can give the discrete spectrum set of focusing case and defocusing case
\addtocounter{equation}{1}
\begin{align}
&\text{focusing case}\quad\epsilon=-1: P=\left\{z_{j}, \epsilon k_{0}^{2}/{z_{j}^{*}},
z_{j}^{*}, \epsilon k_{0}^{2}/{z_{j}}\right\},\tag{\theequation a}\label{2.76a}\\
&\text{defocusing case}\quad\epsilon=1: P=\left\{\xi_{j},\xi_{j}^{*}\right\},\tag{\theequation b}
\end{align}
where $j=1,\cdots,\mathcal{N}$ and we suppose $Imz_{j}>0$ in the $D^{+}$.

Now assume that $\det a(z)$ has $\mathcal{N}$ simple zeros $z_{n}$ ($n=1,\cdots,\mathcal{N}$), i.e., $\det a(z_{n})=0$, which in turn implies that from Eq.\eqref{s11} the Jost eigenfunctions $\psi(x,t;z_{n})$ and $\phi(x,t;z_{n})$ are linearly dependent. Thus there is a nonzero constant $b_{n}$ that satisfies the following equation
\begin{align}\label{T37}
\phi(x,t;z_{n})=\psi(x,t;z_{n})b_{n},
\end{align}
where $b_{n}$ is a $m\times m$ constant matrix and the zeros $z_{n}\in D^{+}$. Similarly for the zeros $z_{n}^{*}\in D^{-}$ of $\det \bar{a}(z)$, according to the symmetry Eq.\eqref{T35} we know that the number of zeros is $\mathcal{N}$. There is also a $m\times m$ constant matrix $\bar{b}_{n}$  such that the Jost eigenfunctions   $\bar{\psi}(x,t;z_{n})$ and $\bar{\phi}(x,t;z_{n})$ are linearly dependent from       Eq.\eqref{s22} for $z_{n}^{*}\in D^{-}$, i.e.,
\begin{align}\label{T38}
\bar{\phi}(x,t;z_{n}^{*})=\bar{\psi}(x,t;z_{n}^{*})\bar{b}_{n},
\end{align}
note that Eqs.\eqref{T37} and \eqref{T38} are some stronger relationships of the Jost eigenfunctions, and we will discuss the generalized case in section 2.8.

In what follows, we derive the residue conditions required to solve the RH problem in the inverse problem. Note that from the Eqs.\eqref{T15} and \eqref{T37}, one has
$\mathcal{\hat{W}}(z)=e^{2i\theta(z)}\mathcal{\hat{V}}(z)b_{n}$. Then the residue condition can be obtained for the zeros $z_{n}\in D^{+}$
\begin{align}\label{2.79}
\mathop{Res}_{z=z_{n}}\left[\mathcal{\hat{W}}(x,t;z)a^{-1}(z)\right]=
e^{2i\theta(z_{n})}\mathcal{\hat{V}}(z_{n})C_{n},\quad
C_{n}=\frac{1}{(\det a)'(z_{n})}b_{n}\text{cof} a(z_{n}),
\end{align}
where $\text{cof} A$ is the cofactor matrix for all matrix $A\in C^{m\times m}$, i.e., $A(z)\text{cof} A(z)=\left(\text{cof}\right.$ $\left(A(z)\right)A(z)=\det A(z)I_{m}$.
Similarly for the zeros $z_{n}^{*}\in D^{-}$, using the Eqs.\eqref{T15} and \eqref{T38} yield that
\begin{align}\label{2.80}
\mathop{Res}_{z=z_{n}^{*}}\left[\mathcal{\bar{W}}(x,t;z)\bar{a}^{-1}(z)\right]=
e^{-2i\theta(z_{n}^{*})}\mathcal{\bar{V}}(z_{n}^{*})\bar{C}_{n},\quad
\bar{C}_{n}=\frac{1}{(\det\bar{a})'(z_{n}^{*})}\bar{b}_{n}\text{cof} \bar{a}(z_{n}^{*}).
\end{align}

It is worth noting that for any matrix $A\in C^{m\times m}$ we have
\begin{align*}
\det(\text{cof}A)=(\det A)^{m-1},
\end{align*}
which in turn implies
\begin{align}
\det(\text{cof}a(z))=\det a(z),
\end{align}
for the special case $m=2$. The purpose is that function $\det(\text{cof}a(z))$ and $\det a(z)$ have zeros of the same order for all $z_{n}\in D^{+}\cap P$, and function $\det\text{cof}\bar{a}(z)$ and $\det \bar{a}(z)$ have zeros of the same order for all $z_{n}^{*}\in D^{-}\cap P$.

For the simple eigenvalues $z_{n}$ and $z_{n}^{*}$, obviously we have
\begin{align}
\mathcal {G}\doteq\mathop{Res}_{z=z_{n}}a^{-1}(z)=\frac{\text{cof}a(z_{n})}
{(\det a)'(z_{n})},\quad
\mathcal {\bar{G}}\doteq\mathop{Res}_{z=z_{n}^{*}}\bar{a}^{-1}(z)=
\frac{\text{cof}\bar{a}(z_{n}^{*})}{(\det\bar{a})'(z_{n}^{*})},
\end{align}
and $\det\mathcal {G}=\det\mathcal {\bar{G}}=0$, which implies that the residue conditions are the matrix with rank $m-1$. Recalling the Eqs.\eqref{2.79} and \eqref{2.80}, one has
\begin{align}
C_{n}=b_{n}\mathcal {G},\quad\bar{C}_{n}=\bar{b}_{n}\mathcal {\bar{G}}.
\end{align}
For the simple eigenvalues, we also know that
\begin{align}
\det C_{n}=\det\bar{C}_{n}=0.
\end{align}

Further generalization, it is assumed that the function $\det a(z)$ has second order zeros at point $z_{n}$, i.e., $(\det a)''(z)\neq0$, then $\det(\text{cof}a(z))$ has zero of order $2(m-1)$ at $z_{n}$. In a neighborhood of $z_{n}$, we have
\begin{align}
a^{-1}(z)=\frac{1}{(z-z_{n})^{2}}\mathcal{G}_{2}+
\frac{1}{z-z_{n}}\mathcal{G}_{1}+\tilde{a}(z),
\end{align}
note that $\tilde{a}(z)$ is analytic at $z_{n}$, then
\begin{align}
&\mathcal{G}_{2}=\mathop{\lim}_{z\rightarrow z_{n}}(z-z_{n})^{2}a^{-1}(z)=
\frac{2}{(\det a)''(z_{n})}\text{cof}a(z_{n}),\\
&\mathcal{G}_{1}=\mathop{\lim}_{z\rightarrow z_{n}}\frac{d}{dz}
\left((z-z_{n})^{2}a^{-1}(z)\right)=\frac{2\text{cof}a'(z_{n})}{(\det a)''(z_{n})}-
\frac{2}{3}\frac{(\det a)'''(z_{n})}{((\det a)''(z_{n}))^{2}}\text{cof}a(z_{n}).
\end{align}
Similar to the analysis in Refs.\cite{Ortiz-2019,PF-2018}, we finally get for the second zero $z_{n}$ and $z_{n}^{*}$
\begin{align}
\mathop{Res}_{z=z_{n}}\left[\mathcal{\hat{W}}(x,t;z)a^{-1}(z)\right]=
e^{2i\theta(z_{n})}\mathcal{\hat{V}}(z_{n})C_{n},\quad
C_{n}=\frac{2}{(\det a)''(z_{n})}b_{n}\text{cof a}'(z_{n}),\\
\mathop{Res}_{z=z_{n}^{*}}\left[\mathcal{\bar{W}}(x,t;z)\bar{a}^{-1}(z)\right]=
e^{-2i\theta(z_{n}^{*})}\mathcal{\bar{V}}(z_{n}^{*})\bar{C}_{n},\quad
\bar{C}_{n}=\frac{2}{(\det\bar{a})''(z_{n}^{*})}\bar{b}_{n}
\text{cof}\bar{a}'(z_{n}^{*}).
\end{align}

The norming constants are given in the above analysis. Next, we will discuss the symmetry relationship between these norming constants $C_{n}$ and $\bar{C}_{n}$, i.e.,
\begin{align}
\bar{C}_{n}=\epsilon C_{n}^{\dagger},
\end{align}
which is proved later in Eq.\eqref{T46}. In addition, the third symmetry restricts the norming constants $C_{n}$ and $\bar{C}_{n}$ to satisfy the following symmetry
\begin{align}\label{CS2}
C_{n}=C_{n}^{T},\quad \bar{C}_{n}=\bar{C}_{n}^{T}.
\end{align}

For the focusing case, because it has quartet discrete eigenvalues, only two of them were discussed earlier, so the remaining two are discussed next. Similarly the Eqs.\eqref{T37} and \eqref{T38}, we introduce
\addtocounter{equation}{1}
\begin{align}
&\phi(x,t;\hat{z}_{n})=\psi(x,t;\hat{z}_{n})\hat{b},\quad
\hat{z}_{n}=\epsilon k_{0}^{2}/z_{n}^{*},\tag{\theequation a}\label{92a}\\
&\bar{\phi}(x,t;\hat{z}_{n}^{*})=\bar{\phi}(x,t;\hat{z}_{n}^{*})\hat{\bar{b}},\quad
\hat{z}_{n}^{*}=\epsilon k_{0}^{2}/z_{n},\tag{\theequation b}\label{92b}
\end{align}
where $\hat{b}$ and $\hat{\bar{b}}$ are the $m\times m$ constant matrices. Recalling the Eqs.\eqref{44a} and \eqref{92a} we have
\begin{align}\label{93}
\phi(x,t;z_{n})=\frac{i\epsilon}{z_{n}}\bar{\phi}(x,t;\hat{z}_{n}^{*})Q_{-}^{\dagger}
=\frac{i\epsilon}{z_{n}}\bar{\psi}(x,t;\hat{z}_{n}^{*})\hat{\bar{b}}_{n}Q_{-}^{\dagger},
\end{align}
in addition, using the Eqs.\eqref{T37} and \eqref{44b} one has
\begin{align}\label{94}
\phi(x,t;z_{n})=\psi(x,t;z_{n})b_{n}
=-\frac{i}{z_{n}}\bar{\psi}(x,t;\hat{z}_{n}^{*})Q_{+}b_{n}.
\end{align}
Combining with the Eqs.\eqref{44b} and \eqref{92a}, we can derive that $\epsilon\bar{b}_{n}^{-1}Q_{+}^{\dagger}=-
Q_{-}\bar{b}_{n}^{-1},$ which in turn implies that the nonzero constants $\bar{b}_{n}$ and $\hat{b}_{n}$ satisfy
\begin{align}
\hat{b}_{n}=-\frac{\epsilon}{k_{0}^{2}}Q_{+}^{\dagger}\bar{b}_{n}Q_{-}^{\dagger}.
\end{align}
Similarly from the Eqs.\eqref{93} and \eqref{94}, we obtain
\begin{align}
\hat{\bar{b}}_{n}=-\frac{\epsilon}{k_{0}^{2}}Q_{+}b_{n}Q_{-}.
\end{align}

In addition, we need to derive the residue conditions at zeros  $\bar{z}_{n}$ and $\bar{z}_{n}^{*}$. According to the specific expression of the norming constants defined by Eqs.\eqref{2.79} and \eqref{2.80}, we first introduce from the Eq.\eqref{46a}
\addtocounter{equation}{1}
\begin{align}
\text{cof}a(\epsilon k_{0}^{2}/z_{n}^{*})=\frac{1}{k_{0}^{2}}\text{cof}(Q_{-}^{\dagger})
\text{cof}\bar{a}(z_{n}^{*})\text{cof}(Q_{+}),\tag{\theequation a}\\
\text{cof}\bar{a}(\epsilon k_{0}^{2}/z_{n})=\frac{1}{k_{0}^{2}}\text{cof}(Q_{+}^{\dagger})
\text{cof}a(z_{n})\text{cof}(Q_{-}), \tag{\theequation b}
\end{align}
and then  differentiating Eq.\eqref{46a} with respect to $z$ and calculating at $z=z_{n}$ or $z=z_{n}^{*}$, one has
\addtocounter{equation}{1}
\begin{align}
(\det a)'(\epsilon k_{0}^{2}/z_{n}^{*})=-\epsilon\left(\frac{z_{n}^{*}}
{k_{0}^{2}}\right)^{2}\frac{\det Q_{+}\det Q_{-}^{\dagger}}{k_{0}^{2m}}
(\det\bar{a})'(z_{n}^{*}),\tag{\theequation a}\\
(\det\bar{a})'(\epsilon k_{0}^{2}/z_{n})=-\epsilon\left(\frac{z_{n}}
{k_{0}^{2}}\right)^{2}\frac{\det Q_{+}^{\dagger}\det Q_{-}}{k_{0}^{2m}}
(\det a)'(z_{n}^{*}). \tag{\theequation b}
\end{align}
Therefore we get
\addtocounter{equation}{1}
\begin{align}
\mathop{Res}_{z=\hat{z}_{n}=\epsilon k_{0}^{2}/z_{n}^{*}}\left[\mathcal{\hat{W}}(x,t;z)a^{-1}(z)\right]=
e^{2i\theta(\hat{z}_{n})}\mathcal{\hat{V}}(\hat{z}_{n})\hat{C}_{n},\tag{\theequation a}\\
\mathop{Res}_{z=\hat{z}_{n}^{*}=\epsilon k_{0}^{2}/z_{n}}\left[\mathcal{\bar{W}}(x,t;z)\bar{a}^{-1}(z)\right]=
e^{-2i\theta(\hat{z}_{n}^{*})}\mathcal{\bar{V}}(\hat{z}_{n}^{*})\hat{\bar{C}}_{n},\tag{\theequation b}
\end{align}
where the norming constants $\hat{C}_{n}$ and $\hat{\bar{C}}_{n}$ admit that
\begin{align}\label{CS}
\hat{C}_{n}=\frac{1}{(z_{n}^{*})^{2}}Q_{+}^{\dagger}\bar{C}_{n}Q_{+}^{\dagger},
\quad \hat{\bar{C}}_{n}=\frac{1}{z_{n}^{2}}Q_{+}C_{n}Q_{+},
\quad \hat{\bar{C}}_{n}=\epsilon \hat{C}_{n}^{\dagger}.
\end{align}

Summarizing the discussion in this section, we can get some discrete data needed in the inverse problem including discrete eigenvalues, reflection coefficients, norming constants, and the symmetry they satisfy.
\begin{prop}
The discrete eigenvalues defined by Eq.\eqref{T36} are  given
\addtocounter{equation}{1}
\begin{align}
&\text{focusing case}\quad\epsilon=-1: P=\left\{z_{j}, \epsilon k_{0}^{2}/{z_{j}^{*}},
z_{j}^{*}, \epsilon k_{0}^{2}/{z_{j}}\right\},\tag{\theequation a}\\
&\text{defocusing case}\quad\epsilon=1: P=\left\{\xi_{j},\xi_{j}^{*}\right\}.\tag{\theequation b}
\end{align}
The residue conditions and the norming constants such that
\addtocounter{equation}{1}
\begin{align*}
&\mathop{Res}_{z=z_{n}}\left[\mathcal{\hat{W}}(x,t;z)a^{-1}(z)\right]=
e^{2i\theta(x,t;z_{n})}\mathcal{\hat{V}}(x,t;z_{n})C_{n},
\tag{\theequation a}\label{102a}\\
&\mathop{Res}_{z=z_{n}^{*}}\left[\mathcal{\bar{W}}(x,t;z)\bar{a}^{-1}(z)\right]=
e^{-2i\theta(x,t;z_{n}^{*})}\mathcal{\bar{V}}(x,t;z_{n}^{*})\bar{C}_{n},
\tag{\theequation b}\label{102b}\\
&\mathop{Res}_{z=\hat{z}_{n}=\epsilon k_{0}^{2}/z_{n}^{*}}\left[\mathcal{\hat{W}}(x,t;z)a^{-1}(z)\right]=
e^{2i\theta(x,t;\hat{z}_{n})}\mathcal{\hat{V}}(x,t;\hat{z}_{n})\hat{C}_{n},
\quad\hat{C}_{n}=\frac{1}{(z_{n}^{*})^{2}}Q_{+}^{\dagger}\bar{C}_{n}Q_{+}^{\dagger},
\tag{\theequation c}\label{102c}\\
&\mathop{Res}_{z=\hat{z}_{n}^{*}=\epsilon k_{0}^{2}/z_{n}}\left[\mathcal{\bar{W}}(x,t;z)\bar{a}^{-1}(z)\right]=
e^{-2i\theta(x,t;\hat{z}_{n}^{*})}\mathcal{\bar{V}}(x,t;\hat{z}_{n}^{*})
\hat{\bar{C}}_{n},\quad \hat{\bar{C}}_{n}=\frac{1}{z_{n}^{2}}Q_{+}C_{n}Q_{+}.\tag{\theequation d}\label{102d}
\end{align*}
It is worth noting that for the defocusing case, Eqs.\eqref{102a} and \eqref{102b} are not considered. In general, when the discrete eigenvalues are simple zeros, for the specific case $m = 2$, which means that the norming constants are a matrix with rank one.
\end{prop}
\subsection{The norming constants}
Based on the strategy in the Ref.\cite{PF-2018}, we discuss the relationships between the norming constants and the rank of the following $2m\times2m$ matrices
\begin{align}\label{t1}
\mathscr{P}(x,t;z)=(\phi(x,t;z), \psi(x,t;z)),\quad
\mathscr{\bar{P}}(x,t;z)=(\bar{\psi}(x,t;z), \bar{\phi}(x,t;z)),
\end{align}
which are analytic in the region $D^{+}$ and $D^{-}$, respectively. Now considering that $z_{n}\in D^{+}$ is the simple zero of the $\det a(z)$, we obviously have the $\det \bar{a}(z^{*}_{n})=0$ with  $(\det\bar{a})'(z^{*}_{n})\neq0$ from the fist symmetry Eq.\eqref{TT27}. Then assume that $\Xi_{n}\in \textbf{C}^{2m}\setminus\{0\}$ is a right null vector of $\mathscr{P}(x,t;z_{n})$ for the arbitrary constant $m$, that is $\Xi_{n}\in\text{ker}\mathscr{P}(x,t;z_{n})$, and introducing the vector
\begin{align}\label{t2}
\Xi_{n}=\left(\begin{array}{c}
\Xi_{n}^{up}\\
\Xi_{n}^{dn}
\end{array}\right),\quad \Xi_{n}^{up},\Xi_{n}^{dn}\in \textbf{C}^{m},
\end{align}
obviously we have from  Eq.\eqref{t1} of $\mathscr{P}(x,t;z_{n})$
\begin{align}\label{t3}
\phi(x,t;z_{n})\Xi_{n}^{up}+\psi(x,t;z_{n})\Xi_{n}^{dn}=0_{2m\times m},
\end{align}
which denotes that there are two non-zero vectors $\eta_{n}$ and $\xi_{n}$ satisfying
\begin{align}\label{t4}
\phi(x,t;z_{n})\eta_{n}=\psi(x,t;z_{n})\xi_{n}, \quad z\in D^{+},
\end{align}
with $\eta_{n}=\Xi_{n}^{up}$ and $\xi_{n}=-\Xi_{n}^{dn}$. It is worth noting that the reason why $\eta_{n}, \xi_{n}\neq0$ is that the first $m$ column and the last $m$ column of matrix $\mathscr{P}(x,t;z_{n})$ are linearly independent. Vice versa, for the nonzero vector $\eta_{n}$ and $\xi_{n}$ defined by   Eq.\eqref{t3}, then the $2m\times1$ vector $\Xi_{n}=(\eta_{n}, -\xi_{n})^{T}$ belongs to ker$\mathscr{P}(x,t;z_{n})$. For the simple zero $z_{n}^{*}\in D^{-}$ of the $\det\bar{a}(z)$, we have similar results for the function $\bar{\mathscr{P}}(x,t;z_{n}^{*})$. In addition, for $\eta_{n}, \xi_{n}\in C^{m}\setminus\{0\}$, we know that $\Xi_{n}=(\eta_{n}, -\xi_{n})^{T}$ is the kernel of $T(z_{n})=\bar{\mathscr{P}}^{\dagger}(x,t;z_{n}^{*})\mathcal {L}_{\epsilon}\mathscr{P}(x,t;z_{n})$, which yields that
\begin{align}\label{t5}
a(z_{n})\eta_{n}=0_{m\times1},\quad
\bar{a}^{\dagger}(z_{n}^{*})\xi_{n}=0_{m\times1}.
\end{align}
From  Eq.\eqref{t5}, we clearly know that $\eta_{n}$ and $\xi_{n}$ are the right null vectors of the  $a(z_{n})$ and $\bar{a}^{\dagger}(z_{n}^{*})$, respectively. In the same way, the ker$\bar{\mathscr{P}}(x,t;z_{n}^{*})$ can be obtained, i.e., the
nonzero vector $\bar{\Xi}_{n}=(\bar{\eta}_{n}, -\bar{\xi}_{n})^{T}$, moreover
\begin{align}
a^{\dagger}(z_{n})\bar{\xi}_{n}=0_{m\times1},\quad
\bar{a}(z_{n}^{*})\bar{\eta}_{n}=0_{m\times1}.
\end{align}
Clearly the nonzero vectors $\bar{\xi}_{n}$ and $\bar{\eta}_{n}$ are the kernel of the $a^{\dagger}(z_{n})$ and $\bar{a}(z_{n}^{*})$.

Note that for the arbitrary $n\times n$ matrix $A$, then $A(z)\text{cof} A(z)=\left(\text{cof} A(z)\right)A(z)=\det A(z)I_{n}$, where cof$A$ is defined as before. It follows that for $m=2$,
$\det a(z)=\det \text{cof}a(z),$
which implies that they have a zero of the same order for the eigenvalue $z_{n}$. Note that the matrices $a(z)$ and $(\text{cof}a(z_{n}))$ are $2\times2$ matrices, we have rank$a(z_{n})$=rank$(\text{cof}$ $a(z_{n}))$, and rank$a(z_{n})=1$ $\Leftrightarrow $ rank$\mathscr{P}(x,t;z_{n})=3$.
In addition, the expressions $a(z_{n})(\text{cof}$ $a(z_{n}))=\left(\text{cof} a(z_{n})\right)a(z_{n})=0_{4\times4}$ and $\bar{a}(z_{n}^{*})\text{cof} \bar{a}(z_{n}^{*})=\left(\text{cof} \bar{a}(z_{n}^{*})\right)\bar{a}(z_{n}^{*})=0_{4\times4}$ denote that the each column of the $\text{cof} a(z_{n})$ is the left and right null vectors of $a(z_{n})$ and each column of the $\text{cof} \bar{a}(z_{n}^{*})$ is the left and right null vectors of $\bar{a}(z_{n}^{*})$. Thus choosing the special case  one has
\begin{align}
0_{4\times2}=\mathscr{P}(x,t;z_{n})\left(\begin{array}{c}
(\text{cof}a(z_{n}))\\
-C_{n}
\end{array}\right)\Leftrightarrow \phi(x,t;z_{n})(\text{cof}a(z_{n}))=
\psi(x,t;z_{n})C_{n}.
\end{align}
Similarly, we have
\begin{align}
0_{4\times2}=\mathscr{\bar{P}}(x,t;z_{n}^{*})\left(\begin{array}{c}
-\bar{C}_{n}\\
(\text{cof}a(z_{n}^{*}))
\end{array}\right)\Leftrightarrow \bar{\phi}(x,t;z_{n}^{*})(\text{cof}a(z_{n}^{*}))=
\bar{\psi}(x,t;z_{n}^{*})\bar{C}_{n}.
\end{align}

\section{Inverse scattering problem with NZBCs}

 The analytical, symmetry, and asymptotic properties of the Jost functions and the scattering data are obtained based on the above analysis. Therefore we can build a generalized Riemann-Hilbert problem for the modified matrix KdV equation with non-zero boundary value conditions.
\begin{prop}
Introducing the sectionally meromorphic matrices
\begin{align}\label{Matr}
\mathcal{Q}(x,t;z)=\left\{\begin{aligned}
&\mathcal{Q}^{+}(x,t;z)=\left(\mathcal{\hat{W}}(x,t;z)a^{-1}(z),
\mathcal{\hat{V}}(x,t;z)\right), \quad z\in D^{+},\\
&\mathcal{Q}^{-}(x,t;z)=\left(\mathcal{\bar{V}}(x,t;z),
\mathcal{\bar{W}}(x,t;z)\bar{a}^{-1}(z)\right), \quad z\in D^{-},
\end{aligned}\right.
\end{align}
then we can get the generalized Riemann-Hilbert problem:\\
$\bullet$ Analyticity: $\mathcal{Q}(x,t;z)$ is an analytic function in $C\setminus\Sigma$.\\
$\bullet$ Jump condition:
\begin{align}\label{Jump}
\mathcal{Q}^{-}(x,t;z)=\mathcal{Q}^{+}(x,t;z)(I_{2m}-G(x,t;z)),\quad z\in\Sigma,
\end{align}
where the jump matrix and the counter are presented in the Figure 2
\begin{align}\label{JM}
G(x,t;z)=\left(\begin{array}{cc}
0_{m} & -e^{-2i\theta(x,t;z)\bar{\rho}(z)}\\
e^{2i\theta(x,t;z)\rho(z)} & 0_{m}
\end{array}\right),
\end{align}

\centerline{\begin{tikzpicture}[scale=0.55]
\path (1,4) -- (9,4) to (9,-4) -- (1,-4);
\path [fill=gray] (1,4) -- (9,4) to
(9,0) -- (1,0);
\path [fill=gray] (-1,4) -- (-9,4) to
(-9,0) -- (-1,0);
\filldraw[white, line width=0.5](3,0)--(7,0) arc (0:180:2);
\filldraw[gray, line width=0.5](1,0)--(3,0) arc (-180:0:2);
\draw[fill] (5,0)node[below left]{$0$} circle [radius=0.07];
\draw[fill] (-5,0)node[below left]{$0$} circle [radius=0.07];
\draw[fill] (3,0)node[below]{} circle [radius=0.07];
\draw[fill] (7,0)node[below]{} circle [radius=0.07];
\draw[fill] (5,2)node[below]{} circle [radius=0.07];
\draw[fill] (5,-2)node[below]{} circle [radius=0.07];
\draw[fill][red] (-3.586,1.414)node[above right]{$\zeta_{n}$} circle [radius=0.07];
\draw[fill][blue] (-3.586,-1.414)node[below right]{$\zeta_{n}^{*}$} circle [radius=0.07];
\draw[fill] (-7,0)node[below left]{$-k_{0}$} circle [radius=0.07];
\draw[fill] (-3,0)node[below right]{$k_{0}$} circle [radius=0.07];
\draw[->][thick](-9,0)--(-8,0);
\draw[-][thick](-8,0)--(-7,0);
\draw[->][thick](-7,0)--(-6,0);
\draw[-][thick](-6,0)--(-5,0);
\draw[->][thick](-5,0)--(-4,0);
\draw[-][thick](-4,0)--(-3,0);
\draw[-][thick](-3,0)--(-2,0);
\draw[->][thick](-2,0)--(-1,0)node[above]{\footnotesize$Rez$};
\draw[->][thick](1,0)--(2,0);
\draw[-][thick](2,0)--(3,0)node[above left]{\footnotesize$-k_{0}$};
\draw[-][thick](3,0)--(4,0);
\draw[<-][thick](4,0)--(5,0);
\draw[-][thick](5,0)--(6,0);
\draw[<-][thick](6,0)--(7,0)node[above right]{\footnotesize$k_{0}$};
\draw[->][thick](7,0)--(8,0);
\draw[->][thick](8,0)--(9,0)node[above]{\footnotesize$Rez$};
\draw[-][thick](5,0)--(5,1);
\draw[-][thick](5,1)--(5,2)node[above right]{\footnotesize$ik_{0}$};
\draw[-][thick](5,2)--(5,3);
\draw[->][thick](5,3)--(5,4)node[right]{\footnotesize$Imz$};
\draw[-][thick](5,0)--(5,-1);
\draw[-][thick](5,-1)--(5,-2)node[below right]{\footnotesize$-ik_{0}$};
\draw[-][thick](5,-2)--(5,-3);
\draw[-][thick](5,-3)--(5,-4);
\draw[-][thick](-5,0)--(-5,1);
\draw[-][thick](-5,1)--(-5,2);
\draw[-][thick](-5,2)--(-5,3);
\draw[->][thick](-5,3)--(-5,4)node[right]{\footnotesize$Imz$};
\draw[-][thick](-5,0)--(-5,-1);
\draw[-][thick](-5,-1)--(-5,-2);
\draw[-][thick](-5,-2)--(-5,-3);
\draw[-][thick](-5,-3)--(-5,-4);
\draw[->][red, line width=0.8] (7,0) arc(0:220:2);
\draw[->][red, line width=0.8] (7,0) arc(0:330:2);
\draw[->][red, line width=0.8] (7,0) arc(0:-330:2);
\draw[->][red, line width=0.8] (7,0) arc(0:-220:2);
\draw[-][line width=0.8] (-3,0) arc(0:220:2);
\draw[-][line width=0.8] (-3,0) arc(0:330:2);
\draw[-][line width=0.8] (-3,0) arc(0:-330:2);
\draw[-][line width=0.8] (-3,0) arc(0:-220:2);
\draw[fill][blue] (4,1) circle [radius=0.035][thick]node[right]{\footnotesize$-k_{0}^{2}/z_{n}$};
\draw[fill][red] (4,-1) circle [radius=0.035][thick]node[right]{\footnotesize$-k_{0}^{2}/z_{n}^{*}$};
\draw[fill][red] (8,3) circle [radius=0.035][thick]node[below]{\footnotesize$z_{n}$};
\draw[fill][blue] (8,-3) circle [radius=0.035][thick]node[above]{\footnotesize$z_{n}^{*}$};
\end{tikzpicture}}

\noindent {\small \textbf{Figure 2.} Left/Right: Based on the relationships Eq.($2.101$), the distribution of discrete eigenvalues for the defocusing and focusing  cases can be shown in the complex $z$-plane. As shown in this figure, the red counter is the jump conditions about the Riemann-Hilbert. }

$\bullet$ Asymptotic behaviors:
\begin{align}\label{SJ}
\left\{\begin{aligned}
&\mathcal{Q}^{\pm}(x,t;z)=I_{2m}+O(1/z),\qquad\quad z\rightarrow\infty,\\
&\mathcal{Q}^{\pm}(x,t;z)=-(i/z)\underline{\sigma}_{3}\underline{Q}_{+}+Q(1),\quad
z\rightarrow0.
\end{aligned}\right.
\end{align}
\end{prop}
\begin{proof}
From Eq.($2.22$), one has
\addtocounter{equation}{1}
\begin{align}
&\mathcal{\bar{V}}(x,t;z)=\mathcal{\hat{W}}(x,t;z)a^{-1}(z)-e^{2i\theta(x,t;z)}
\mathcal{\hat{V}}(x,t;z)\rho(z),\tag{\theequation a}\label{3.5a}\\
&\mathcal{\bar{W}}(x,t;z)\bar{a}^{-1}(z)=\mathcal{\hat{V}}(x,t;z)+e^{-2i\theta(x,t;z)}
\mathcal{\bar{V}}(x,t;z)\bar{\rho}(z).\tag{\theequation b}\label{3.5b}
\end{align}
Inserting Eq.\eqref{3.5a} into Eq.\eqref{3.5b} yields that
\begin{align}\label{T41}
\mathcal{\bar{W}}(x,t;z)\bar{a}^{-1}(z)=(I_{m}-\rho(z)\bar{\rho}(z))
\mathcal{\hat{V}}(x,t;z)+e^{-2i\theta(x,t;z)}\mathcal{\hat{W}}(x,t;z)
a^{-1}(z)\bar{\rho}(z)).
\end{align}
The Eqs.\eqref{3.5b} and \eqref{T41}  are written in matrix form
\begin{align}
\left(\mathcal{\bar{V}},
\mathcal{\bar{W}}\bar{a}^{-1}(z)\right)=\left(\mathcal{\hat{W}}a^{-1}(z),
\mathcal{\hat{V}}\right)\left(\begin{array}{cc}
I_{m} & -e^{-2i\theta(x,t;z)}\bar{\rho}(z)\\
e^{2i\theta(x,t;z)}\rho(z) & I_{m}-\rho(z)\bar{\rho}(z)
\end{array}\right),
\end{align}
which  leads to the jump condition  \eqref{Jump}.
Based on the asymptotic behavior of the modified eigenfunctions and scattering matrix, the condition Eq.\eqref{SJ} can be derived.
\end{proof}
\subsection{Reconstruction formula}
In this section, we will solve the RH  problem \eqref{Jump}, and then use the obtained solution to establish a relationship with the modified eigenfunctions, that is, obtain the expression of the modified eigenfunctions. Finally, we will combine the asymptotic behavior of the modified eigenfunctions to restore the potential $Q(x,t)$, the solution of the modified KdV equation.
\begin{lem}\label{lem3.2}
Supposing that the functions $f_{\pm}(z)$ are the analytic functions in the upper and lower half planes of $z$, and as $|z|\rightarrow\infty$, then $f_{\pm}(z)\rightarrow0$.
Considering the Cauchy projection operator
\begin{align}
P^{\pm}f(z)=\frac{1}{2\pi i}\int_{\Sigma}\frac{f(\zeta)}{\zeta-(z\pm i0)}d\zeta,
\end{align}
one has
\begin{align}
P^{\pm}(f_{\pm}(z))=\pm f(z),\quad
P^{\pm}(f_{\mp}(z))=0.
\end{align}
where $\int_{\Sigma}$ represents the integral along the counters shown in Fig. 2.
\end{lem}
\begin{prop}
The solution of the RH  problem \eqref{Jump} under the simple zeros case for the focusing modified KdV equation can be written as
\begin{align}\label{RHP-jie}
\mathcal{Q}(x,t;z)=I_{2m}-\frac{i}{z}\underline{\sigma}_{3}\underline{Q}+
\sum_{n=1}^{2\mathcal{N}}\frac{\mathop{Res}\limits_{z=\zeta_{n}}
\mathcal{Q}^{+}}{z-\zeta_{n}}+
\sum_{n=1}^{2\mathcal{N}}\frac{\mathop{Res}\limits_{z=\zeta_{n}^{*}}
\mathcal{Q}^{-}}{z-\zeta_{n}^{*}}+
\frac{1}{2\pi i}\int_{\Sigma}\frac{(\mathcal{Q}^{+}G)(x,t;\xi)}
{\xi-z}d\xi,
\end{align}
where $\zeta_{n}=z_{n}$ and $\zeta_{n+\mathcal {N}}=\epsilon k_{0}^{2}/z_{n}^{*}$ for $n=1,2,\cdots,\mathcal {N}$.
\end{prop}
\begin{proof}
Subtracting out the asymptotic behaviors and the pole contributions, we can get
\begin{align}\label{T48}
\begin{split}
&\mathcal{Q}^{-}-I_{2m}+(i/z)\underline{\sigma}_{3}\underline{Q}_{+}-
\sum_{n=1}^{2\mathcal{N}}\frac{\mathop{Res}\limits_{z=\zeta_{n}^{*}}
\mathcal{Q}^{-}}{z-\zeta_{n}^{*}}-
\sum_{n=1}^{2\mathcal{N}}\frac{\mathop{Res}\limits_{z=\zeta_{n}}
\mathcal{Q}^{+}}{z-\zeta_{n}}\\&=\mathcal{Q}^{+}-
I_{2m}+(i/z)\underline{\sigma}_{3}\underline{Q}_{+}-
\sum_{n=1}^{2\mathcal{N}}\frac{\mathop{Res}\limits_{z=\zeta_{n}}
\mathcal{Q}^{+}}{z-\zeta_{n}}-
\sum_{n=1}^{2\mathcal{N}}\frac{\mathop{Res}\limits_{z=\zeta_{n}^{*}}
\mathcal{Q}^{-}}{z-\zeta_{n}^{*}}-\mathcal{Q}^{+}G,
\end{split}
\end{align}
note that the first four terms on the left side of the Eq.\eqref{T48} are analytic in the region $D^{-}$, and the zero point of the fifth term is $\zeta_{n}\in D^{+}$, so the left side of  Eq.\eqref{T48} is analytic in $D^{-}$, and is $O(1/z)$ as $z\rightarrow\infty$. Similarly the first four terms on the right side of  Eq.\eqref{T48} are analytic in the region $D^{+}$, and is $O(1/z)$ as $z\rightarrow\infty$. Combining with the \textbf{Lemma} \eqref{lem3.2}, and applying
the projection operator $P^{-}$ and $P^{+}$ to both ends of the Eq.\eqref{T48}, we
can get the solution Eq.\eqref{RHP-jie}.
\end{proof}

According to the above analysis, the formal solution of RH problem is obtained. Combining with the definition of the meromorphic function $\mathcal {Q}$, we can further recover the modified eigenfunctions by using the asymptotic property of the modified eigenfunctions. Finally, we can get the formal solution of the modified KdV equation \eqref{T1}. Also note that Eq.\eqref{RHP-jie} depends on the residue conditions at zeros $z=z_{n}$ and $z=\epsilon k_{0}^{2}/z_{n}^{*}$ in $D^{+}$, as well as $z=z_{n}^{*}$ and $z=\epsilon k_{0}^{2}/z_{n}$ in $D^{-}$. Therefore using  Eqs.\eqref{2.79} and \eqref{2.80} yield the following relationships
\addtocounter{equation}{1}
\begin{align}
&\mathop{Res}_{z=\zeta_{n}}\mathcal {Q}^{+}=\left(
e^{2i\theta(z_{n})}\mathcal{\hat{V}}(z_{n})C_{n},\quad 0_{2m\times m}\right),\quad
n=1,2,\cdots,2\mathcal {N}, \tag{\theequation a}\label{3.12a}\\
&\mathop{Res}_{z=\zeta_{n}^{*}}\mathcal {Q}^{-}=\left( 0_{2m\times m},\quad e^{-2i\theta(z_{n}^{*})}\mathcal{\bar{V}}(z_{n}^{*})\bar{C}_{n}\right),\quad
n=1,2,\cdots,2\mathcal {N}.\tag{\theequation b}\label{3.12b}
\end{align}
Substituting  Eqs.\eqref{3.12a} and \eqref{3.12b} into the solution Eq.\eqref{RHP-jie} and evaluating the last $m$ columns of the Eq.\eqref{RHP-jie} at the points $z=z_{n}$ and $z=\epsilon k_{0}^{2}/z_{n}^{*}$, we can get the following equation by comparing with the last $m$ columns of $\mathcal {Q}$
\begin{align}\label{jie-1}
\begin{split}
\mathcal{\hat{V}}(x,t;\zeta_{n})=\left(\begin{array}{c}
                       -iQ_{+}/\zeta_{n}\\
                        I_{m}
                     \end{array}\right)
+\sum_{j=1}^{2\mathcal {N}}
\frac{e^{-2i\theta(x,t;\zeta_{j}^{*})}}{\zeta_{n}-\zeta_{j}^{*}}
\mathcal{\bar{V}}(x,t;\zeta_{j}^{*})\bar{C}_{j}
+\frac{1}{2\pi i}\int_{\Sigma}\frac{(\mathcal{Q}^{+}G)_{2}(\zeta)}{\zeta-\zeta_{n}}\,d\zeta,
\end{split}
\end{align}
where $\bar{C}_{i}=\hat{\bar{C}}_{i}$, for $i=\mathcal {N}+1, \mathcal {N}+2,\cdots,2\mathcal {N}$.
Similarly the first $m$ columns of Eq.\eqref{RHP-jie} at the points $z=z_{n}^{*}$ and $z=\epsilon k_{0}^{2}/z_{n}$ can be obtained by
\begin{align}\label{jie-2}
\begin{split}
\mathcal{\bar{V}}(x,t;\zeta_{n}^{*})=\left(\begin{array}{c}
                      I_{m}\\
                      i\epsilon Q_{+}^{\dagger}/\zeta_{n}^{*}
                     \end{array}\right)
+\sum_{j=1}^{2\mathcal {N}}
\frac{e^{2i\theta(x,t;\zeta_{j})}}{\zeta_{n}^{*}-\zeta_{j}}
\mathcal{\hat{V}}(x,t;\zeta_{j})C_{j}
+\frac{1}{2\pi i}\int_{\Sigma}\frac{(\mathcal{Q}^{+}G)_{1}(\zeta)}{\zeta-\zeta_{n}^{*}}\,d\zeta,
\end{split}
\end{align}
where $n=1,2,\cdots,\mathcal {N}$, $(\mathcal{Q}^{+}G)_{j}$ $(j=1,2)$ denote the first and last $m$ columns of the product respectively,  and $C_{i}=\hat{C}_{i}$, for $i=\mathcal {N}+1, \mathcal {N}+2,\cdots,2\mathcal {N}$.
\begin{prop}
The potential $Q(x,t)$ can be determined by
\begin{align}\label{T42}
Q(x,t)=Q_{+}+i\sum_{n=1}^{2\mathcal {N}}e^{-2i\theta(x,t;\zeta_{n}^{*})}
\mathcal{\bar{V}}^{up}(x,t;\zeta_{n}^{*})\bar{C}_{n}
+\frac{1}{2\pi }\int_{\Sigma}e^{-2i\theta(x,t;\zeta_{j})}
\mathcal{\bar{V}}^{up}(x,t;\zeta_{n})\bar{\rho}(\zeta)\,d\zeta.
\end{align}
\end{prop}
\begin{proof}
Taking $\mathcal{Q}=\mathcal{Q}^{+}$ and combining with the expression $\mathcal{Q}^{+}(x,t;z)=\left(\mathcal{\hat{W}}(x,t;z)a^{-1}(z),\right.$
$\left.\mathcal{\hat{V}}(x,t;z)\right)$, one has
\begin{align}
\mathcal{Q}^{+}_{2}=\mathcal{\hat{V}}=\left(\begin{array}{c}
0_{m}\\
I_{m}\end{array}\right)+
\frac{1}{z}\left\{\left(\begin{array}{c}
-iQ_{+}\\
0_{m}\end{array}\right)+\sum_{n=1}^{2\mathcal {N}}e^{-2i\theta(x,t;\zeta_{n}^{*})}
\mathcal{\bar{V}}^{up}(x,t;\zeta_{n}^{*})\bar{C}_{n}
-\frac{1}{2\pi i}\int_{\Sigma}(\mathcal{Q}^{+}G)_{2}(\zeta)\,d\zeta\right\}
\end{align}
as $z\rightarrow\infty$, and comparing with the asymptotic behavior Eq.\eqref{66d} yields
\begin{align}\label{T43}
Q=Q_{+}+i\sum_{n=1}^{2\mathcal {N}}e^{-2i\theta(x,t;\zeta_{n}^{*})}
\mathcal{\bar{V}}^{up}(x,t;\zeta_{n}^{*})\bar{C}_{n}-\frac{1}{2\pi i}\int_{\Sigma}(\mathcal{Q}^{+}G)_{2}(\zeta)\,d\zeta.
\end{align}
Note from   Eqs.\eqref{2.22a}, \eqref{Matr} and the jump matrix \eqref{JM}, one has
\begin{align}\label{T44}
\begin{split}
(\mathcal{Q}^{+}G)_{2}(\zeta)&=-e^{-2i\theta(x,t;\zeta)}
\mathcal{\hat{W}}(x,t;\zeta)a^{-1}(\zeta)\bar{\rho}(\zeta)+\mathcal{\hat{V}}(x,t;\zeta)
\rho(\zeta)\bar{\rho}(\zeta)\\&=-e^{-2i\theta(x,t;\zeta)}
\mathcal{\bar{V}}(x,t;\zeta)\bar{\rho}(\zeta).
\end{split}
\end{align}
Inserting   Eq.\eqref{T44} into Eq.\eqref{T43}, one has  the potential $Q(x,t)$ Eq.\eqref{T42}.

Taking a similar approach to $\mathcal{Q}=\mathcal{Q}^{-}$, we get
\begin{align}\label{T45}
Q^{\dagger}(x,t)=Q_{+}^{\dagger}-i\epsilon\sum_{n=1}^{2\mathcal {N}}e^{2i\theta(x,t;\zeta_{n})}
\mathcal{\hat{V}}^{dn}(x,t;\zeta_{n})C_{n}
+\frac{\epsilon}{2\pi }\int_{\Sigma}e^{2i\theta(x,t;\zeta)}
\mathcal{\hat{V}}^{dn}(x,t;\zeta)\rho(\zeta)\,d\zeta,
\end{align}
where $(\mathcal{Q}G)_{1}(\zeta)=e^{2i\theta(x,t;\zeta)}
\mathcal{\hat{V}}^{dn}(x,t;\zeta)\rho(\zeta)$.
\end{proof}

According to the construction of the above potential functions $Q(x,t)$ and $Q^{\dagger}(x,t)$, and considering the asymptotic behavior of the eigenfunctions Eqs.\eqref{66c} and \eqref{66d}, i.e., $\mathcal {\bar{V}}^{up}\sim I_{m}$ and $\mathcal {\hat{V}}^{dn}\sim I_{m}$ as $x\rightarrow\infty$ for $z\in D^{-}$ and $z\in D^{+}$, we get the result by comparing the Hermite conjugation of the Eq.\eqref{T42} with Eq.\eqref{T45}
\begin{align}\label{T46}
\bar{C}_{n}=\epsilon\bar{C}_{n}^{\dagger}, \quad n=1,2,\cdots,2\mathcal {N}.
\end{align}
In addition, the symmetry of potential function also requires that the norming constants satisfy the symmetries
\begin{align}
C_{n}^{T}=C_{n},\quad \bar{C}_{n}^{T}=\bar{C}_{n},\quad n=1,2,\cdots,2\mathcal {N}.
\end{align}
\subsection{The formal solution of focusing KdV with Reflection-less potential}
In this section we will consider soliton solutions of the focusing modified KdV equation with reflection-less potential, i.e., the reflection coefficients $\rho(z)=0$ and $\bar{\rho}(z)=0$ for $z\in\Sigma$. Note for $\rho(z)=\bar{\rho}(z)=0$, from the jump condition Eq.\eqref{Jump} one has $\mathcal {Q}^{+}=\mathcal {Q}^{-}$, and then the inverse problem can be transformed to a closed algebraic system. Finally the soliton solutions of the modified KdV equation \eqref{T1} can be derived.
\begin{prop}
The $N$-soliton solutions of the focusing modified KdV equation with reflection-less potential can be written as
\begin{align}\label{Qjie}
Q(x,t)=Q_{+}+i\sum_{n=1}^{2\mathcal {N}}e^{-2i\theta(x,t;z_{n}^{*})}X_{n}\bar{C}_{n},
\end{align}
where $X_{n}$ is determined by   Eq.\eqref{T51}.
\end{prop}
\begin{proof}
Recall that the discrete eigenvalues for the focusing case  \eqref{2.76a}, with $\zeta_{\mathcal {N}+j}=-k_{0}^{2}/z_{j}^{*}$ and the symmetries about the norming constants  \eqref{CS2} and \eqref{CS}
\begin{align}
\bar{C}_{n}=-C^{\dagger}_{n},\quad
C_{\mathcal {N}+j}=\frac{Q_{+}^{\dagger}\bar{C}_{j}Q_{+}^{\dagger}}{(z_{j}^{*})^{2}},
\end{align}
for $j=1,2,\cdots,\mathcal {N}$ and $n=1,2,\cdots,2\mathcal {N}$. For simplicity, taking the notation
\begin{align}\label{notation}
c(x,t;z)=\frac{C_{j}}{z-\zeta_{j}}e^{2i\theta(x,t;\zeta_{j})},\quad
j=1,2,\cdots,2\mathcal {N}.
\end{align}
Note that from the Eq.\eqref{T42}, we need to derive the expression of the modified eigenfunctions. Using the Eqs.\eqref{jie-1} and \eqref{jie-2} yield that
\addtocounter{equation}{1}
\begin{align}
&\mathcal{\hat{V}}(x,t;\zeta_{n})=\left(\begin{array}{c}
                       -iQ_{+}/\zeta_{n}\\
                        I_{m}
                     \end{array}\right)
+\sum_{j=1}^{2\mathcal {N}}
\frac{e^{-2i\theta(x,t;\zeta_{j}^{*})}}{\zeta_{n}-\zeta_{j}^{*}}
\mathcal{\bar{V}}(x,t;\zeta_{j}^{*})\bar{C}_{j},\tag{\theequation a}\\
&\mathcal{\bar{V}}(x,t;\zeta_{n})=\left(\begin{array}{c}
                      I_{m}\\
                      i\epsilon Q_{+}^{\dagger}/\zeta_{n}^{*}
                     \end{array}\right)
+\sum_{j=1}^{2\mathcal {N}}
\frac{e^{2i\theta(x,t;\zeta_{j})}}{\zeta_{n}^{*}-\zeta_{j}}
\mathcal{\hat{V}}(x,t;\zeta_{j})C_{j},\tag{\theequation b}
\end{align}
which in turn imply that
\addtocounter{equation}{1}
\begin{align}
&\mathcal{\hat{V}}^{up}(\zeta_{j})=-iQ_{+}/\zeta_{j}-\sum_{l=1}^{2\mathcal {N}}
\mathcal{\bar{V}}^{up}(\zeta_{l}^{*})c_{l}^{\dagger}(\zeta_{j}^{*}),\quad
j=1,2,\cdots,2\mathcal {N},\tag{\theequation a}\label{T47}\\
&\mathcal{\bar{V}}^{up}(\zeta_{n}^{*})=I_{m}+\sum_{j=1}^{2\mathcal {N}}
\mathcal{\hat{V}}^{up}(\zeta_{j})c_{j}(\zeta_{n}^{*}),\qquad\quad
n=1,2,\cdots,2\mathcal {N},\tag{\theequation b}\label{T49}
\end{align}
for brevity, the subscript $x$ and $t$ are omitted. Inserting   Eq.\eqref{T47} into   Eq.\eqref{T49}, one has
\begin{align}\label{T50}
\mathcal{\bar{V}}^{up}(\zeta_{n}^{*})=I_{m}-iQ_{+}\sum_{j=l}^{2\mathcal {N}}
c_{j}(\zeta_{n}^{*})/\zeta_{j}-\sum_{j=l}^{2\mathcal {N}}\sum_{l=1}^{2\mathcal {N}}
\mathcal{\bar{V}}^{up}(\zeta_{n}^{*})c_{l}^{\dagger}(\zeta_{j}^{*})
c_{j}(\zeta_{n}^{*}),\quad n=1,2,\cdots,2\mathcal {N}.
\end{align}
In order to get the eigenfunction $\mathcal{\bar{V}}^{up}(\zeta_{n}^{*})$, we transform Eq.\eqref{T50} into the matrix form. Assume that
\addtocounter{equation}{1}
\begin{align}
&\textbf{X}=(X_{1},X_{2},\cdots,X_{2\mathcal {N}})^{T},\quad
X_{n}=\mathcal{\bar{V}}^{up}(\zeta_{n}^{*}),\tag{\theequation a}\label{T51}\\
&\textbf{Y}=(Y_{1},Y_{2},\cdots,Y_{2\mathcal {N}})^{T},\quad
Y_{n}=I_{m}-iQ_{+}\sum_{j=l}^{2\mathcal {N}}
c_{j}(\zeta_{n}^{*})/\zeta_{j},\tag{\theequation b}\label{T52}\\
&\Omega=(\Omega_{n,l}),\quad \Omega_{n,l}=\sum_{j=l}^{2\mathcal {N}}
c_{l}^{\dagger}(\zeta_{j}^{*})
c_{j}(\zeta_{n}^{*}),\tag{\theequation c}\label{T53}
\end{align}
for $n,j=1,2,\cdots,2\mathcal {N}$. Based on these notations,   Eq.\eqref{T50} can be written as
\begin{align}\label{51}
K\textbf{X}=\textbf{Y},\quad K=I_{p}+\Omega,
\end{align}
where $I_{p}$ is the unity matrix and $p=2m\mathcal {N}$. The solution  \eqref{Qjie} is proved by substituting the solution obtained from solving   Eq.\eqref{T51} by Cramer's law into   Eq.\eqref{T42} without the reflection-less potential.
\end{proof}

In addition, we observe that even if discrete eigenvalues are quartets in the nonzero boundary conditions, the potential function $Q(x,t)$ is still  expressed by $2\mathcal {N}$ unknowns determined by Eq.\eqref{T50}. Moreover we have
\addtocounter{equation}{1}
\begin{align}
&\mathcal{\hat{V}}^{up}(x,t;\zeta_{n})=-\frac{i}{\zeta_{n}}
\mathcal{\bar{V}}^{up}(x,t;\zeta_{n+\mathcal {N}})Q_{+}, \tag{\theequation a}\\
&\mathcal{\hat{V}}^{up}(x,t;\zeta_{n+\mathcal {N}})=-\frac{i\epsilon
\zeta_{n}^{*}}{k_{0}^{2}}\mathcal{\bar{V}}^{up}(x,t;\zeta_{n}^{*})Q_{+},
 \tag{\theequation b}
\end{align}
for all $n=1,2,\cdots,\mathcal {N}$.

\section{Soliton solutions}
In this section, we will give the soliton solutions of the focusing modified KdV equation under the condition that the potential function $Q(x,t)$ is a scalar and   $2\times2$ symmetric matrix.
\begin{thm}
Assume that the eigenvalue  $\zeta_{1}\in D^{+}$ in the upper half plane, namely $|\zeta_{1}|>k_{0}$ and $Im\zeta_{1}>0$.
The potential $Q(x,t)$ can be given by for $m=2$ and $\mathcal {N}=1$
\begin{align}\label{T54}
Q(x,t)=Q_{+}-iX_{1}e^{-2i\theta(x,t;\zeta_{1}^{*})}C_{1}^{\dagger}+
\frac{iX_{2}e^{2i\theta(x,t;\zeta_{1})}Q_{+}C_{1}Q_{+}}{\zeta_{1}^{2}},
\end{align}
where $X_{i}$ ($i=1,2$) are determined by   Eq.\eqref{T51}.
\end{thm}

In what following, we will use the residue conditions, discrete eigenvalues and arbitrary norming constants to give the specific expressions of $X_{1}$ and $X_{2}$, and then the potential function $Q(x,t)$ can also be represented by these data, so as to select the appropriate parameters to study the propagation behavior of the solution. Note that for $m=2$, from   Eq.\eqref{T51} one has
\addtocounter{equation}{1}
\begin{align}
&X_{1}\left(I_{2}+\frac{i}{(\zeta_{1}^{*})^{2}+k_{0}^{2}}C_{1}^{\dagger}
Q_{+}^{\dagger}e^{-2i\theta(\zeta_{1}^{*})}\right)=I_{2}-\frac{i}{\zeta_{1}}
X_{2}Q_{+}c_{1}(\zeta_{1}^{*}),\tag{\theequation a}\\
&X_{2}\left(I_{2}-\frac{i}{\zeta_{1}^{2}+k_{0}^{2}}Q_{+}C_{1}
e^{2i\theta(\zeta_{1})}\right)=I_{2}+\frac{i\zeta_{1}^{*}}{k_{0}^{2}}
X_{1}Q_{+}c_{2}(\zeta_{2}^{*}),\tag{\theequation b}
\end{align}
and with the notation Eq.\eqref{notation}, $\zeta_{\mathcal {N}+j}=-k_{0}^{2}/\zeta_{j}^{*}$ as well as the associated norming constant such that $C_{\mathcal {N}+j}=Q_{+}^{\dagger}\bar{C}_{j}Q_{+}^{\dagger}$. From the Eq.\eqref{102c} for $j=1,2,\cdots,\mathcal {N}$, we have
\begin{align}
c_{1}(\zeta_{1}^{*})=\frac{C_{1}}{\zeta_{1}^{*}
-\zeta_{1}}e^{2i\theta(x,t;\zeta_{1})},\quad
c_{2}(\zeta_{2}^{*})=\frac{\zeta_{1}}{\zeta_{1}^{*}k_{0}^{2}(\zeta_{1}^{*}-\zeta_{1})}
Q_{+}^{\dagger}C_{1}^{\dagger}Q_{+}^{\dagger}e^{-2i\theta(x,t;\zeta_{1}^{*})}.
\end{align}
Obviously  the expression of $X_{1}$ and $X_{2}$ can be derived by the  Eqs.($4.2$)
\addtocounter{equation}{1}
\begin{align}
&X_{1}=\left[I_{2}-\frac{i}{\zeta_{1}}A_{2}^{-1}Q_{+}c_{1}(\zeta_{1}^{*})\right]
\left[A_{1}-\frac{\zeta_{1}^{*}}{\zeta_{1}k_{0}^{2}}Q_{+}
c_{2}(\zeta_{2}^{*})A_{2}^{-1}Q_{+}c_{1}(\zeta_{1}^{*})\right]^{-1},
\tag{\theequation a}\label{T55}\\
&X_{2}=\left[I_{2}+\frac{i\zeta_{1}^{*}}{k_{0}^{2}}A_{1}^{-1}Q_{+}c_{2}(\zeta_{2}^{*})\right]
\left[A_{2}-\frac{\zeta_{1}^{*}}{\zeta_{1}k_{0}^{2}}Q_{+}
c_{1}(\zeta_{1}^{*})A_{1}^{-1}Q_{+}c_{2}(\zeta_{2}^{*})\right]^{-1},
\tag{\theequation b}
\end{align}
where
\addtocounter{equation}{1}
\begin{align}
&A_{1}=I_{2}+\frac{i}{(\zeta_{1}^{*})^{2}+k_{0}^{2}}C_{1}^{\dagger}
Q_{+}^{\dagger}e^{-2i\theta(\zeta_{1}^{*})},\tag{\theequation a}\\
&A_{2}=A_{1}^{\dagger}=I_{2}-\frac{i}{\zeta_{1}^{2}+k_{0}^{2}}Q_{+}C_{1}
e^{2i\theta(\zeta_{1})},\tag{\theequation b}\label{T56}
\end{align}
and the norming constant $C_{1}$ is
\begin{align}
C_{1}=\left(\begin{array}{cc}
\varrho_{1} & \varrho_{0}\\
\varrho_{0} & \varrho_{2}
\end{array}\right),
\end{align}
with $\varrho_{i}\in\textbf{C}$ for ($i=0,1,2$).
\subsection{The scalar focusing modified KdV equation}
For scalar case, i.e., the $Q(x,t)=q(x,t)$, then from Eq.\eqref{T54}, the NZBCs \eqref{T2} and \eqref{T3},  the solution of the focusing modified KdV equation \eqref{T1}   can be written as for $\mathcal {N}=1$
\begin{align}
q(x,t)=q_{+}-iX_{1}e^{-2i\theta(\zeta_{1}^{*})}\varrho_{1}^{\dagger}+
i X_{2}e^{2i\theta(\zeta_{1})}\frac{q_{+}\varrho_{1}q_{+}}{\zeta_{1}^{2}},
\end{align}
where
\addtocounter{equation}{1}
\begin{align}
&X_{1}=\left[1-\frac{i}{\zeta_{1}}A_{2}^{-1}q_{+}c_{1}(\zeta_{1}^{*})\right]
\left[A_{1}-\frac{\zeta_{1}^{*}}{\zeta_{1}k_{0}^{2}}q_{+}
c_{2}(\zeta_{2}^{*})A_{2}^{-1}q_{+}c_{1}(\zeta_{1}^{*})\right]^{-1},
\tag{\theequation a}\\
&X_{2}=\left[1+\frac{i\zeta_{1}^{*}}{k_{0}^{2}}A_{1}^{-1}q_{+}c_{2}(\zeta_{2}^{*})\right]
\left[A_{2}-\frac{\zeta_{1}^{*}}{\zeta_{1}k_{0}^{2}}q_{+}
c_{1}(\zeta_{1}^{*})A_{1}^{-1}q_{+}c_{2}(\zeta_{2}^{*})\right]^{-1},
\tag{\theequation b}\\
&A_{1}=1+\frac{i}{(\zeta_{1}^{*})^{2}+k_{0}^{2}}\varrho_{1}^{\dagger}
q_{+}^{\dagger}e^{-2i\theta(\zeta_{1}^{*})},\quad
c_{1}(\zeta_{1}^{*})=\frac{\varrho_{1}}{\zeta_{1}^{*}
-\zeta_{1}}e^{2i\theta(x,t;\zeta_{1})},\tag{\theequation c}\\
&A_{2}=1-\frac{i}{\zeta_{1}^{2}+k_{0}^{2}}q_{+}\varrho_{1}
e^{2i\theta(\zeta_{1})},\quad
c_{2}(\zeta_{2}^{*})=\frac{\zeta_{1}}{\zeta_{1}^{*}k_{0}^{2}(\zeta_{1}^{*}-\zeta_{1})}
q_{+}^{\dagger}\varrho_{1}^{\dagger}q_{+}^{\dagger}e^{-2i\theta(x,t;\zeta_{1}^{*})}.
\tag{\theequation d}
\end{align}
\begin{rem}
Similar to the study of nonlinear Schr\"{o}dinger equation \cite{Biondini-2014}, the Tajiri-Watanabe breather solution \cite{Tajiri-1998} will be obtained by taking the eigenvalue $\zeta_{1}\in D^{+}$. When the eigenvalue $\zeta_{1}$ is on the imaginary axis ($Im\zeta_{1}>k_{0}$), the Kuznetsov-Ma solution \cite{Kuznetsov-1977,Ma-1979} can be obtained. When the eigenvalue is on the circle $C_{0}$, the Akhmediev breather solution \cite{Akhmediev-1987} can be obtained.
\end{rem}
The dynamic behavior of the one-soliton solution can be expressed as Figure 3 by selecting appropriate parameters

{\rotatebox{0}{\includegraphics[width=3.6cm,height=3.0cm,angle=0]{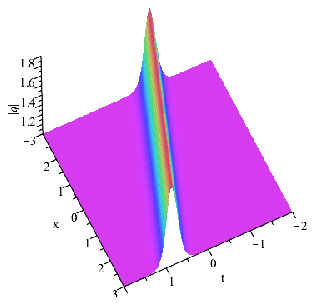}}}
\quad{\rotatebox{0}{\includegraphics[width=3.6cm,height=3.2cm,angle=0]{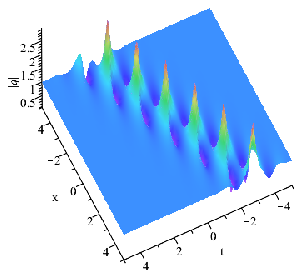}}}
\qquad{\rotatebox{0}{\includegraphics[width=3.6cm,height=3.0cm,angle=0]{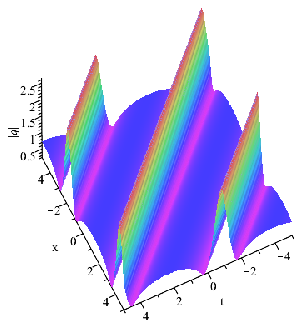}}}

$\qquad\quad\quad(\textbf{a1})\quad \ \qquad\quad\qquad\qquad\qquad(\textbf{b1})
~~\qquad\qquad\qquad\qquad\qquad(\textbf{c1})$\\

{\rotatebox{0}{\includegraphics[width=3.6cm,height=3.0cm,angle=0]{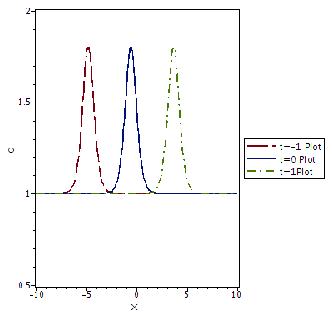}}}
\quad{\rotatebox{0}{\includegraphics[width=3.6cm,height=3.0cm,angle=0]{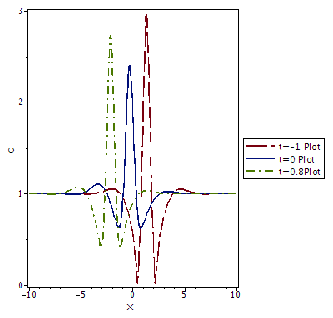}}}
\qquad{\rotatebox{0}{\includegraphics[width=3.6cm,height=3.0cm,angle=0]{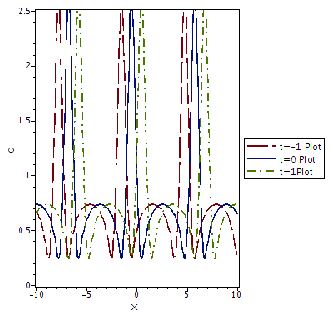}}}

$~\quad\quad\quad(\textbf{a2})\quad \ \quad\qquad\qquad\qquad\qquad(\textbf{b2})
~~\qquad\qquad\qquad\qquad\qquad(\textbf{c2})$\\

\noindent { \small \textbf{Figure 3.} $\textbf{(a1)}$ The bright soliton solution with the parameters $k_{0}=1$, $z_{1}=2i$, $\varrho_{1}=1$; $\textbf{(b1)}$ The breather solution with the parameters $k_{0}=1$, $z_{1}=1+\frac{3}{2}i$, $\varrho_{1}=1$; $\textbf{(c1)}$ The soliton solution with the parameters $k_{0}=1$, $z_{1}=\frac{1}{2}+\frac{\sqrt{3}}{2}i$, $\varrho_{1}=1$; $\textbf{(a1-c1)}$ The propagation view of corresponding solution at different times.}

For $\mathcal {N}=2$, the two-soliton solution can be given by the Eq.\eqref{Qjie}, i.e.,
\begin{align}\label{T59}
\begin{split}
q(x,t)=q_{+}&-ie^{-2i\theta(\zeta_{1}^{*})}X_{1}C_{1}^{\dagger}-
ie^{-2i\theta(\zeta_{2}^{*})}X_{2}C_{2}^{\dagger}\\&+\frac{iq_{+}C_{1}q_{+}}{\zeta_{1}^{2}}
e^{2i\theta(\zeta_{1})}X_{3}+\frac{iq_{+}C_{2}q_{+}}{\zeta_{2}^{2}}
e^{2i\theta(\zeta_{2})}X_{4},
\end{split}
\end{align}
where the $X_{i}$ $(i=1,2,3,4)$ are determined by
\addtocounter{equation}{1}
\begin{align*}
&X_{1}\left(1-\frac{i\zeta_{1}^{*}}{k_{0}^{2}}q_{+}c_{3}(\zeta_{1}^{*})\right)=
1-\frac{1}{\zeta_{1}}X_{3}q_{+}c_{1}(\zeta_{1}^{*})
-\frac{1}{\zeta_{2}}X_{4}q_{+}c_{2}(\zeta_{1}^{*})+
\frac{i\zeta_{2}^{*}}{k_{0}^{2}}X_{2}
q_{+}c_{4}(\zeta_{1}^{*}),\tag{\theequation a}\\
&X_{2}\left(1-\frac{i\zeta_{2}^{*}}{k_{0}^{2}}q_{+}c_{4}(\zeta_{2}^{*})\right)=
1-\frac{1}{\zeta_{1}}X_{3}q_{+}c_{1}(\zeta_{2}^{*})
-\frac{1}{\zeta_{2}}X_{4}q_{+}c_{2}(\zeta_{2}^{*})+
\frac{i\zeta_{1}^{*}}{k_{0}^{2}}X_{1}
q_{+}c_{3}(\zeta_{2}^{*}),\tag{\theequation b}\\
&X_{3}\left(1+\frac{1}{\zeta_{1}}q_{+}c_{1}(\zeta_{3}^{*})\right)=
1-\frac{1}{\zeta_{2}}X_{4}q_{+}c_{2}(\zeta_{3}^{*})+
\frac{i\zeta_{1}^{*}}{k_{0}^{2}}X_{1}q_{+}c_{3}(\zeta_{3}^{*})
+\frac{i\zeta_{2}^{*}}{k_{0}^{2}}X_{2}
q_{+}c_{4}(\zeta_{3}^{*}),\tag{\theequation c}\\
&X_{4}\left(1+\frac{1}{\zeta_{2}}q_{+}c_{2}(\zeta_{3}^{*})\right)=
1-\frac{1}{\zeta_{1}}X_{3}q_{+}c_{1}(\zeta_{4}^{*})+
\frac{i\zeta_{1}^{*}}{k_{0}^{2}}X_{1}q_{+}c_{3}(\zeta_{4}^{*})
+\frac{i\zeta_{2}^{*}}{k_{0}^{2}}X_{2}
q_{+}c_{4}(\zeta_{4}^{*}).\tag{\theequation d}
\end{align*}
The dynamic behavior of the two-soliton solution can be expressed as Figure 4 by selecting appropriate parameters

{\rotatebox{0}{\includegraphics[width=3.75cm,height=3.5cm,angle=0]{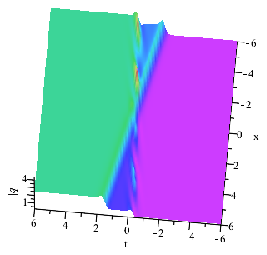}}}
\qquad\qquad\qquad\qquad
{\rotatebox{0}{\includegraphics[width=3.75cm,height=3.3cm,angle=0]{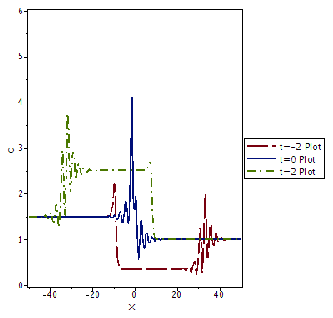}}}

 $\qquad\quad\qquad(\textbf{a})\qquad \ \qquad\qquad~\qquad\qquad \qquad\qquad\qquad(\textbf{b})$\\
\noindent {\small \textbf{Figure 4.} Propagation of the two-soliton solution  with the the parameters $\textbf{(a)}$ $k_{0}=1$, $z_{1}=2i$, $z_{2}=-2+i/2$, $C_{1}=C_{2}=1$. $\textbf{(b)}$ the wave propagation of the two-soliton solution with $t=-2$, $t=0$, $t=2$.}

\subsection{The potential is the symmetric matrix}
Similar to the scalar case, in order to obtain the soliton solution for the potential $Q$ is the $2\times2$ symmetric matrix, we need to accurately calculate the expressions of $X_{1}$ and $X_{2}$. Now for convenience, let's assume $Q_{+}=k_{0}I_{2}$, and from Eqs.($4.6$) and ($4.7$) we have
\addtocounter{equation}{1}
\begin{align}
&X_{1}=\Upsilon_{1} \Gamma_{1}^{-1},\quad X_{2}=\Upsilon_{2} \Gamma_{2}^{-1},
\tag{\theequation a}\\
&\Upsilon_{1}=I_{2}-h_{11}C_{1}-h_{12}I_{2},\quad
\Upsilon_{2}=I_{2}+h_{21}C_{1}-h_{22}I_{2},\tag{\theequation b}\\
&\Gamma_{1}=I_{2}+e^{-2i\theta{(\zeta_{1}^{*})}}\left(g_{11}C_{1}^{\dagger}
-g_{12}C_{1}^{\dagger}C_{1}\right),\tag{\theequation c}\label{T57}\\
&\Gamma_{2}=I_{2}+e^{2i\theta{(\zeta_{1})}}\left(g_{21}C_{1}
-g_{22}C_{1}C_{1}^{\dagger}\right),\tag{\theequation d}\label{T58}
\end{align}
where
\addtocounter{equation}{1}
\begin{align}
&h_{11}=\frac{ik_{0}}{\zeta_{1}(\zeta_{1}^{*}-\zeta_{1})}
\frac{e^{2i\theta{(\zeta_{1})}}}{\varpiup_{1}^{*}},\quad
h_{12}=\frac{k_{0}^{2}\det C_{1}}{\zeta_{1}(\zeta_{1}^{*}-\zeta_{1})
(\zeta_{1}^{2}+k_{0}^{2})}\frac{e^{4i\theta{(\zeta_{1})}}}{\varpiup_{1}^{*}},
\tag{\theequation a}\\
&h_{21}=\frac{i\zeta_{1}}{k_{0}(\zeta_{1}^{*}-\zeta_{1})}\frac
{e^{-2i\theta{(\zeta_{1}^{*})}}}{\varpiup_{1}},\quad
h_{22}=\frac{\zeta_{1}\det C_{1}^{\dagger}}{(\zeta_{1}^{*}-\zeta_{1})
((\zeta_{1}^{*})^{2}+k_{0}^{2})}\frac{e^{-4i\theta{(\zeta_{1}^{*})}}}{\varpiup_{1}},
\tag{\theequation b}\\
&g_{11}=\frac{ik_{0}}{(\zeta_{1}^{*})^{2}+k_{0}^{2}}\left(
1+\frac{(\zeta_{1}^{*})^{2}+k_{0}^{2}}{\zeta_{1}^{2}+k_{0}^{2}}
\frac{\det C_{1}}{(\zeta_{1}^{*}-\zeta_{1})^{2}}\frac{e^{4i\theta{(\zeta_{1})}}}
{\varpiup_{1}^{*}}\right),\tag{\theequation c}\\
&g_{21}=-\frac{ik_{0}}{\zeta_{1}^{2}+k_{0}^{2}}\left(
1+\frac{\zeta_{1}^{2}+k_{0}^{2}}{(\zeta_{1}^{*})^{2}+k_{0}^{2}}
\frac{\det C_{1}^{\dagger}}{(\zeta_{1}^{*}-\zeta_{1})^{2}}
\frac{e^{-4i\theta{(\zeta_{1}^{*})}}}{\varpiup_{1}^{*}}\right),
\tag{\theequation d}\\
&g_{12}=\frac{1}{(\zeta_{1}^{*}-\zeta_{1})^{2}}
\frac{e^{2i\theta{(\zeta_{1})}}}{\varpiup_{1}^{*}},\quad
g_{22}=\frac{1}{(\zeta_{1}^{*}-\zeta_{1})^{2}}
\frac{e^{-2i\theta{(\zeta_{1}^{*})}}}{\varpiup_{1}},\tag{\theequation e}
\end{align}
with
\addtocounter{equation}{1}
\begin{align}
&A_{1}^{-1}=\frac{1}{\varpiup_{1}}\left(I_{2}+
\frac{ik_{0}}{(\zeta_{1}^{*})^{2}+k_{0}^{2}}e^{-2i\theta
{(\zeta_{1}^{*})}}\text{cof}(C_{1}^{\dagger})\right),\tag{\theequation a}\\
&A_{2}^{-1}=\frac{1}{\varpiup_{1}^{*}}\left(I_{2}-
\frac{ik_{0}}{\zeta_{1}^{2}+k_{0}^{2}}e^{2i\theta
{(\zeta_{1})}}\text{cof}(C_{1})\right),\tag{\theequation b}\\
&\varpiup_{1}=\det A_{1}=1+\frac{ik_{0}}{(\zeta_{1}^{*})^{2}+k_{0}^{2}}
e^{-2i\theta{(\zeta_{1}^{*})}}trace(C_{1}^{\dagger})-\frac{k_{0}^{2}}
{((\zeta_{1}^{*})^{2}+k_{0}^{2})}e^{-4i\theta
{(\zeta_{1}^{*})}}\det C_{1}^{\dagger}.\tag{\theequation c}
\end{align}
Obviously when the matrix $C_{1}$ is rank $1$, i.e., $\det C_{1}=0$, the above equations can be written as
\addtocounter{equation}{1}
\begin{align}
&\Upsilon_{1}=I_{2}-h_{11}C_{1},\quad \Upsilon_{2}=I_{2}+h_{21}C_{1}^{\dagger},
\tag{\theequation a}\\
&\Gamma_{1}=I_{2}+e^{-2i\theta{(\zeta_{1}^{*})}}\left(
\frac{ik_{0}}{(\zeta_{1}^{*})^{2}+k_{0}^{2}}C_{1}^{\dagger}
-g_{12}C_{1}^{\dagger}C_{1}\right), \tag{\theequation b}\\
&\Gamma_{2}=I_{2}+e^{2i\theta{(\zeta_{1})}}\left(-
\frac{ik_{0}}{(\zeta_{1})^{2}+k_{0}^{2}}C_{1}-g_{22}C_{1}C_{1}^{\dagger}\right).
\tag{\theequation c}
\end{align}

The various soliton solutions can be obtained by using the above expressions.
In Fig. 5, for the general discrete eigenvalue $z_{1}\in D^{+}$, $(a1-c1)$ show a soliton solution of   Eq.\eqref{T1} with the $2\times2$ symmetric matrix $Q$ Eq.\eqref{Q} when the norming  constant is a full rank matrix, and $(a2-c2)$ are the propagation views of the corresponding soliton solution at different times. The difference between Fig. 6 and Fig. 5 is that the norming  constant is not a full rank matrix but is equal to one. In Appendix D, we discuss the large $x$ asymptotic behavior when the  norming  constant $\det C_{1}=0$ and $\det C_{1}\neq0$.

{\rotatebox{0}{\includegraphics[width=3.6cm,height=3.0cm,angle=0]{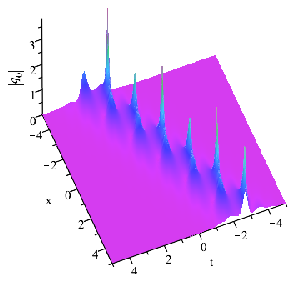}}}
\quad{\rotatebox{0}{\includegraphics[width=3.6cm,height=3.2cm,angle=0]{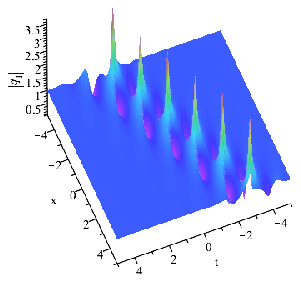}}}
\qquad{\rotatebox{0}{\includegraphics[width=3.6cm,height=3.0cm,angle=0]{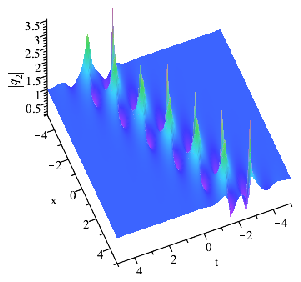}}}

$\qquad\quad\quad(\textbf{a1})\quad \ \qquad\qquad\qquad\qquad\qquad(\textbf{b1})
~~\qquad\qquad\qquad\qquad\qquad(\textbf{c1})$\\

{\rotatebox{0}{\includegraphics[width=3.6cm,height=3.0cm,angle=0]{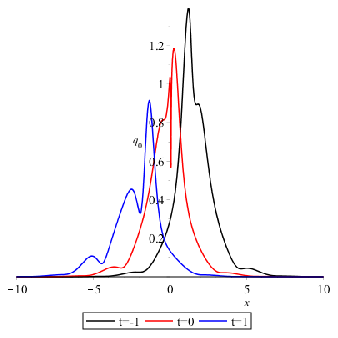}}}
\quad{\rotatebox{0}{\includegraphics[width=3.6cm,height=3.2cm,angle=0]{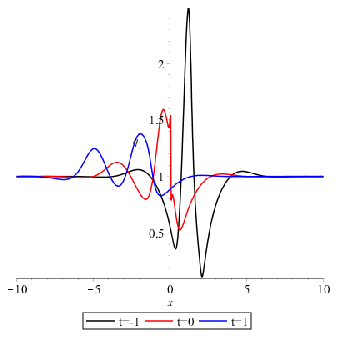}}}
\qquad{\rotatebox{0}{\includegraphics[width=3.6cm,height=3.0cm,angle=0]{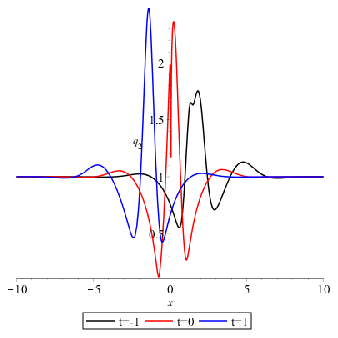}}}

$\qquad\quad\quad(\textbf{a2})\quad \ \qquad\qquad\qquad\qquad\qquad(\textbf{b2})
~~\qquad\qquad\qquad\qquad\qquad(\textbf{c2})$\\

\noindent { \small \textbf{Figure 5.} $\textbf{(a1-c1)}$ The breather wave of the solution Eq.\eqref{T54}, i.e., three components $(q_{0},q_{1},q_{2})$ with $Q_{+}=I_{2}$, $\zeta_{1}=1+3i/2$, $\varrho_{1}=\varrho_{0}\equiv1$, $\varrho_{2}=3$, then $\det C_{1}\neq0$. $\textbf{(a2-c2)}$ The propagation view of soliton solution in different time.}

{\rotatebox{0}{\includegraphics[width=3.6cm,height=3.0cm,angle=0]{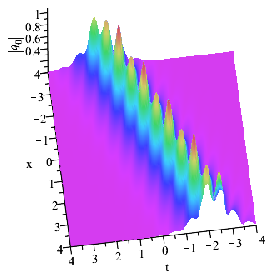}}}
\quad{\rotatebox{0}{\includegraphics[width=3.6cm,height=3.2cm,angle=0]{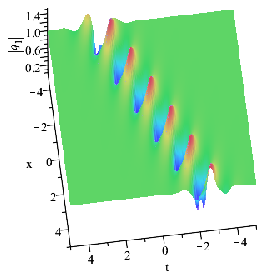}}}
\qquad{\rotatebox{0}{\includegraphics[width=3.6cm,height=3.0cm,angle=0]{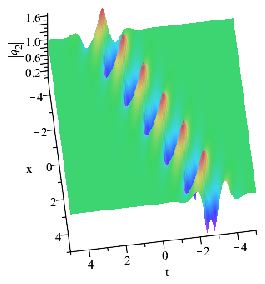}}}

$\qquad\quad\quad(\textbf{a1})\quad \ \qquad\qquad\qquad\qquad\qquad(\textbf{b1})
~~\qquad\qquad\qquad\qquad\qquad(\textbf{c1})$\\

{\rotatebox{0}{\includegraphics[width=3.6cm,height=3.0cm,angle=0]{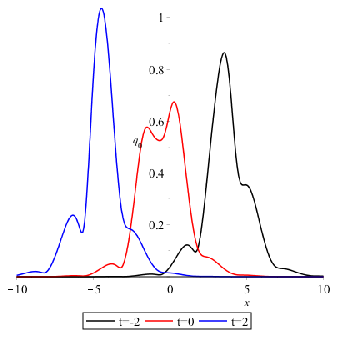}}}
\quad{\rotatebox{0}{\includegraphics[width=3.6cm,height=3.2cm,angle=0]{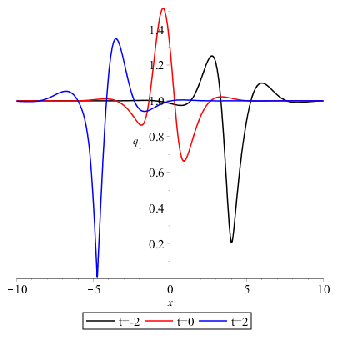}}}
\qquad{\rotatebox{0}{\includegraphics[width=3.6cm,height=3.0cm,angle=0]{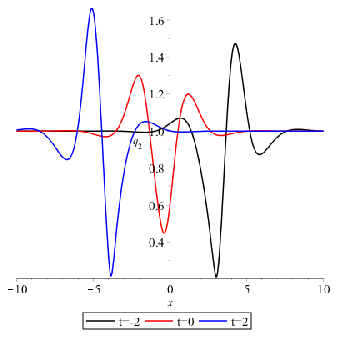}}}

$\qquad\quad\quad(\textbf{a2})\quad \ \qquad\qquad\qquad\qquad\qquad(\textbf{b2})
~~\qquad\qquad\qquad\qquad\qquad(\textbf{c2})$\\

\noindent { \small \textbf{Figure 6.} $\textbf{(a1-c1)}$ The breather wave of the solution Eq.\eqref{T54}, i.e., three components $(q_{0},q_{1},q_{2})$ with $Q_{+}=I_{2}$, $\zeta_{1}=1+3i/2$, $\varrho_{1}=1,$ $\varrho_{0}=i$, $\varrho_{2}=-1$, then $\det C_{1}=0$. $\textbf{(a2-c2)}$ the propagation view of soliton solution in different time.}

In Figs 7 and 8, we present the soliton solutions when the discrete eigenvalue $z_{1}\in D^{+}$  is purely imaginary ($z_{1}=iZ, Z>k_{0}$). Similar to the Figs 5 and 6, the two cases are discussed including $\det C_{1}=0$ and $\det C_{1}\neq0$.

{\rotatebox{0}{\includegraphics[width=3.6cm,height=3.0cm,angle=0]{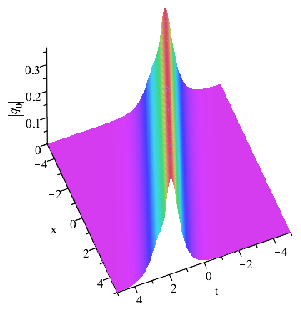}}}
\quad{\rotatebox{0}{\includegraphics[width=3.6cm,height=3.2cm,angle=0]{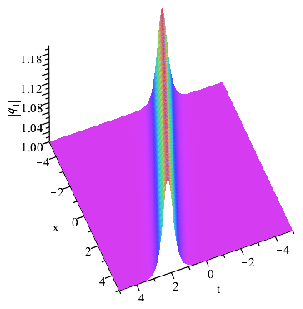}}}
\qquad{\rotatebox{0}{\includegraphics[width=3.6cm,height=3.0cm,angle=0]{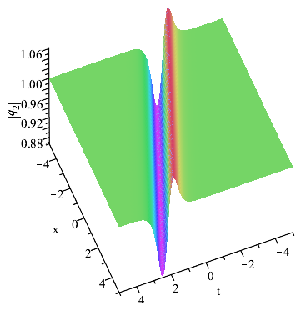}}}

$\qquad\quad\quad(\textbf{a1})\quad \ \qquad\qquad\qquad\qquad\qquad(\textbf{b1})
~~\qquad\qquad\qquad\qquad\qquad(\textbf{c1})$\\

{\rotatebox{0}{\includegraphics[width=3.6cm,height=3.0cm,angle=0]{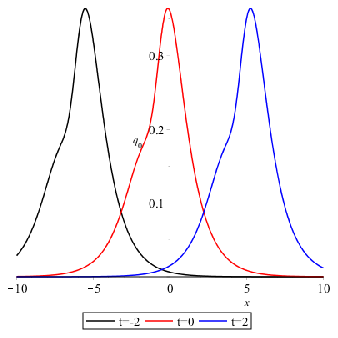}}}
\quad{\rotatebox{0}{\includegraphics[width=3.6cm,height=3.2cm,angle=0]{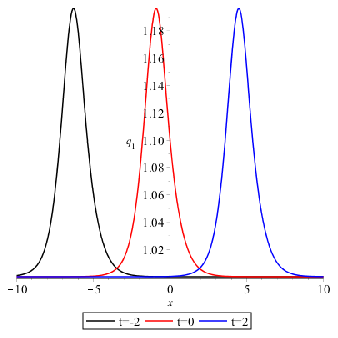}}}
\qquad{\rotatebox{0}{\includegraphics[width=3.6cm,height=3.0cm,angle=0]{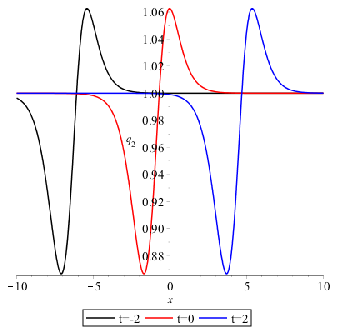}}}

$\qquad\quad\quad(\textbf{a2})\quad \ \qquad\qquad\qquad\qquad\qquad(\textbf{b2})
~~\qquad\qquad\qquad\qquad\qquad(\textbf{c2})$\\

\noindent { \small \textbf{Figure 7.} $\textbf{(a1-b1)}$ The bright soliton solution, $\textbf{(c1)}$ the bright-dark soliton solution  i.e., three components $(q_{0},q_{1}$ $q_{2})$ with $Q_{+}=I_{2}$, $\zeta_{1}=3i/2$, $\varrho_{1}=1,$ $\varrho_{0}=i$, $\varrho_{2}=-1$, then $\det C_{1}=0$. $\textbf{(a2-c2)}$ the propagation view of soliton solution in different time.}

{\rotatebox{0}{\includegraphics[width=3.6cm,height=3.0cm,angle=0]{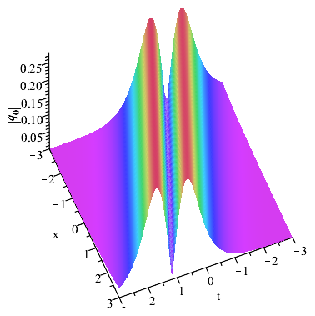}}}
\quad{\rotatebox{0}{\includegraphics[width=3.6cm,height=3.2cm,angle=0]{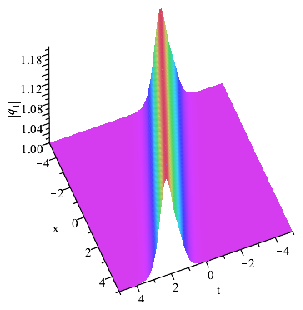}}}
\qquad{\rotatebox{0}{\includegraphics[width=3.6cm,height=3.0cm,angle=0]{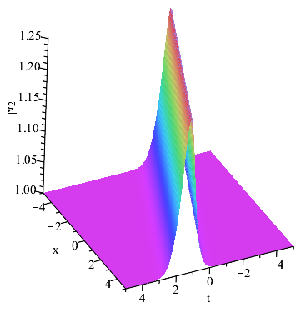}}}

$\qquad\quad\quad(\textbf{a1})\quad \ \qquad\qquad\qquad\qquad\qquad(\textbf{b1})
~~\qquad\qquad\qquad\qquad\qquad(\textbf{c1})$\\

{\rotatebox{0}{\includegraphics[width=3.6cm,height=3.0cm,angle=0]{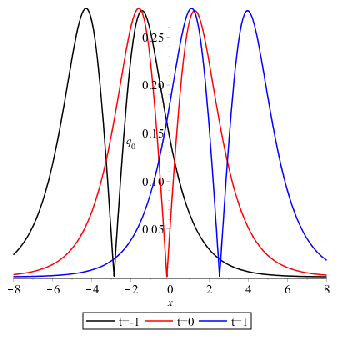}}}
\quad{\rotatebox{0}{\includegraphics[width=3.6cm,height=3.2cm,angle=0]{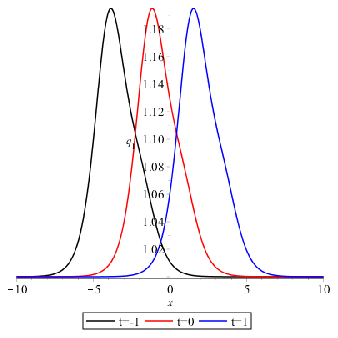}}}
\qquad{\rotatebox{0}{\includegraphics[width=3.6cm,height=3.0cm,angle=0]{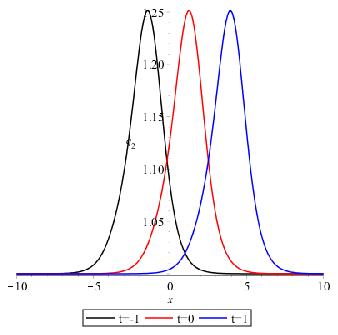}}}

$\qquad\quad\quad(\textbf{a2})\quad \ \qquad\qquad\qquad\qquad\qquad(\textbf{b2})
~~\qquad\qquad\qquad\qquad\qquad(\textbf{c2})$\\
\noindent { \small \textbf{Figure 8.} $\textbf{(a1)}$ The $M$-type soliton solution, $\textbf{(b1-c1)}$ the bright soliton solution, i.e.,  three components $(q_{0},q_{1},q_{2})$ with $Q_{+}=I_{2}$, $\zeta_{1}=3i/2$, $\varrho_{1}=1,$ $\varrho_{0}=1$, $\varrho_{2}=2$, then $\det C_{1}\neq0$. $\textbf{(a2-c2)}$ the propagation view of soliton solution in different time.}

In what follows, the soliton solution of   Eq.\eqref{T1} will be derived based on   Eq.\eqref{T59} for $\mathcal {N}=2$ shown in Figure 9.
Note that the $X_{i}$ $(i=1,2,3,4)$ are determined by
\addtocounter{equation}{1}
\begin{align}
&X_{1}\left(1-\frac{i\zeta_{1}^{*}}{k_{0}^{2}}Q_{+}c_{3}(\zeta_{1}^{*})\right)=
1-\frac{X_{3}Q_{+}}{\zeta_{1}}c_{1}(\zeta_{1}^{*})
-\frac{X_{4}Q_{+}}{\zeta_{2}}c_{2}(\zeta_{1}^{*})+
\frac{iX_{2}Q_{+}\zeta_{2}^{*}}{k_{0}^{2}}c_{4}(\zeta_{1}^{*}),
\tag{\theequation a}\\
&X_{2}\left(1-\frac{i\zeta_{2}^{*}}{k_{0}^{2}}Q_{+}c_{4}(\zeta_{2}^{*})\right)=
1-\frac{X_{3}Q_{+}}{\zeta_{1}}c_{1}(\zeta_{2}^{*})
-\frac{X_{4}Q_{+}}{\zeta_{2}}c_{2}(\zeta_{2}^{*})+
\frac{iX_{1}Q_{+}\zeta_{1}^{*}}{k_{0}^{2}}c_{3}(\zeta_{2}^{*}),
\tag{\theequation b}\\
&X_{3}\left(1+\frac{1}{\zeta_{1}}Q_{+}c_{1}(\zeta_{3}^{*})\right)=
1-\frac{X_{4}Q_{+}}{\zeta_{2}}c_{2}(\zeta_{3}^{*})+
\frac{iX_{1}Q_{+}\zeta_{1}^{*}}{k_{0}^{2}}c_{3}(\zeta_{3}^{*})
+\frac{iX_{2}Q_{+}\zeta_{2}^{*}}{k_{0}^{2}}c_{4}(\zeta_{3}^{*}),
\tag{\theequation c}\\
&X_{4}\left(1+\frac{1}{\zeta_{2}}q_{+}c_{2}(\zeta_{3}^{*})\right)=
1-\frac{X_{3}Q_{+}}{\zeta_{1}}c_{1}(\zeta_{4}^{*})+
\frac{iX_{1}Q_{+}\zeta_{1}^{*}}{k_{0}^{2}}c_{3}(\zeta_{4}^{*})
+\frac{iX_{2}Q_{+}\zeta_{2}^{*}}{k_{0}^{2}}c_{4}(\zeta_{4}^{*}).
\tag{\theequation d}
\end{align}
with
\begin{align*}
&c_{1}(\zeta_{1}^{*})=\frac{C_{1}}{\zeta_{1}^{*}-\zeta_{1}}e^{2i\theta(\zeta_{1})},
\quad
c_{1}(\zeta_{2}^{*})=\frac{C_{1}}{\zeta_{1}^{*}-\zeta_{1}}e^{2i\theta(\zeta_{1})},\\
&c_{1}(\zeta_{3}^{*})=-\frac{\zeta_{1}C_{1}}{k_{0}^{2}+\zeta_{1}^{2}}
e^{2i\theta(\zeta_{1})},\quad
c_{1}(\zeta_{4}^{*})=-\frac{\zeta_{2}C_{1}}{k^{2}_{1}+\zeta_{1}\zeta_{2}}
e^{2i\theta(\zeta_{1})},\\
&c_{2}(\zeta_{1}^{*})=\frac{C_{2}}{\zeta_{1}^{*}-\zeta_{2}}e^{2i\theta(\zeta_{2})},
\quad
c_{2}(\zeta_{2}^{*})=\frac{C_{2}}{\zeta_{2}^{*}-\zeta_{2}}e^{2i\theta(\zeta_{2})},\\
&c_{2}(\zeta_{3}^{*})=-\frac{\zeta_{1}C_{2}}{k^{2}_{0}+\zeta_{1}\zeta_{2}}
e^{2i\theta(\zeta_{2})},\qquad
c_{2}(\zeta_{4}^{*})=-\frac{\zeta_{2}C_{2}}{k^{2}_{1}+\zeta_{2}^{2}}
e^{2i\theta(\zeta_{2})},\\
&c_{3}(\zeta_{1}^{*})=-\frac{Q_{+}^{\dagger}C_{1}^{\dagger}Q_{+}^{\dagger}}
{\zeta_{1}^{*}[k^{2}_{0}+(\zeta_{1}^{*})^{2}]}e^{-2i\theta(\zeta_{1}^{*})},\quad
c_{3}(\zeta_{2}^{*})=-\frac{Q_{+}^{\dagger}C_{1}^{\dagger}Q_{+}^{\dagger}}
{\zeta_{1}^{*}[k^{2}_{0}+\zeta_{1}^{*}\zeta_{2}^{*}]}e^{-2i\theta(\zeta_{1}^{*})},\\
&c_{3}(\zeta_{3}^{*})=\frac{\zeta_{1}Q_{+}^{\dagger}C_{1}^{\dagger}Q_{+}^{\dagger}}
{\zeta_{1}^{*}k^{2}_{0}(\zeta_{1}^{*}-\zeta_{1})}e^{-2i\theta(\zeta_{1}^{*})},\quad
c_{3}(\zeta_{4}^{*})=\frac{\zeta_{2}Q_{+}^{\dagger}C_{1}^{\dagger}Q_{+}^{\dagger}}
{\zeta_{1}^{*}k^{2}_{0}(\zeta_{1}^{*}-\zeta_{2})}e^{-2i\theta(\zeta_{1}^{*})},
\end{align*}
\begin{align*}
&c_{4}(\zeta_{1}^{*})=-\frac{Q_{+}^{\dagger}C_{2}^{\dagger}Q_{+}^{\dagger}}
{\zeta_{2}^{*}[k^{2}_{0}+\zeta_{1}^{*}\zeta_{2}^{*}]}e^{-2i\theta(\zeta_{2}^{*})},\quad
c_{4}(\zeta_{2}^{*})=-\frac{Q_{+}^{\dagger}C_{2}^{\dagger}Q_{+}^{\dagger}}
{\zeta_{2}^{*}[k^{2}_{0}+(\zeta_{2}^{*})^{2}]}e^{-2i\theta(\zeta_{2}^{*})},\\
&c_{4}(\zeta_{3}^{*})=\frac{\zeta_{1}Q_{+}^{\dagger}C_{2}^{\dagger}Q_{+}^{\dagger}}
{\zeta_{2}^{*}k^{2}_{0}(\zeta_{2}^{*}-\zeta_{1})}e^{-2i\theta(\zeta_{2}^{*})},\quad
c_{4}(\zeta_{4}^{*})=\frac{\zeta_{2}Q_{+}^{\dagger}C_{1}^{\dagger}Q_{+}^{\dagger}}
{\zeta_{2}^{*}k^{2}_{0}(\zeta_{2}^{*}-\zeta_{2})}e^{-2i\theta(\zeta_{2}^{*})}.
\end{align*}
Now we just discuss the special case for the norming constants $C_{1}=\left(\begin{array}{cc}
0 & 1\\
1 & 0
\end{array}\right)$ and $C_{2}=\left(\begin{array}{cc}
1 & 0\\
0 & 0
\end{array}\right)$, and give the soliton solution $q_{1}$ as example for $\mathcal {N}=2$ shown in Fig 9 due to under the conditions $C_{1}$ and $C_{2}$, the solution of the equations about the $q_{0}$ and $q_{2}$ are  not related to eigenvalue $z_{2}$, namely

{\rotatebox{0}{\includegraphics[width=3.75cm,height=3.5cm,angle=0]{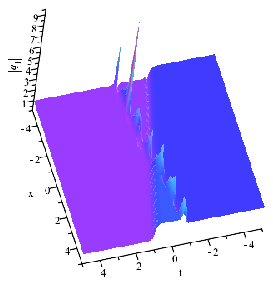}}}
\qquad\qquad\qquad\qquad
{\rotatebox{0}{\includegraphics[width=3.75cm,height=3.3cm,angle=0]{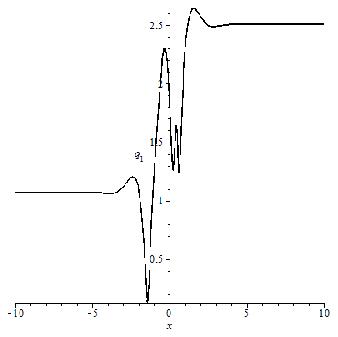}}}

 $\qquad\quad\qquad(\textbf{a})\qquad \ \qquad\qquad\quad\qquad\qquad \qquad\qquad\qquad(\textbf{b})$\\
\noindent {\small \textbf{Figure 9.} Propagation of the solution  with the the parameters  $k_{0}=1$, $z_{1}=-2i$, $z_{2}=2+2i$. $\textbf{(a)}$ the two-soliton solution. $\textbf{(b)}$ the wave propagation of the two-soliton solution with  $t=0$.}

\section{Conclusions and discussions}
We have systematically established the inverse scattering theory of the focusing and defocusing mmKdV equation with nonzero boundary conditions at infinity. In the direct scattering problem,the properties of the Jost eigenfunctions and  the scattering matrix have established, including analytical, asymptotic and symmetrical. According to the analysis of Jost eigenfunctions and scattering data, we have introduced the meromorphic functions to establish a suitable RH problem. We have used Cauchy projection operator and Plemelj's formula to solve this RH problem to reconstruct the modified eigenfunction, and then used these results and the asymptotic behaviours of the modified eigenfunctions to construct the potential function of mmKdV equation. Finally, we have given the exact solution of the mmKdV equation with no reflection.

In order to study the dynamic behavior of the solution of the mmKdV equation, we have studied in detail in two cases, including the scalar form and the potential function is a $2\times2$ symmetric matrix. In addition, we have analyzed that the value of NZBCs at infinity can be arbitrarily selected under the condition that the Eq.\eqref{T1} is satisfied.

\section*{Acknowledgements}
This work was supported by the Natural Science Foundation of Jiangsu Province under Grant No. BK20181351, the Qinglan Project of Jiangsu Province of China, the National Natural Science Foundation of China under Grant No. 11975306, and the Six Talent Peaks Project in Jiangsu Province under Grant No. JY-059, the Fundamental Research Fund for the Central Universities under the Grant Nos. 2019ZDPY07 and 2019QNA35, and the Future Outstanding Talent Assistance project for Postgraduate Research \& Practice Innovation Program of Jiangsu Province under Grant No.
2020WLJCRCZL031.

\section*{Appendix A. Proof of  Eqs 4.11:}
From the Eqs.\eqref{T55}-\eqref{T56}, one has
\begin{align*}
X_{1}=\left[I_{2}-\frac{i}{\zeta_{1}}A_{2}^{-1}Q_{+}c_{1}(\zeta_{1}^{*})\right]
\left[A_{1}-\frac{\zeta_{1}^{*}}{\zeta_{1}k_{0}^{2}}Q_{+}
c_{2}(\zeta_{2}^{*})A_{2}^{-1}Q_{+}c_{1}(\zeta_{1}^{*})\right]^{-1}\triangleq
\Upsilon_{1} \Gamma_{1}^{-1},
\end{align*}
where
\begin{align*}
\Upsilon_{1}&=I_{2}-\frac{i}{\zeta_{1}}A_{2}^{-1}Q_{+}c_{1}(\zeta_{1}^{*})\\&=
I_{2}-\frac{i}{\zeta_{1}}\frac{1}{\varpiup_{1}^{*}}\left(I_{2}-
\frac{ik_{0}}{\zeta_{1}^{2}+k_{0}^{2}}e^{2i\theta
{(\zeta_{1})}}\text{cof}(C_{1})\right)k_{0}I_{2}\frac{C_{1}}{\zeta_{1}^{*}
-\zeta_{1}}e^{2i\theta(\zeta_{1})}\\&=
I_{2}-\frac{ik_{0}}{\zeta_{1}(\zeta_{1}^{*}-\zeta_{1})}
\frac{e^{2i\theta{(\zeta_{1})}}}{\varpiup_{1}^{*}}C_{1}-\frac{k_{0}^{2}\det C_{1} }{\zeta_{1}(\zeta_{1}^{*}-\zeta_{1})
(\zeta_{1}^{2}+k_{0}^{2})}\frac{e^{4i\theta{(\zeta_{1})}}}{\varpiup_{1}^{*}}I_{2},\\
\Gamma_{1}&=A_{1}-\frac{\zeta_{1}^{*}}{\zeta_{1}k_{0}^{2}}Q_{+}
c_{2}(\zeta_{2}^{*})A_{2}^{-1}Q_{+}c_{1}(\zeta_{1}^{*})\\&=
I_{2}+\frac{ik_{0}^{2}C_{1}^{\dagger}}{(\zeta_{1}^{*})^{2}+k_{0}^{2}}
e^{-2i\theta(\zeta_{1}^{*})}I_{2}-\frac{C_{1}^{\dagger}e^{-2i\theta(\zeta_{1}^{*})}}
{\zeta_{1}^{*}-\zeta_{1}}I_{2}\frac{1}{\varpiup_{1}^{*}}\left(I_{2}-
\frac{ik_{0}}{\zeta_{1}^{2}+k_{0}^{2}}e^{2i\theta
{(\zeta_{1})}}\text{cof}(C_{1})\right)\frac{C_{1}}{\zeta_{1}^{*}
-\zeta_{1}}e^{2i\theta(\zeta_{1})}\\&=
I_{2}+C_{1}^{\dagger}e^{-2i\theta(\zeta_{1}^{*})}\left[
\frac{ik_{0}}{(\zeta_{1}^{*})^{2}+k_{0}^{2}}\left(
1+\frac{(\zeta_{1}^{*})^{2}+k_{0}^{2}}{\zeta_{1}^{2}+k_{0}^{2}}
\frac{\det C_{1}}{(\zeta_{1}^{*}-\zeta_{1})^{2}}\frac{e^{4i\theta{(\zeta_{1})}}}
{\varpiup_{1}^{*}}\right)I_{2}-\frac{1}{(\zeta_{1}^{*}-\zeta_{1})^{2}}
\frac{e^{2i\theta{(\zeta_{1})}}}{\varpiup_{1}^{*}}C_{1}\right].
\end{align*}
The $X_{2}$ can be proved in the same way.

\section*{Appendix B. The boundary condition $Q_{+}$}
In this Appendix, we will explain the reason why the boundary value condition $Q_{+}$ can be chosen $k_{0}I_{m}$ as $x\rightarrow+\infty$ generally.

Note that whether the solution $Q(x,t)$ of the Eq.\eqref{T1} is left-multiplication or right-multiplication by an arbitrary constant unitary matrix, it is still the solution of the Eq.\eqref{T1}. Thus for the potential $Q$ with an arbitrary condition $Q_{+}$ as $x\rightarrow+\infty$, without loss of generality we can introduce a new potential $\tilde{Q}=BQC$ satisfying $\tilde{Q}=k_{0}I_{m}$, where $B$ and $C$ are the arbitrary constant unitary matrix. It's not difficult to verify that $\tilde{Q}$ is the solution of the Eq.\eqref{T1} and means that $BQC=k_{0}I_{m}$. Obviously one has
\begin{align*}
\tilde{Q}(x,t)=BQ(x,t)(Q_{+}^{\dagger}/k_{0})B^{\dagger}.
\end{align*}
In addition, the new potential $\tilde{Q}$ is symmetric matrix, which leads to
\begin{equation*}
Q(x,t)=B^{\dagger}B^{*}(Q_{+}^{\dagger}/k_{0})Q(x,t)B^{T}B(Q_{+}/k_{0}).\eqno(\text{B}1)
\end{equation*}
Note $(\text{B}1)$ is satisfied when
\begin{align*}
B^{T}B=Q_{+}^{\dagger}/k_{0}
\end{align*}
and the Takagi's factorization algorithm can guarantee the existence of $B$.

\section*{Appendix C. The trace formula and theta condition}
The so-called trace formula is to use the scattering data including the discrete eigenvalues and reflection coefficients to represent the scattering coefficients  for the scalar equation as shown in Ref.\cite{ZS-1972}. Moreover, when under non-zero boundary conditions, the trace formula can also show the asymptotic phase difference of the potential function and the scattering data, that is, theta condition in Ref.\cite{Faddeev-1987,Biondini-2014}. For vector or matrix type partial differential equations (PDE), there are similar trace formula  and theta condition, but it should be noted that for matrix PDE, reconstruction of the scattering data $a(z)$ and $\bar{a}(z)$ is more difficult based on this fact  the Eqs.\eqref{SSS1} and \eqref{T34} need  to be decomposed. Therefore we next derive the trace formula for $\det a(z)$, which gives a weak version of theta condition for the focusing modified KdV equation. The defocusing modified KdV equation can be shown in a similar way.

 \textbf{Simple zeros:}  Assume that the set Eq.\eqref{2.76a} is the simple zeros of $\det a(z)$. Recalling the \textbf{Proposition} $2.4$, one has that $z=z_{n}$ and $z=-k_{0}^{2}/z_{n}^{*}$ are the zeros of $\det a(z)$, which is analytic in $D^{+}$, and that $z=z_{n}^{*}$ and $z=-k_{0}^{2}/z_{n}$ are the zeros of $\det \bar{a}(z)$, which is analytic in $D^{-}$. Then we introduce
\begin{align}
\vartheta^{+}(z)=\det a(z)\prod_{n=1}^{\mathcal {N}}\frac{(z-z_{n}^{*})
(z+q_{0}^{2}/z_{n})}{(z-z_{n})(z+q_{0}^{2}/z_{n}^{*})},\\
\vartheta^{-}(z)=\det\bar{a}(z)\prod_{n=1}^{\mathcal {N}}\frac{(z-z_{n})
(z+q_{0}^{2}/z_{n}^{*})}{(z-z_{n}^{*})(z+q_{0}^{2}/z_{n})},
\end{align}
which are analytic in $D^{+}$ and $D^{-}$ respectively, and have no zeros  point in their respective regions $D^{\pm}$. Note that $\det S(z)=1$ and the asymptotic behavior of $S(z)$ as $z\rightarrow\infty$ we have
\begin{align*}
\vartheta^{+}(z)\vartheta^{-}(z)=\det\left(I_{m}+
\rho^{\dagger}(z)\rho(z)\right)^{-1},\quad z\in\Sigma,
\end{align*}
which can be written as a scalar RH problem
\begin{align*}
\log\vartheta^{+}(z)-\log(1/\vartheta^{-}(z))=-\log
\det\left(I_{m}+\rho^{\dagger}(z)\rho(z)\right),\quad z\in\Sigma.
\end{align*}
Combining \textbf{Lemma} 3.2, we can get the  solution of the RH problem
\begin{align*}
\log \vartheta^{\pm}(z)=\mp\frac{1}{2\pi i}
\int_{\Sigma}\frac{\det\left(I_{m}+\rho^{\dagger}(z)\rho(z)\right)}{\zeta-z}
d\zeta.
\end{align*}
Then the trace formula can be given
\begin{equation*}
\det a(z)=exp\left(-\frac{1}{2\pi i}
\int_{\Sigma}\frac{\det\left(I_{m}+\rho^{\dagger}(z)\rho(z)\right)}{\zeta-z}
d\zeta\right)\prod_{n=1}^{\mathcal {N}}\frac{(z-z_{n}^{*})
(z+q_{0}^{2}/z_{n})}{(z-z_{n})(z+q_{0}^{2}/z_{n}^{*})}.
\end{equation*}
Using the asymptotic behavior of the scattering matrix Eq.\eqref{SJJ} as $z\rightarrow0$ one has that
\begin{align*}
\det a(z)=\frac{1}{k_{0}^{2m}}\det Q_{+}\det Q_{-}^{\dagger},
\end{align*}
and that
\begin{align*}
\det Q_{+}\det Q_{-}^{\dagger}={k_{0}^{2m}}exp\left(-\frac{1}{2\pi i}
\int_{\Sigma}\frac{\det\left(I_{m}+\rho^{\dagger}(z)\rho(z)\right)}{\zeta}
d\zeta\right)\prod_{n=1}^{\mathcal {N}}e^{4i\delta_{n}},
\end{align*}
where the $\delta_{n}$ represents the phase of $z_{n}$, namely $z_{n}=|z_{n}|e^{i\delta_{n}}$.

In addition from the condition Eq.\eqref{T3}, we known that $\det Q_{\pm}=k_{0}^{m}$. When we take the notation
\begin{align*}
\det Q_{+}=k_{0}^{m}e^{i\theta_{+}},\quad
\det Q_{-}=k_{0}^{m}e^{i\theta_{-}},
\end{align*}
then the $\theta$-condition can be derived
\begin{align*}
\det Q_{+}\det Q_{-}^{\dagger}&=k_{0}^{2m}e^{i\theta_{+}}e^{-i\theta_{-}}\\&=
k_{0}^{2m}exp\left(-\frac{1}{2\pi i}
\int_{\Sigma}\frac{\det\left(I_{m}+\rho^{\dagger}(z)\rho(z)\right)}{\zeta}
d\zeta\right)\prod_{n=1}^{\mathcal {N}}e^{4i\delta_{n}},
\end{align*}
which can reduce to
\begin{align*}
\theta_{+}-\theta_{-}=\frac{1}{2\pi}\int_{\Sigma}\det\left(I_{m}+
\rho^{\dagger}(z)\rho(z)\right)\frac{d\zeta}{\zeta}+4\sum_{n=1}^{\mathcal {N}}\delta_{n}
\end{align*}

 \textbf{Double zeros:}   Assume that the discrete spectrum are the double zeros, which means that $\det a(z_{n})=\det a'(z_{n})=0$ and $\det a''(z_{n})\neq0$. Similar to the simple zeros case, introducing two analytic functions  $\vartheta^{\pm}(z)$ which have no zero point in $D^{\pm}$ respectively
\begin{align*}
\vartheta^{+}(z)=\det a(z)\prod_{n=1}^{\mathcal {N}}\frac{(z-z_{n}^{*})^{2}
(z+q_{0}^{2}/z_{n})^{2}}{(z-z_{n})^{2}(z+q_{0}^{2}/z_{n}^{*})^{2}},\\
\vartheta^{-}(z)=\det\bar{a}(z)\prod_{n=1}^{\mathcal {N}}\frac{(z-z_{n})^{2}
(z+q_{0}^{2}/z_{n}^{*})^{2}}{(z-z_{n}^{*})^{2}(z+q_{0}^{2}/z_{n})^{2}},
\end{align*}
we finally get the trace formula and $\theta$-condition
\begin{align*}
\det a(z)=exp\left(-\frac{1}{2\pi i}
\int_{\Sigma}\frac{\det\left(I_{m}+\rho^{\dagger}(z)\rho(z)\right)}{\zeta-z}
d\zeta\right)\prod_{n=1}^{\mathcal {N}}\frac{(z-z_{n}^{*})^{2}
(z+q_{0}^{2}/z_{n})^{2}}{(z-z_{n})^{2}(z+q_{0}^{2}/z_{n}^{*})^{2}}
\end{align*}
and
\begin{align*}
\theta_{+}-\theta_{-}=\frac{1}{2\pi}\int_{\Sigma}\det\left(I_{m}+
\rho^{\dagger}(z)\rho(z)\right)\frac{d\zeta}{\zeta}+8\sum_{n=1}^{\mathcal {N}}\delta_{n}
\end{align*}

\section*{Appendix D. Phase difference of one-soliton solution as $x\rightarrow\pm\infty$}
In this appendix, we will consider the asymptotic behavior of  the one-soliton solution of the equation under the condition $\det C_{1}=0$ or $\det C_{1}\neq0$ as
$x\rightarrow\pm\infty$. Resorting to the Eq.\eqref{T14} $\theta(x,t;z)=\lambda(z)\left(x+(4k^{2}(z)+2k_{0}^{2})t\right)$, and the Eq.\eqref{T10} for $\epsilon=-1$, one has
\begin{align*}
\lambda(z_{n})&=\frac{1}{2}\left(z_{n}+\frac{k_{0}^{2}}{z_{n}}\right)=\frac{1}{2}
\frac{|z_{n}|^{2}z_{n}+z_{n}^{*}k_{0}^{2}}{|z_{n}|^{2}}=\frac{1}{2}
\left(z_{n}+\frac{z_{n}^{*}k_{0}^{2}}{|z_{n}|^{2}}\right)\\&\Rightarrow
Im\lambda(z_{n})=\frac{1}{2}\left(Imz_{n}\right)\left(1-\frac{k_{0}^{2}}
{|z_{n}|^{2}}\right)=-Im\lambda(z_{n}^{*}).
\end{align*}
where $z_{n}\in\Sigma$ and we know that $e^{2i\theta(z_{n})}$ and $e^{-2i\theta(z_{n}^{*})}$ grow exponentially as $x\rightarrow-\infty$, and decay
exponentially as $x\rightarrow+\infty$. In addition, as $x\rightarrow+\infty$ using the Eqs.(4.11), (4.12) and (4.13) yield that $\varpiup_{1}\rightarrow1$, and that $\Upsilon_{1}, \Upsilon_{2}, \Gamma_{1}, \Gamma_{2}, X_{1}$ and $X_{2}\rightarrow I_{2}$, so $Q(x,t)\rightarrow Q_{+}$. In what follows, we discuss the asymptotic behavior as $x\rightarrow-\infty$. Note that there are two cases at this time, i.e., $\det C_{1}=0$, and $\det C_{1}\neq0$.\\
For the case $\det C_{1}=0$:
\begin{align*}
\left\{\begin{aligned}
&\varpiup_{1}=\frac{ik_{0}}{(\zeta_{1}^{*})^{2}+k_{0}^{2}}
e^{-2i\theta{(\zeta_{1}^{*})}}trace(C_{1}^{\dagger}),\\
&\Upsilon_{1}=I_{2}+\frac{\zeta_{1}^{2}+k_{0}^{2}}{\zeta_{1}
(\zeta_{1}^{*}-\zeta_{1})trace(C_{1})}C_{1},\\
&\Upsilon_{2}=I_{2}+\frac{\zeta_{1}((\zeta_{1}^{*})^{2}+k_{0}^{2})^{2}}
{k_{0}^{2}(\zeta_{1}^{*}-\zeta_{1})trace(C_{1}^{\dagger})},\\
&\Gamma_{1}=\frac{ik_{0}}{(\zeta_{1}^{*})^{2}+k_{0}^{2}} C_{1}^{\dagger}\left(
I_{2}-\frac{|\zeta_{1}^{2}+k_{0}^{2}|^{2}}{k_{0}^{2}(\zeta_{1}^{*}-\zeta_{1})
trace(C_{1})}C_{1}\right)e^{-2i\theta(z_{n}^{*})},\\
&\Gamma_{2}=-\frac{ik_{0}}{\zeta_{1}^{2}+k_{0}^{2}}C_{1}\left(
I_{2}-\frac{|\zeta_{1}^{2}+k_{0}^{2}|^{2}}{k_{0}^{2}(\zeta_{1}^{*}-\zeta_{1})
trace(C_{1}^{\dagger})}C_{1}^{\dagger}\right)e^{2i\theta(z_{n})}.
\end{aligned}\right.
\end{align*}
For the case $\det C_{1}\neq0$:
\begin{align*}
\left\{\begin{aligned}
&\varpiup_{1}=-\frac{k_{0}^{2}}{(\zeta_{1}^{*})^{2}+k_{0}^{2}}
e^{-4i\theta(z_{n}^{*})}\det(C_{1}^{\dagger}),\\
&\Upsilon_{1}=\frac{|\zeta_{1}|^{2}+k_{0}^{2}}{\zeta_{1}
(\zeta_{1}^{*}-\zeta_{1})}I_{2},\\
&\Upsilon_{2}=\frac{\zeta_{1}^{*}(|\zeta_{1}|^{2}+k_{0}^{2})}
{k_{0}^{2}(\zeta_{1}^{*}-\zeta_{1})}I_{2},\\
&\Gamma_{1}=-\frac{i(|\zeta_{1}|^{2}+k_{0}^{2})^{2}}{k_{0}(\zeta_{1}^{*}-\zeta_{1})^{2}
((\zeta_{1}^{*})^{2}+k_{0}^{2})}C_{1}^{\dagger}e^{-2i\theta{(\zeta_{1}^{*})}},\\
&\Gamma_{2}=\frac{i(|\zeta_{1}|^{2}+k_{0}^{2})^{2}}{k_{0}(\zeta_{1}^{*}-\zeta_{1})^{2}
(\zeta_{1}^{2}+k_{0}^{2})}C_{1}e^{2i\theta{(\zeta_{1})}}.
\end{aligned}\right.
\end{align*}
For the case $\det C_{1}\neq0$, according to the asymptotic property obtained above, we can deduce the asymptotic property of $X_{1}=\Upsilon_{1}\Gamma_{1}^{-1}$ and $X_{2}=\Upsilon_{2}\Gamma_{2}^{-1}$ and judge that they are exponentially decaying so as to further analyze the phase difference between potential function $Q(x,t)$ and boundary value condition $Q_{-}$ as $x\rightarrow-\infty$. From the Eq.\eqref{T54}
\begin{align*}
Q(x,t)&=k_{0}I_{2}-i\Upsilon_{1}\Gamma_{1}^{-1}e^{-2i\theta
(x,t;\zeta_{1}^{*})}C_{1}^{\dagger}+
\frac{ik_{0}^{2}}{\zeta_{1}^{2}}\Upsilon_{2}\Gamma_{2}^{-1}
C_{1}e^{2i\theta(x,t;\zeta_{1})}\\
&=k_{0}I_{2}+I_{2}\frac{k_{0}(\zeta_{1}^{*}-\zeta_{1})
\left[(\zeta_{1}^{*})^{2}+k_{0}^{2}+\zeta_{1}^{-1}\zeta_{1}^{*}
(\zeta_{1}^{2}+k_{0}^{2})\right]}{k_{0}(|\zeta_{1}|^{2}+k_{0}^{2})}\\
&=e^{-4i\tau}k_{0}I_{2},
\end{align*}
which in turn implies
\begin{align*}
Q(x,t)\sim Q_{-}=e^{-4i\tau}k_{0}I_{2},
\end{align*}
where $\tau$ denotes the phase of the $\zeta_{1}$.

For the case $\det C_{1}=0$, from the Eqs.\eqref{T57} and \eqref{T58} we have
\begin{align*}
&\Gamma_{1}^{-1}=\frac{1}{\Delta_{1}}\left[I_{2}+e^{-2i\theta{(\zeta_{1}^{*})}}\text{cof}
\left(\frac{ik_{0}}{(\zeta_{1}^{*})^{2}+k_{0}^{2}}I_{2}
-\frac{1}{(\zeta_{1}^{*}-\zeta_{1})^{2}}\frac{e^{2i\theta{(\zeta_{1})}}}
{\varpiup_{1}^{*}}C_{1}\right)\text{cof}(C_{1}^{\dagger})\right],\\
&\Gamma_{2}^{-1}=\frac{1}{\Delta_{1}^{*}}\left[I_{2}+e^{2i\theta{(\zeta_{1})}}
\text{cof}\left(-\frac{ik_{0}}{\zeta_{1}^{2}+k_{0}^{2}}I_{2}
-\frac{1}{(\zeta_{1}^{*}-\zeta_{1})^{2}}\frac{e^{-2i\theta{(\zeta_{1}^{*})}}}
{\varpiup_{1}}C_{1}^{\dagger}\right)\text{cof}(C_{1})\right],\\
&\Delta_{1}=1+e^{-2i\theta{(\zeta_{1}^{*})}}\text{trace}\left(
\frac{ik_{0}}{(\zeta_{1}^{*})^{2}+k_{0}^{2}}C_{1}^{\dagger}
-\frac{1}{(\zeta_{1}^{*}-\zeta_{1})^{2}}\frac{e^{2i\theta{(\zeta_{1})}}}
{\varpiup_{1}^{*}}C_{1}^{\dagger}C_{1}\right).
\end{align*}
Note that $\det C_{1}=0$ yields that $\text{cof}(C_{1})C_{1}=\text{cof}(C_{1}^{\dagger})C_{1}^{\dagger}=0$. Therefore as $x\rightarrow-\infty$, one has
\begin{align*}
&\varpi_{1}\thicksim\frac{ik_{0}}{(\zeta_{1}^{*})^{2}+k_{0}^{2}}
e^{-2i\theta(\zeta_{1}^{*})}\text{trace}(C_{1}^{\dagger}),\\
&\varpi_{1}^{*}\thicksim-\frac{ik_{0}}{\zeta_{1}^{2}+k_{0}^{2}}
e^{2i\theta(\zeta_{1})}\text{trace}(C_{1}),\\
\end{align*}
and then
\begin{align*}
\Delta_{1}\thicksim e^{-2i\theta{(\zeta_{1}^{*})}}\text{trace}\left[
\frac{ik_{0}}{(\zeta_{1}^{*})^{2}+k_{0}^{2}}C_{1}^{\dagger}+\frac{1}{(\zeta_{1}^{*}-
\zeta_{1})^{2}}\frac{\zeta_{1}^{2}+k_{0}^{2}}{ik_{0}\text{trace}
(C_{1})}C_{1}^{\dagger}C_{1}\right]\doteq e^{-2i\theta{(\zeta_{1}^{*})}}\Delta_{1}^{-}.
\end{align*}
As a consequence, the asymptotic behaviour of the one-soliton solution can be derived by as $x\rightarrow-\infty$ for $\det C_{1}=0$, i.e.,
\begin{align*}
&Q(x,t)=k_{0}I_{2}-i\Upsilon_{1}\Gamma_{1}^{-1}C_{1}^{\dagger}
e^{-2i\theta(\zeta_{1}^{*})}+
\frac{ik_{0}^{2}}{\zeta_{1}^{2}}\Upsilon_{2}\Gamma_{2}^{-1}
C_{1}e^{2i\theta(\zeta_{1})}\\ &\thicksim Q_{+}-\frac{1}{\Delta_{1}}\left[
I_{2}+\frac{\zeta_{1}^{2}+k_{0}^{2}C_{1}}{\zeta_{1}(\zeta_{1}^{*}-\zeta_{1})}\right]
C_{1}^{\dagger}e^{-2i\theta(\zeta_{1}^{*})}+\frac{1}{\Delta_{1}^{*}}\left[
I_{2}+\frac{\zeta_{1}((\zeta_{1}^{*})^{2}+k_{0}^{2})}{k_{0}^{2}(\zeta_{1}^{*}
-\zeta_{1})\text{trace}(C_{1}^{\dagger})}C_{1}^{\dagger}\right]
\frac{k_{0}^{2}C_{1}}{\zeta_{1}^{2}}e^{2i\theta(\zeta_{1})}\\
&\thicksim Q_{+}-\frac{iC_{1}^{\dagger}}{\Delta_{1}^{-}}+\frac{i(\zeta_{1}^{*})^{2}
+k_{0}^{2}}{\zeta_{1}(\zeta_{1}^{*}-\zeta_{1})\text{trace}(C_{1}^{\dagger})
(\Delta_{1}^{-})^{*}}(C_{1}C_{1}^{\dagger}+C_{1}^{\dagger}C_{1})+\frac{ik_{0}^{2}
C_{1}}{\zeta_{1}^{2}(\Delta_{1}^{-})^{*}}.
\end{align*}

\end{document}